\documentclass{article}
\usepackage[utf8]{inputenc}
\usepackage{fullpage}
\usepackage{amsmath}
\usepackage{amssymb}
\usepackage{amsthm}
\usepackage{dsfont}
\usepackage{physics}
\usepackage{tikz}
\usepackage{array}
\usepackage{comment}
\usepackage{float}
\usepackage{authblk}

\usepackage[colorlinks=true,linkcolor=blue,allcolors=blue]{hyperref}
\usepackage[capitalize,nameinlink]{cleveref}
\usepackage{todonotes}
\usepackage{lmodern}
\usetikzlibrary{decorations.pathmorphing}

\tikzset{snake it/.style={decorate, decoration=snake}}

\DeclareMathOperator{\BQP}{\mathsf{BQP}}

\DeclareMathOperator{\NP}{\mathsf{NP}}

\DeclareMathOperator{\SharpP}{\mathsf{\#P}}

\DeclareMathOperator{\NEXP}{\mathsf{NEXP}}

\DeclareMathOperator{\QMA}{\mathsf{QMA}}

\DeclareMathOperator{\QIP}{\mathsf{QIP}}

\DeclareMathOperator{\PostBQP}{\mathsf{PostBQP}}

\DeclareMathOperator{\LH}{\mathsf{LH}}

\newcommand{\poly}{\operatorname{poly}}
\newcommand{\polylog}{\operatorname{polylog}}

\newcommand{\loglog}{\operatorname{loglog}}

\newcommand{\id}{\mathrm{I}}

\newcommand{\supp}{\operatorname{supp}}
\newcommand{\loc}{\operatorname{loc}}
\newcommand{\err}{\operatorname{err}}
\newcommand{\FT}{\mathrm{FT}}

\newtheorem{theorem}{Theorem}[section]
\newtheorem{definition}[theorem]{Definition}

\newtheorem{remark}[theorem]{Remark}
\newtheorem{lemma}[theorem]{Lemma}

\newtheorem{fact}[theorem]{Fact}
\newtheorem{claim}[theorem]{Claim}

\newtheorem{corollary}[theorem]{Corollary}
\newtheorem{task}[theorem]{Task}

\newcommand{\anu}[1]{\todo[inline, color=red!30]{Anurag: #1}}

\title{Circuit-to-Hamiltonian from tensor networks and fault tolerance}
\author{}
\author{Anurag Anshu\thanks{School of Engineering and Applied Sciences, Harvard University, USA} \qquad Nikolas P. Breuckmann\thanks{School of Mathematics, University of Bristol, UK} \qquad Quynh T. Nguyen\thanks{School of Engineering and Applied Sciences and Harvard Quantum Initiative, Harvard University, USA}}

\date{}

\begin{document}

\maketitle

\begin{abstract}
    We define a map from an arbitrary quantum circuit to a local Hamiltonian whose ground state encodes the quantum computation.
    All previous  maps relied on the Feynman-Kitaev construction, which introduces an ancillary `clock register' to track the computational steps. 
    Our construction, on the other hand, relies on injective tensor networks with associated parent Hamiltonians, avoiding the introduction of a clock register. 
    This comes at the cost of the ground state containing only a noisy version of the quantum computation, with independent stochastic noise. 
    We can remedy this - making our construction robust - by using quantum fault tolerance.
    In addition to the stochastic noise, we show that any state with energy density exponentially small in the circuit depth encodes a noisy version of the quantum computation with adversarial noise. 
    We also show that any `combinatorial state' with energy density polynomially small in depth encodes the quantum computation with adversarial noise. 
    This serves as evidence that any state with energy density polynomially small in depth has a similar property.

    As applications, we give a new proof of the $\QMA$-completeness of the local Hamiltonian problem (with logarithmic locality) and show that contracting injective tensor networks to additive error is $\BQP$-hard. We also discuss the implication of our construction to the quantum PCP conjecture, combining with an observation that $\QMA$ verification can be done in logarithmic depth.

\end{abstract}

\section{Introduction}
The Feynman-Kitaev `clock based' mapping~\cite{kitaev2002classical} from quantum circuits to local Hamiltonians is the central tool bridging quantum complexity theory and quantum many-body physics. The mapping and its variants have been used to justify the hardness of computing the ground energy of natural local Hamiltonians \cite{3LocHamCompl,aharonov2009power,kempe2006complexity,locham2d,nagaj,breuckmann2014space}. It has been used to construct explicit local Hamiltonians with `complex' ground states - in terms of large entanglement entropy \cite{Irani10, AHLNSV14} or circuit depth \cite{NVY18}. Other important applications include the equivalence of adiabatic and circuit models \cite{aharonov2008adiabatic}, delegation of quantum computing \cite{gheorghiu2019verification} etc. However, a well known limitation of the Feynman-Kitaev mapping is the soundness. While quantum computations that output `accept' with probability (near) $1$ get mapped to (near) frustration-free local Hamiltonians, the quantum computations that output `reject' with high probability get mapped to local Hamiltonians with ground energy density $1/\poly(\text{number of gates})$. This serves as the main bottleneck to the quantum PCP conjecture \cite{aharonov2002quantum, aharonov2013guest}, which seeks a constant energy density in the rejecting case.

An alternative mapping of quantum computation to many-body systems was laid out in \cite{schuch2007computational} by using measurement-based quantum computing (MBQC). It was shown that running MBQC and post-selecting on `no correction' led to a tensor network which encoded the result of the quantum computation. However, this technique does not yield a desired circuit-to-Hamiltonian mapping due to two issues. First, the encoding tensor network may not be the ground state of any local Hamiltonian. Second, the tensor networks also capture quantum computation with post-selection, which leads to a class much larger than $\QMA$.

Our starting point is the observation that both the issues no longer exist if we consider the class of injective tensor networks: injectivity prevents us from post-selecting on events of very small probability and injective tensor networks also have a natural parent Hamiltonian. The price we pay is that the injective tensor network represents a noisy version of the quantum circuit. This is handled by replacing the circuit by its a fault-tolerant version. 

The details of the construction appear in \cref{sec:model}, where we use standard teleportation instead of measurement-based quantum computing. A high level overview is as follows, using a simple circuit $U_2U_1\ket{0}$ involving 1 qubit gates on $\ket{0}$. Introduce 5 qubits in the state $\ket{0}\otimes (\id\otimes U_1)\ket{\Phi_{\id}}\otimes (\id\otimes U_2)\ket{\Phi_{\id}}$, where $\ket{\Phi_{\id}}=\frac{1}{\sqrt{2}}(\ket{00}+\ket{11})$. Projecting qubits 1, 2 and 3, 4 with $\ketbra{\Phi_{\id}}$ would lead to the desired state $U_2U_1\ket{0}$ on qubit 5. This procedure defines a tensor network state known as PEPS~\cite{verstraete2008matrix}. However, this tensor network state is not the ground state of a local Hamiltonian, and the reason is that the map $\ketbra{\Phi_{\id}}$ is not injective. Instead, we replace $\ketbra{\Phi_{\id}}$ with the invertible map $\ketbra{\Phi_{\id}}+\delta(\id-\ketbra{\Phi_{\id}})$ for $\delta >0$. It can be verified that the last qubit is now a noisy version of the original circuit (with depolarizing noise of strength $O(\delta^2)$) and qubits 1, 2 and 3, 4 record the Pauli errors.
Importantly, this injective PEPS state admits a local and frustration-free `parent Hamiltonian' as desired. The described scheme applies to general quantum circuits, giving a mapping from any quantum circuit to a local Hamiltonian whose ground state represents the noisy quantum computation with stochastic iid noise with strength $O(\delta^2)$ per wire (see \Cref{sec:injectiveTN}.)

Our main technical contribution is a characterization of low-energy states of the parent Hamiltonian as adversarial computations of the quantum circuit. In particular, we consider an adversarial noise model where, in each layer of the circuit, a certain fraction of qubits are deviated arbitrarily by the adversary. Consider a quantum circuit $W$ of depth $D$ (that may be, for examples, a $\QMA$ verification circuit or a $\BQP$ circuit) and let its parent Hamiltonian be $H_\mathrm{parent}=\sum_{i=1}^m h_i$ (assume here $ 0 \leq h_i \leq 1$ for simplicity). We exhibit the following properties for low-energy states of $H_\mathrm{parent}$:
\begin{itemize}
    \item For a circuit of depth $D=O(\log |W|)$, any state $\ket{\psi}$ with energy density $e^{-\Omega(D \log D)}$, i.e.,
    \begin{align*}
        \frac{1}{m}\bra{\psi} H_\mathrm{parent} \ket{\psi} \leq e^{-\Omega(D \log D)},
    \end{align*}
    can be viewed as the output of the circuit with $O(\delta^2)$ fraction of adversarial noise per layer. See \Cref{subsec:exposound} for the precise statement and proof.
     \item Any combinatorial state with energy density (equal to the fraction of violated constraints) $\frac{1}{\poly(D)}$, i.e.,
     \begin{align*}
         \frac{1}{m}|\{i: \bra{\psi}h_{i}\ket{\psi}\neq 0\}| \leq \frac{1}{\poly(D)},
     \end{align*}
     can be viewed as the output of the circuit with $O(\delta^2)$ fraction of adversarial noise per layer. See \Cref{subsec:combsound} for the precise statement and proof.
\end{itemize}

Informally, such adversarial noise model arises naturally since each Hamiltonian term in $H_\mathrm{parent}$ is enforcing the application of a gate, and thus we expect the energy violation of a Hamiltonian term to be reflected as a fault at the corresponding location in the quantum circuit (in addition to the stochastic local noise of strength $O(\delta^2)$ already existing in the ground state). The above characterizations are simplified statements that neglect the stochastic noise by simply choosing $\delta=\frac{1}{\poly(D)}$ -- this keeps the fraction of adversarial noise per layer to be $\frac{1}{\poly(D)}$, which is an error budget low enough that the adversary cannot stop the whole computation. The main open question is that whether any states with energy density $1/\poly(D)$ can also be viewed as the output of adversarial noisy version. The second result above (on combinatorial states) is evidence in its favor. 

\begin{table}
\centering
\begin{tabular}{|m{1.1cm} | m{7cm}| m{7cm} | } 
  \hline
  &\textbf{Feynman-Kitaev construction \cite{kitaev2002classical}} & \textbf{Present construction} \\ 
  \hline
  Ground state& Superposition over partial computations of $W$ & Tensor network encoding a noisy version of $W$ with i.i.d noise per wire \\ 
  \hline
   Low-energy states& States with energy density $\frac{O(1)}{|W|^3}$ encode $W$ & Combinatorial states with $\frac{O(1)}{D}$ fraction violations encode a noisy version of~$W$ with adversarial noise (\Cref{thm:combin}). 
   \begin{itemize}
       \item States with energy density $e^{-\Omega(D\log D)}$ (for $D=O(\log |W|)$) encode a noisy version of $W$ with adversarial noise (\Cref{thm:energy}). 
 \end{itemize}
  \\
       \hline
     Limi-tation & There exists a combinatorial state with $\frac{O(1)}{|W|}$ fraction of violations containing no information about~$W$ (see~\Cref{remark:FKfail}). & There exists a combinatorial state with $\frac{O(1)}{D}$ fraction of
      violations contain no information about~$W$. 
   \\ 
  \hline
\end{tabular}
\caption{A comparison between the Feynman-Kitaev mapping and our construction for quantum circuit $W$ of depth $D$. Our main open question is that any state with energy density $\frac{1}{\poly(D)}$ encode noisy version of $W$ with adversarial noise. Since we can choose $D=O(\log|W|)$ in $\QMA$ protocols (\Cref{sec:shallow_circuits} and~\Cref{app:logdepth}), this serves as a link between polylog quantum PCP and adversarial quantum fault tolerance.}
\label{table:clock}
\end{table}

\vspace{0.1in}

\noindent {\bf Connection to quantum PCP conjecture:} The quantum PCP conjecture \cite{aharonov2002quantum, aharonov2013guest} states that it is $\QMA$-hard to decide if the ground energy density of a local Hamiltonian problem is less than a given number $a$ or more than $a+\Delta$ for a constant $\Delta$. A  `polylog weaker' version of this conjecture - $\QMA$ hardness of deciding that ground energy density is $\leq a$ or $>a+\frac{1}{\mathrm{polylog} n}$ - is also open (even when the locality is relaxed to be $\polylog n$). In the equivalent formulation in terms of probabilistic proof checking~\cite{aharonov2013guest}, this polylog-weaker quantum PCP conjecture is expressed as the (presumed) inclusion $\QMA \overset{?}{\subseteq} \mathsf{QPCP}[\polylog]$. See Appendix \ref{append:QMApoly} for a discussion on known soundness results.

Our attempt in this work is to link adversarial quantum fault tolerance with the above `polylog weaker' quantum PCP. At a high level, we expect such a connection due to the correspondence between Hamiltonians and quantum circuits \cite{kitaev2002classical} and the view that quantum PCP conjecture is about adversarial violations of local Hamiltonian terms. An issue with this is that quantum PCP conjecture expects soundness against constant fraction of violations, but in a depth-$D$ quantum circuit we would be able to deal with adversarial faults on at most a fraction 
$\frac{1}{D}$ of the gates (the threshold above which the adversary could irrecoverably corrupt an entire layer of gates). However, as shown in \Cref{sec:shallow_circuits} (and also \Cref{app:logdepth}), $\QMA$ verification can be achieved in logarithmic depth ($D=O(\log n)$; $n$ is the number of qubits in the $\QMA$ verifier circuit). Thus, if we seek the polylog quantum PCP, the connection with adversarial quantum fault tolerance can be more transparent.

Our result takes a step towards this connection by showing that combinatorial states with $\frac{1}{\text{poly}(D)}$ fraction violations encode a circuit with adversarial errors\footnote{This feature indeed holds in the classical Cook-Levin mapping, but fails in the Feynman-Kitaev mapping (see~\Cref{remark:FKfail}).}. Suppose $\frac{1}{\text{poly}(D)}$-energy density states in our construction also encode a circuit with adversarial error, which is our main open question. 
And suppose any $O(\log(n))$-depth $\QMA$ verifier can be transformed into a $\polylog(n)$-depth $\QMA$ verifier that is sound against $\frac{1}{\polylog(n)}$ fraction of adversarial errors in the circuit.\footnote{Note that we also need soundness against a superposition over adversarial errors - see~\Cref{sec:soundness}.}
Then the polylog weaker version of quantum PCP holds.  

Classical analogue of this line of argument was first presented in \cite{gal1995fault} using the classical Cook-Levin mapping. We discuss a version of this in \Cref{append:classicalPCP}. One may also wonder why we expect the route via quantum fault tolerance to be useful in quantum PCP conjecture. Quantum fault tolerance is well known to handle the no-cloning issue and successfully quantize classical fault tolerance. Given that no-cloning is the main barrier in quantizing Dinur's proof \cite{Dinur07} of the PCP theorem, one expects quantum fault tolerance to provide a way around the barrier.

\vspace{0.1in}
\noindent {\bf New proof of $\QMA$-completeness of local Hamiltonian:} The first application of quantum circuit-to-Hamiltonian mapping was Kitaev's proof of the $\QMA$-completeness of the local Hamiltonian problem.
In \Cref{sec:qma}, we apply our construction to give a new proof of this seminal result for the case of logarithmic-local Hamiltonians. In particular, we prove that determining if the ground energy density of a $O(\log n)$-local Hamiltonian family is less than a given number $a$ or more than $a+ \frac{1}{\poly(n)}$ is $\QMA$-complete. While we have not been able to prove the same statement for the constant locality case (see discussion below), we remark that our proof is completely independent of the Feynman-Kitaev clock construction. 

Our starting point is the observation that a constant-depth-overhead quantum fault tolerance scheme (for $\QMA$) against iid stochastic noise, combined with our construction, would yield the desired new proof of $\QMA$-completeness for $O(1)$-local Hamiltonians. Constant-depth-overhead fault tolerance is a very interesting question on its own since it is possible in classical computation~\cite{pippenger1985networks}, whereas its possibility in the quantum case is unknown. We make progress towards this question by constructing a fault tolerance scheme using a recent linear-distance quantum LDPC code family due to Leverrier and Zémor~\cite{leverrier2023decoding}. However, the constructed fault-tolerant circuit still requires a $O(\loglog n)$-depth overhead of noiseless classical operations. We get around this limitation by using a few layers of $O(\log n)$-local gates (hence increasing the parent Hamiltonian's locality to $O(\log n)$) and exploiting the structure of the $\QMA$ circuit. The proof relies significantly on the linear distance and a local parallel decoder of the Leverrier-Zémor code~\cite{leverrier2023decoding}.

\vspace{0.1in}

\noindent {\bf Complexity of injective tensor networks:}
Injective tensor networks constitute a more physical family of quantum states and have been shown to be efficiently preparable on a quantum computer~\cite{schwarz2012preparing, ge2016rapid} and contractable in classical quasi-polynomial time~\cite{schwarz2017approximating} under assumptions on the parent Hamiltonian spectral gap. However, the lack of the postselection ability makes it less clear how to characterize injective TN from a complexity-theoretic point of view. 

Combining our construction with existing quantum fault-tolerance schemes for local stochastic noise~\cite{aharonov1999fault}, we conclude that preparing injective TN states on a quantum computer is $\BQP$-hard. This can be seen as a complement to prior works~\cite{schwarz2012preparing,ge2016rapid}, that showed preparing injective TN states under spectral gap assumptions is in $\BQP$. Compared with the $\PostBQP$-hardness shown in~\cite{schuch2007computational}, the $\BQP$-hardness naturally reflects the non-postselecting nature of injective TN. Regarding the classical complexity of injective TN, our construction also implies that evaluating local observable expectation values on injective-TN states is $\BQP$-hard to $O(1)$-additive error. In addition, we show the same task for a non-local observable is $\SharpP$-hard to $O(1)$-multiplicative error.

\paragraph{Organization:} The remaining of the paper is organized as follows. We start by introducing injective tensor networks and our circuit-to-Hamiltonian mapping in~\Cref{sec:model}. In~\Cref{sec:background}, we introduce some preliminaries about Hamiltonian complexity and quantum fault tolerance. In~\Cref{sec:soundness} we prove our main structural results for the parent Hamiltonian, including characterizations of low-energy states and a spectral gap lower bound. We prove that $\QMA$ verification circuits can be assumed to be logarithmic depth in~\Cref{sec:shallow_circuits} (another proof is given in~\Cref{app:logdepth}). We give a new proof of the $\QMA$-completeness of log-local Hamiltonians in~\Cref{sec:qma}. We prove new computational hardness results for injective tensor networks in~\Cref{sec:bqp-hardness}. Finally, we discuss some open questions in~\Cref{sec:open}.

\section{The Model}\label{sec:model}
Let us first outline the general idea behind the construction:
given a quantum circuit $W$, we consider a tensor network associated with the implementation of $W$ (\Cref{sec:circ2TN}).
We make the tensor network injective by perturbing each of its projectors $P_i$ by some small amount $\delta$ (\Cref{sec:injectiveTN}), so that we can associate it with a parent Hamiltonian (\Cref{subsec:parentHam}).
However, these local perturbations are unwanted.
Crucially, we observe that they can be interpreted as Pauli-errors occurring during the execution of $W$.
Hence, we have to consider a fault-tolerant version of~$W$, which requires us to implement a quantum error correction protocol within the model itself. In~\Cref{subsec:bypass} we give an explanation why our construction can overcome the conceptual challenge of local indistinguishability in quantum circuit-to-Hamiltonian. We discuss related prior works in~\Cref{subsec:priorworks}.

\subsection{Notations}
Let the EPR states be $\ket{\Phi_I}= \frac{1}{\sqrt{2}}(\ket{00}+\ket{11})$, $\ket{\Phi_X}= (\id \otimes X) \ket{\Phi_I}, \ket{\Phi_{XZ}}= (\id \otimes XZ) \ket{\Phi_I}, \ket{\Phi_Z}= (\id \otimes Z) \ket{\Phi_I}$. Denote $\mathcal{P}=\{I,X,XZ,Z\}$.
For an operator $A$ in a Hilbert space with tensor product structure $\mathcal{H}=(\mathbb{C}^d)^{\otimes n}$, we denote by $\supp(A)$ the span of eigenvectors of $A$ with nonzero eigenvalues and by $\loc(A)$ the set of subsystems on which $A$ acts nontrivially.

\subsection{Quantum circuit to tensor network}\label{sec:circ2TN}
We will now discuss the definition of the tensor network $T$ associated to the circuit $W$. Let $n$ be the total number of qubits on which $W$ operates and $D$ its depth.
For simplicity and without loss of generality, let us assume that $W$ consists of 2-qubit gates arranged in a brickwork layout, see \Cref{fig:brickCirc}. 
The generalization to arbitrary circuits is straightforward.
\begin{figure}[h]
    \centering
    \begin{tikzpicture}[scale=0.4]
  \clip (-0.25,-0.5) rectangle (19.25,11.5);
  \foreach \x in {0,8,16,24} {
    \foreach \y in {0,4,8} {
      \draw[ultra thick] (\x+1,\y) rectangle (\x+2,\y+3);
      \draw[thick] (\x-2,\y+0.5) -- (\x+1,\y+0.5); 
      \draw[thick] (\x-2,\y+2.5) -- (\x+1,\y+2.5); 
    }
  }
  \foreach \x in {4,12,20} {
    \foreach \y in {-2,2,6,10} {
      \draw[ultra thick] (\x+1,\y) rectangle (\x+2,\y+3);
      \draw[thick] (\x-2,\y+0.5) -- (\x+1,\y+0.5); 
      \draw[thick] (\x-2,\y+2.5) -- (\x+1,\y+2.5); 
    }
  }
\end{tikzpicture}
    \caption{
        The circuit $W$ consisting of a collection of gates (black boxes).
        This layout suffices to implement an arbitrary quantum circuit. However, our construction applies to general circuit layouts.
    }
    \label{fig:brickCirc}
\end{figure}
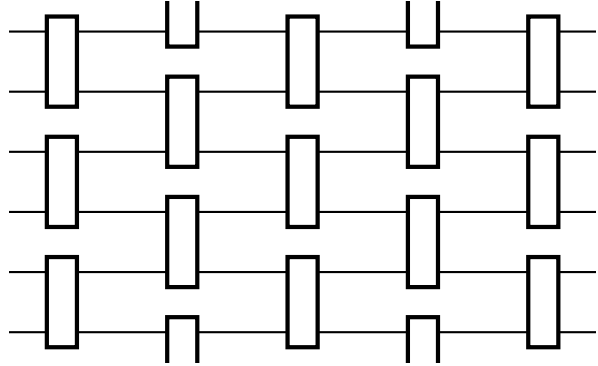
Consider a 2-qubit gate $U^{(\ell)}_{p,q}$ acting one qubits $p$ and $q$ at layer $\ell$, we assign a 4-qubit Choi state $\ket{\Phi_{U}}$ encoding of the gate as follows
\begin{equation}
    \ket{\Phi^{(\ell)}_{p,q}} = [I_{1,2}\otimes (U^{(\ell)}_{p,q})_{3,4}] \left( \ket{00}_{1,3} + \ket{11}_{1,3}  \right) \otimes \left( \ket{00}_{2,4} + \ket{11}_{2,4}  \right) /2.
\end{equation}
See \Cref{fig:state} for a diagrammatic representation of this state.

\begin{figure}[h]
    \centering
    \begin{tikzpicture}[scale=0.6]
      \path [draw=blue,snake it, thick] (-0.5,2.5) -- (1.5,2.5); 
      \path [draw=blue,snake it, thick] (-0.5,0.5) -- (1.5,0.5); 
      \draw[ultra thick] (0.9,0) rectangle (2.1,3) node[right] {$U^{(\ell)}_{p,q}$};
      \draw [fill=black] (-0.5,0.5) circle [radius=0.1] node[left] {$2$}; 
      \draw [fill=black] (-0.5,2.5) circle [radius=0.1] node[left] {$1$}; 
      \draw [fill=black] (1.5,0.5) circle [radius=0.1] node[right] {$4$}; 
      \draw [fill=black] (1.5,2.5) circle [radius=0.1] node[right] {$3$}; 
    \end{tikzpicture} 
    \caption{
        Representation of the state $\ket{\Phi^{(\ell)}_{p,q}}$.
        Qubits 1 and 3, as well as 2 and 4, are in the Bell state $\ket{\Phi_I}$, which is indicated by the blue wavy lines.
        The unitary $U^{(\ell)}_{p,q}$ is applied to qubits 3 and 4 (black box).
    }
    \label{fig:state}
\end{figure}
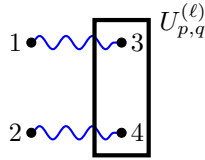

Our starting point is that the state $\ket{\Phi^{(\ell)}_{p,q}}$ allows implementing the gate $U^{(\ell)}_{p,q}$ via teleportation and postselection. For example, the application of gate $U^{(1)}_{p,q}$ on an input state $\ket{\xi}$ is simulated by projecting the joint system $\ket{\xi_{p,q}}\otimes \ket{\Phi^{(1)}_{p,q}}_{1,2}$ onto the EPR state $\ket{\Phi_I} =\left( \ket{00} + \ket{11} \right) \otimes \left( \ket{00} + \ket{11}  \right) /2$. More generally, a gate $U^{(t)}_{p,q}$ can be effected by applying the projector
\begin{equation}\label{eqn:projector}
    P= \ket{\Phi_I}\bra{\Phi_I}
\end{equation}
onto qubits $p,q$ of the input state to the gate and qubits 1, 2 of $\ket{\Phi^{(\ell)}_{p,q}}$, see \Cref{fig:brickTN}.

For brevity we will often denote $\ket{\Phi_U} = (\id \otimes U)\ket{\Phi_I}^{\otimes 2}$, where $U$ is a two-qubit gate acting on the second qubits of $\ket{\Phi_I}$, leaving the location in spacetime of $U$ implicit.

Extending the previous idea to the entire circuit, we can encode any $n$-qubit quantum circuit $W$ into a tensor network. In particular, we have $n$ qubits in the first column of the tensor network storing the input state, and the other columns storing the EPR encoding of the gates. The total number of qubits of the tensor network is $(2D+1)n$.
Define the $(2D+1)n$-qubit product state
\begin{equation}
    \ket{\Phi_{W,\xi}}= \ket{\xi} \otimes \bigotimes_{\ell, p,q} \ket{\Phi^{(\ell)}_{p,q}},
\end{equation}
where $\ket{\xi}$ is the $n$-qubit input state of the circuit $W$. For example, let $\ket{\xi}=\ket{0}^{\otimes n}$. 
Then applying the EPR projector $\Pi_W \triangleq \bigotimes_{\ell,p,q} P^{(\ell)}_{p,q} $ on $\ket{\Phi_{W,\xi}}$ results in the output state in the last column (up to normalization)
\begin{equation}
    \Pi_W \ket{\Phi_{W,\xi}} \propto \ket{\Phi_I} ^{\otimes nD} \otimes \left( W\ket{0}^{\otimes n} \right) .
\end{equation}
Tensor networks of this form are in general termed projected entangled pair states (PEPS).

\begin{figure}
    \centering
    \begin{tikzpicture}[scale=0.4]
  \clip (-0.25,-0.5) rectangle (20.73,11.5);
  \foreach \x in {0,8,16,24} {
    \foreach \y in {0,4,8} {
      \path [draw=blue,snake it, thick] (\x-0.5,\y+2.5) -- (\x+1.5,\y+2.5); 
      \path [draw=blue,snake it, thick] (\x-0.5,\y+0.5) -- (\x+1.5,\y+0.5); 
      \draw[ultra thick] (\x+1,\y) rectangle (\x+2,\y+3);
      \draw [fill=black] (\x-0.5,\y+0.5) circle [radius=0.1]; 
      \draw [fill=black] (\x-0.5,\y+2.5) circle [radius=0.1]; 
      \draw [fill=black] (\x+1.5,\y+0.5) circle [radius=0.1]; 
      \draw [fill=black] (\x+1.5,\y+2.5) circle [radius=0.1]; 
      \draw[gray,rounded corners,fill=gray, opacity=0.1] (\x+0.75,\y-0.25) rectangle (\x+4,\y+1.25); 
      \draw[gray,rounded corners,fill=gray, opacity=0.1] (\x+0.75,\y+1.75) rectangle (\x+4,\y+3.25); 
    }
  }
  \foreach \x in {4,12,20} {
    \foreach \y in {-2,2,6,10} {
      \path [draw=blue,snake it, thick] (\x-0.5,\y+2.5) -- (\x+1.5,\y+2.5); 
      \path [draw=blue,snake it, thick] (\x-0.5,\y+0.5) -- (\x+1.5,\y+0.5); 
      \draw[ultra thick] (\x+1,\y) rectangle (\x+2,\y+3);
      \draw [fill=black] (\x-0.5,\y+0.5) circle [radius=0.1]; 
      \draw [fill=black] (\x-0.5,\y+2.5) circle [radius=0.1]; 
      \draw [fill=black] (\x+1.5,\y+0.5) circle [radius=0.1]; 
      \draw [fill=black] (\x+1.5,\y+2.5) circle [radius=0.1]; 
      \draw[gray,rounded corners,fill=gray, opacity=0.1] (\x+0.75,\y-0.25) rectangle (\x+4,\y+1.25); 
      \draw[gray,rounded corners,fill=gray, opacity=0.1] (\x+0.75,\y+1.75) rectangle (\x+4,\y+3.25); 
    }
  }
\end{tikzpicture}
    \caption{
        The circuit $W$ (\Cref{fig:brickCirc}) converted into a tensor network.
        We introduce a Bell pair for every position in the circuit (black dots connected by a wavy line) and apply the unitary operation corresponding to the location in the circuit (cf.\ \Cref{fig:state}).
        We then apply projectors on pairs of qubits (gray boxes).
    }
    \label{fig:brickTN}
\end{figure}
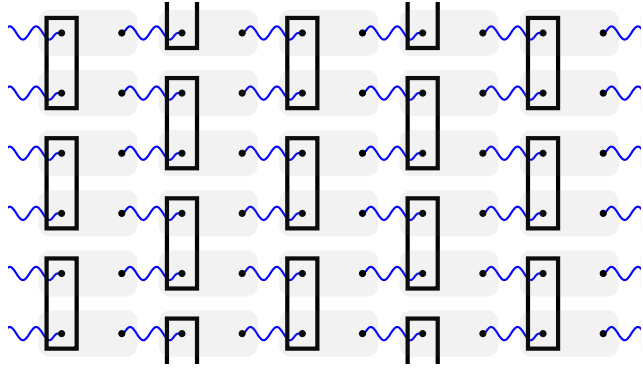

\subsection{Making the tensor network injective}\label{sec:injectiveTN}
We say that a tensor network is \emph{$\delta$-injective} when its local maps have singular values lower bounded by $\delta$. The tensor network defined in the previous section is non-injective since the projectors are singular. To make the tensor network injective, we follow the procedure in~\cite{anshu2022construction} and replace the projectors $P$ by a $\delta$-perturbation
\begin{equation}
    Q = \ket{\Phi_I} \bra{\Phi_I} + \delta \sum_{P \in \{X,XZ,Z\}} \ket{\Phi_P} \bra{\Phi_P}.
    \label{eqn:perturbed_projectors}
\end{equation}
Applying the invertible map $Q$ on every other pair of row-adjacent qubits in $\ket{\Phi_{W,\xi}} = \ket{\xi} \bigotimes_{\ell \in [D], p,q} \ket{\Phi^{(\ell)}_{p,q}}$ we obtain the injective PEPS state
\begin{equation}
    \ket{\Psi_{W,\xi}} \triangleq Q^{\otimes nD} \ket{\Phi_{W,\xi}} .
\end{equation}
We introduce several notations. Let $T$ be the number of gates, let $\ket{\Phi_{\Vec{P}}} = \bigotimes_{i=1}^T \ket{\Phi_{P_i}}$ for $\Vec{P} \in \mathcal{P}^{\otimes T}$ and let $|\Vec{P}|$ denote the number of nontrivial operators in $\Vec{P}$. 
Let $W_\ell =\otimes_{i \in \text{ layer } \ell} U_{i}$ and $\tilde{P}_\ell=\otimes_{i \in \text{ layer } \ell} P_{i}$ be the unitaries and the errors in the $\ell$-th layer of $W$. 
Abusing notation, we sometimes denote $U_i \in W_\ell$ and $P_i \in \Tilde{P}_\ell$ to mean that the unitaries and Pauli errors are in layer $\ell$.

The key observation is that the injective tensor network represents a noisy version of the quantum computation.

\begin{claim}\label{claim:gs}
The state $\ket{\Psi_{W,\xi}}$ can be expanded as $\ket{\Psi_{W,\xi}} \propto \sum_{\Vec{P} \in \mathcal{P}^{\otimes T}} \delta^{|\Vec{P}|} \ket{\Phi_{\Vec{P}}} \otimes  (U_{T} P_{T} \hdots U_1 P_1 \ket{\xi})$. 
\label{claim:iidnoise}
\end{claim}
\begin{proof} Expanding $Q$ we have 
\begin{align}
    \ket{\Psi_{W,\xi}} &= \sum_{\Vec{P} \in \mathcal{P}^{\otimes T}} \delta^{|\Vec{P}|} \ket{\Phi_{\Vec{P}}}  \bra{\Phi_{\Vec{P}}}\ket{\Phi_{W,\xi}}.
\end{align}
Performing teleportation for each term in the summand, we find $\bra{\Phi_{\Vec{P}}}\ket{\Phi_{W,\xi}}\propto U_{T} P_{T} \hdots U_1 P_1  \ket{\xi}$ as the state in the rightmost column of the tensor network.
\end{proof}
In other words, $\ket{\Psi_{W,\xi}}$, up to normalization, contains a noisy quantum computation with purified local depolarizing channels. The local i.i.d.\ depolarizing noise rate is $p=\delta^2/(1+3\delta^2)$. Here `purified' means that the EPR states in the bulk of the tensor network record the occurred errors.

~\Cref{claim:iidnoise} can be alternatively written as $\ket{\Psi_{W,\xi}} \propto  \sum_{\Vec{P} \in \mathcal{P}^{\otimes T}} \delta^{|\Vec{P}|} \ket{\Phi_{\Vec{P}}} \otimes  (W_{D} \tilde{P}_{D} \hdots W_1 \tilde{P}_1 \ket{\xi})$. 

We can define a unitary
\begin{equation}
\label{eq:undounit}
    V = \sum_{\Vec{P} \in \mathcal{P}^{\otimes nD}} \ket{\Phi_{\Vec{P}}} \bra{\Phi_{\Vec{P}} } \otimes  (W_{D} \tilde{P}_{D} \hdots W_1 \tilde{P}_{1}),
\end{equation}
such that 
$$V^{\dagger}\ket{\Psi_{W,\xi}} \propto \sum_{\Vec{P} \in \mathcal{P}^{\otimes T}} \delta^{|\Vec{P}|} \ket{\Phi_{\Vec{P}}} \otimes   \ket{\xi}.$$
Note that the state $\sum_{\Vec{P} \in \mathcal{P}^{\otimes T}} \delta^{|\Vec{P}|} \ket{\Phi_{\Vec{P}}}$ is a simple product state since the noise is i.i.d local. Thus, when $\ket{\xi}=\ket{0}^{\otimes n}$ (which arises for computations in $\BQP$), the state $\ket{\Psi_{W,\xi}}$ can be prepared by a quantum circuit. This is similar to the Feynman-Kitaev history state~\cite{kitaev2002classical}, which can be prepared efficiently for quantum computations in $\BQP$.

In general, we are not restricted to choosing the same injectivity parameter $\delta$ across the whole circuit. In fact, some of our later results are proved by varying $\delta$ between locations in the circuit.

\subsection{The parent Hamiltonian}\label{subsec:parentHam}
The nice property of the injective tensor network state $\ket{\Psi_{W,\xi}}$ is that it is the unique ground state of a local Hamiltonian. In particular, we consider the $n(2D+1)$-qubit Hilbert space containing the PEPS state~$\ket{\Psi_{W,\xi}}$ corresponding to a circuit $W$.

Let $\Lambda = \delta \ket{\Phi_I} \bra{\Phi_I} + \sum_{p \in \{X, XZ,Z\}}  \ket{\Phi_p} \bra{\Phi_p}$, such that $Q \propto \Lambda^{-1}$.

\begin{definition}[Parent Hamiltonian]\label{clm:parent_hamiltonian}
Associate for each gate two-qubit gate $U$ in the circuit an 8-qubit Hamiltonian term $h_U=\Lambda^{\otimes 4}(\id - \ket{\Phi_U} \bra{\Phi_U})\Lambda^{\otimes 4}$. Furthermore, suppose the initial state $\ket{\xi}$ is the unique ground state of a frustration-free local Hamiltonian $H_{\xi}= \sum_{j}  g_j$. Then the unnormalized state $\Phi_{W,\xi}$ is the unique ground state of the frustration-free Hamiltonian $H_\mathrm{parent} = \sum_{j} \Lambda^{\otimes N(j)} g_j \Lambda^{\otimes N(j)} + \sum_{U \in W} h_U $, where $N(j)$ is the set of EPR locations that have intersecting support with $g_j$ (see~\Cref{fig:combin}).
We refer to the first set of terms as~$H_\mathrm{in}$ (input) and the second set as~$H_\mathrm{prop}$ (propagation).
\end{definition}

\begin{figure}
    \centering
    \includegraphics[width=0.7\textwidth]{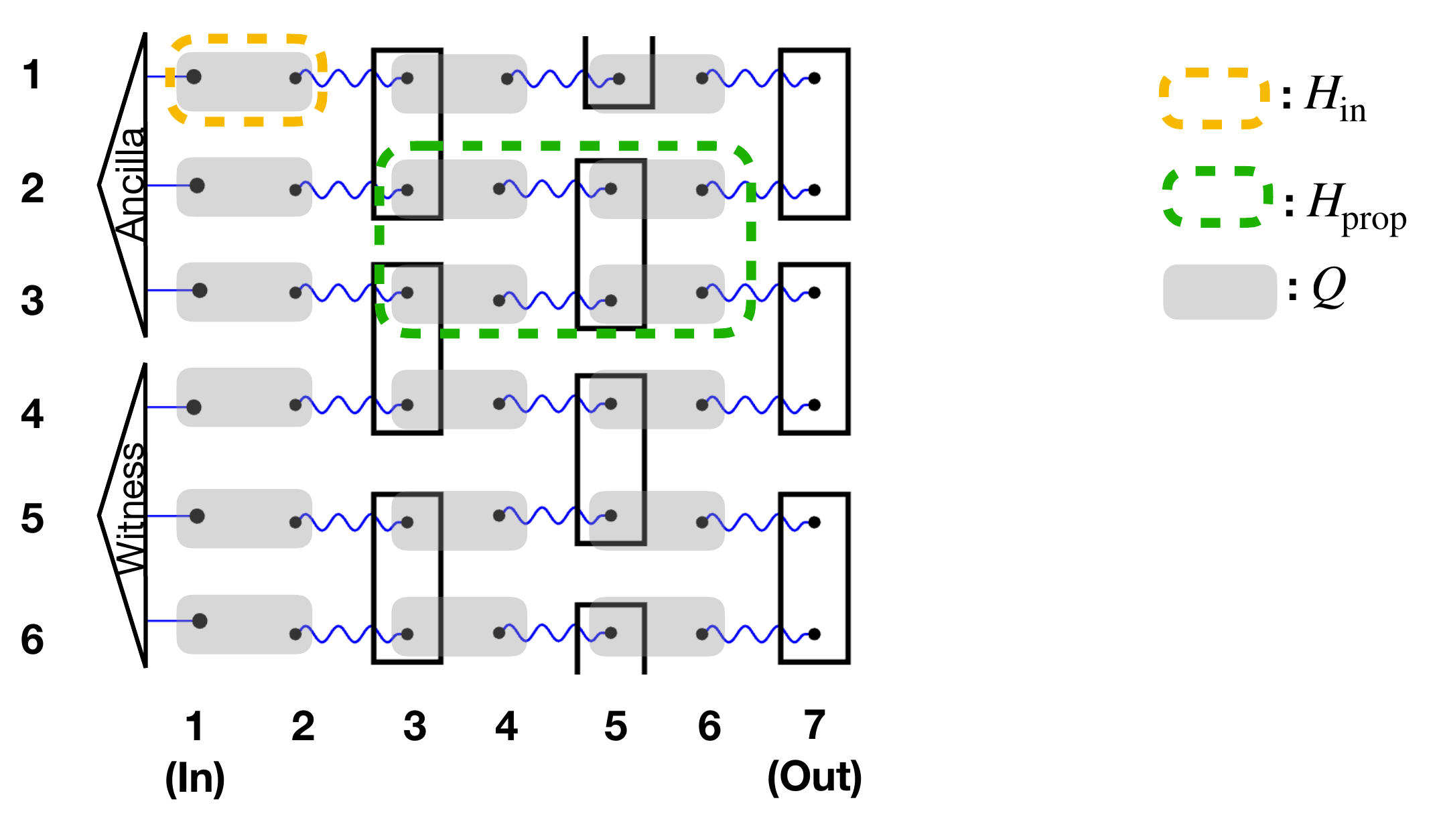}
    \caption{An injective PEPS encoding noisy quantum computation shown with $n=6$ qubits (black dots), of which $a=3$ are ancillas, and $D=3$ layers of two-qubit gates in the brickwork architecture. The computation goes from left to right, with qubits on column 1 being the input. \textbf{Gates} are encoded in 4-qubit Choi states (see~\Cref{fig:state}) placed on columns (2,3), (4,5), and so on. Applying the invertible map $Q$ (gray box) as defined in~\Cref{eqn:perturbed_projectors} generates a noisy computation on the last column (indexed 7). The qubit pairs where $Q$ is applied are called \emph{shifted} EPR locations. We refer to the last column of qubits in the PEPS as the \textit{output column}. \textbf{Noisy computation}: After $Q$ is applied, the output column can be interpreted as a noisy computation where for each layer of the circuit, the present noise pattern is specified by the EPR states at the shifted EPR locations (the word `shifted' is to avoid confusion with the original locations of the Choi state encodings: the shifted EPR locations are the same as shifting the original Choi state's locations one step to the left). Due to this correspondence, we refer to the first two columns (indexed 1,2) as the \textit{first layer}, the next two columns (indexed 3,4) as the \textit{second layer}, and so on. \textbf{Parent Hamiltonian}: A propagation term (dashed green) acts on $8$ qubits, while an initialization term (dashed yellow) acts on the first $2$ qubits and only on each ancilla row (indexed 1,2,3).}
    \label{fig:combin}
\end{figure}

An example of $H_\xi$ is $H_{\ket{0^n}} = \sum_{i=1}^{n} \ket{1}\bra{1}_i $ which ensures the input state is $\ket{\xi}= \ket{0}^{\otimes n}$. Or if we want $\ket{\xi}$ to be a codeword of a stabilizer code, we can take $H_\xi$ to be the code Hamiltonian.

In a $\QMA$ protocol, we do not fully  specify the initial state. Instead, the initial state is of the form $\ket{0^a}\ket{\xi}$, where $\ket{\xi}$ is any $(n-a)$-qubit witness coming from the prover. 
So $H_\mathrm{in} =  \sum_{j=1}^{a} \Lambda \ket{1}\bra{1}_j^\mathrm{in}  \Lambda$ has a ground space of degeneracy $2^{n-a}$ and so does~$H_\mathrm{parent}$. See \Cref{fig:combin} for an example of an injective PEPS and its parent Hamiltonian.
Later when we work with states of this form, we will continue denoting the ground states as $\ket{\Psi_{W,\xi}}$, leaving the ancillas $\ket{0^a}$ in the initial state implicit. Finally, similar to the Feynman-Kitaev construction, we can use an output check term $H_\mathrm{out} \triangleq \ket{0}\bra{0}^\mathrm{out}_{j}$ to verify qubit $j$ in the output.

\subsection{Bypassing the local indistinguishability issue}\label{subsec:bypass}

A conceptual challenge that any circuit-to-Hamiltonian construction must resolve is local indistinguishability. As discussed in \cite{Vidnotes20}, the argument is as follows. Consider a $n$-qubit quantum state $\ket{\psi}$ that is subject to either $\id_n$ unitary or the $Z_n$ unitary on the last qubit. It is possible that $\ket{\psi}$ and $Z_n\ket{\psi}$ are locally indistinguishable (such as in the context of CAT states). But then how can a local constraint detect the difference between the two actions of unitaries? The Feynman-Kitaev approach solves this problem by using the clock register - see \cite{Vidnotes20}.

In our context, noise plays a crucial role in handling the local indistinguishability issue and showing that local changes can be locally detected. Indeed, consider $h_{\id_n}$ and $h_{Z_n}$ as the two tensor network Hamiltonian terms corresponding to the two possible gates. Let $\pi_{\id_n}$ and $\pi_{Z_n}$ be their respective ground spaces. If $\ket{\psi}$ was subject to the $\id_n$ unitary, the corresponding tensor network state $\ket{\Psi}$ would be in the support of $\pi_{\id_n}$. Thus, we can lower bound the energy of $\ket{\Psi}$ with respect to $h_{Z_n}$ by upper bounding 
$$\|\pi_{Z_n}\ket{\Psi}\|=\|\pi_{Z_n}\pi_{\id_n}\ket{\Psi}\|\leq \|\pi_{Z_n}\pi_{\id_n}\|.$$
We argue that $\|\pi_{Z_n}\pi_{\id_n}\|\leq 1- \delta^6/2$ when $\delta <\frac{1}{4}$. For this, we will show that the overlap between any two vectors from the two subspaces is $\leq 1- \frac{\delta^6}{2}$.

We have the following characterization of the ground space of the propagation term $h_U$ for general one-qubit gate $U$ (the two-qubit case can be similarly derived).
\begin{claim} A general vector in the ground space of $h_U = \Lambda^{\otimes 2}(\id - \ket{\Phi_U}\bra{\Phi_U})\Lambda^{\otimes 2}$ can be written as 
$\sum_{\vec{P}=(P_1,P_2) \in \mathcal{P}^{\otimes 2}}\delta^{|\vec{P}|}\Tr[MP_1 U P_2]\ket{\Phi_{\vec{P}}}$, for arbitrary one-qubit operator $M$.
\end{claim}
\begin{proof} The ground space of $ (\Lambda^{-1})^{\otimes 2} h_U (\Lambda^{-1})^{\otimes 2}$ is $\mathrm{span}\{\ket{\psi_1}\ket{\Phi_U}\ket{\psi_2}: \forall \ket{\psi_1},\ket{\psi_2} \in \mathbb{C}^{2} \}$. Thus the ground space of $h_U$ is
\begin{align*}
&\mathrm{span}_{\ket{\psi_1},\ket{\psi_2}} \{\Lambda^{\otimes 2} \ket{\psi_1}\ket{\Phi_U}\ket{\psi_2}\} \\
&=\mathrm{span}_{\ket{\psi_1},\ket{\psi_2}} \left\{\sum_{\Vec{P} \in \mathcal{P}^{\otimes 2}} \bra{\psi_1}P_1 U P_2 \ket{\psi_2} \ket{\Phi_{\Vec{P}}} \right\}\\
&=\mathrm{span}_{M \in \mathbb{C}^{2 \times 2}} \left\{\sum_{\Vec{P} \in \mathcal{P}^{\otimes 2}} \Tr[M P_1 U P_2] \ket{\Phi_{\Vec{P}}} \right\}. \qedhere
\end{align*}
\end{proof}

Applying the above claim, a general state in $\pi_{\id_n}$ can be written as 
$\sum_{\vec{P} \in \mathcal{P}^{\otimes 2}}\delta^{|\vec{P}|}\Tr{MP_1P_2}\ket{\Phi_{\vec{P}}},$
for arbitrary operator $M$ subject to the normalization condition $\sum_{\vec{P}\in \mathcal{P}^{\otimes 2}}\delta^{2|\vec{P}|}|\Tr(MP_1P_2)|^2=1$. Similarly, a general state in $\pi_{Z_n}$ can be written as $\sum_{\Vec{P}\in \mathcal{P}^{\otimes 2}}\delta^{|\vec{P}|}\Tr{ZNP_1ZP_2}\ket{\Phi_{\vec{P}}}$, for arbitrary operator $N$ (we add a Pauli Z in front of $N$ for later convenience) subject to the normalization condition $\sum_{\vec{P}\in \mathcal{P}^{\otimes 2}}\delta^{2|\vec{P}|}|\Tr(NP_1P_2)|^2=1$. It is clear that the two vectors must be different - no matrices $M,N$ satisfies $\Tr(MP_1P_2) = \Tr(ZNP_1ZP_2) = (-1)^{\text{Ind}(P_2\in{X,Y})}\Tr(NP_1P_2)$ for all Paulis $P_1,P_2$, where $\text{Ind}$ denotes the indicator variable. Otherwise, $M=(-1)^{\text{Ind}(P_2\in{X,Y})}N$ for all $P_2$, which forces $M=0$. We can also get a quantitative bound by arguing that the $\ell_1$ distance
$\sum_{\vec{P} \in \mathcal{P}^{\otimes 2}}\delta^{2|\vec{P}|}|\Tr((M-(-1)^{\text{Ind}(P_2\in{X,Y})}N)P_1P_2)|^2$ must be $\geq \delta^6$. For contradiction, suppose the opposite holds. Then for all $P_1,P_2$, $|\Tr((M-(-1)^{\text{Ind}(P_2\in{X,Y})}N)P_1P_2)|\leq \delta$. For a fixed $P_1P_2$, we can choose $P_1,P_2$ such that $P_2\in \{\id, Z\}$ as well as $P_2\in \{X, Y\}$. This means we have 
$$|\Tr((M\pm N)P_1P_2)|\leq \delta, \implies |\Tr(M P_1P_2)|\leq \delta.$$ The implication uses triangle inequality. This forces the normalization condition to be $$\sum_{\vec{P}\in \mathcal{P}^{\otimes 2}}\delta^{2|\vec{P}|}|\Tr(MP_1P_2)|^2\leq 16\delta^2 <1,$$ a contradiction.

\subsection{Connection with prior works}\label{subsec:priorworks}

A scheme related to ours is that of  Ref.~\cite{bacon2017sparse} in which the authors give a construction of quantum error-correcting subsystem codes with almost linear distance.
Their construction can be understood as a map from fault-tolerant Clifford circuits that facilitate check measurements to a set of non-commuting Pauli-check operators.
More concretely, each location in the circuit is associated with a qubit and each Clifford gate is associated with a Pauli operator that stabilizes the gate.
For example, the idling gate (wire) is stabilized by~$XX$ and~$ZZ$ operating on the in- and out-locations. 
The main difference with our setting is that we do not need to assume Clifford circuit. Furthermore, our Hamiltonian remains frustration-free, whereas the Hamiltonian in Ref~\cite{bacon2017sparse} is frustrated. Another difference is that we associate two qubits per circuit location that are projected onto an EPR state, cf.\ \Cref{fig:brickTN}.

In Ref.~\cite{bartlett2006simple} Bartlett and Rudolph show using PEPS that a fault-tolerant cluster state, which is a universal resource state for MBQC, can be robustly encoded into the ground state of a Hamiltonian consisting of planar, 2-local interaction terms.
They also note that the approximation error can be interpreted as stochastic Pauli-noise and that the energy gap of their construction is independent of system size.
The difference to our approach is that Bartlett and Rudolph use tensor networks to obtain a resource state that can be used for quantum computation via MBQC, whereas our scheme encodes a quantum computation into a tensor network.

In \cite{aharonov2022translationally} Aharonov and Irani consider a mapping of classical computation into a CSP, which we may think of as a classical local Hamiltonian.
More concretely, they consider a two-dimensional $L\times L$ grid with translation invariant constraints and show that approximating the ground state energy to an additive $\Theta(\sqrt[4]{L})$ is $\NEXP$-complete.
They do so by encoding a computation into a tiling problem.
The computation is fault tolerant by running the same computation several times in parallel to enforce a large cost for an incorrect computation.
In contrast, our model is fully quantum and thus requires the quantum fault tolerance theorem of Ref.~\cite{aharonov1999fault}.

\section{Background}\label{sec:background}

\subsection{Hamiltonian complexity}

Here, we give a brief introduction to the complexity class $\QMA$ and main lemmas used in this work.

\begin{definition} The class $\QMA_w[c,s]$ is the class of promise problems $A=(A_\mathrm{yes}, A_\mathrm{no})$ with the property that, for
every instance $x$, there exists a uniformly generated verifier quantum circuit $V_x$ with the following properties: $V_x$ is of size $\poly(|x|)$ and acts on an input state $\ket{0^{\otimes m}}$ together with a witness state $\ket{\xi}$ of size $w$ supplied by an all-powerful prover, with both $m,w=\poly(|x|)$. Upon measuring the decision qubit $o$, the verifier accepts if $o=1$, and rejects otherwise. If $x \in A_\mathrm{yes}$, then $\exists \ket{\xi}$ such that $\Pr[o=1] \geq c$ (completeness). If $x \in A_\mathrm{no}$, then $\forall \ket{\xi}$, $\Pr[o=1]\leq s$ (soundness), such that $c-s \geq 1/\poly(|x|)$.
\end{definition}

It is well-known that the parameters $c,s$ can be amplified, even without increasing the witness size.
\begin{lemma}[Weak $\QMA$ amplification~\cite{kitaev2002classical}] For any $r=\poly(|x|)$, $\QMA_w[2/3,1/3]=\QMA_{w'}[1-2^{-r}, 2^{-r}]$ where $w'=\poly(w)$.
\label{lem:amplify}  
\end{lemma}

\begin{lemma}[Strong $\QMA$ amplification~\cite{marriott2005quantum}] For any $r = \poly(|x|)$, $\QMA_w[2/3,1/3]=\QMA_{w}[1-2^{-r}, 2^{-r}]$.
\label{lem:strongamplify}
\end{lemma}

\begin{definition}[$k$-Local Hamiltonian problem]
\textbf{Input}: $H_1,H_2,\hdots,H_T$ set of $T=\poly(n)$ Hermitian matrices with bounded spectral norm $\|H_i\| \leq 1$ acting on the Hilbert space of $n$ qubits. In addition, each term acts nontrivially on at most $k$ qubits and is described by $\poly(n)$ bits. Furthermore, we are given two real numbers $a, b$ (described by $\poly(n)$ bits) such that $b-a > 1/\poly(n)$. \textbf{Output}: Promised either the smallest eigenvalue of $H = H_1 + H_2 + \hdots H_T$ is smaller than $a$ or all eigenvalues are larger than $b$, decide which case it is. We denote this problem by $k$-$\LH[a,b]$, or sometimes, $k$-$\LH(b-a)$.
\end{definition}

The $k$-$\LH$ is in $\QMA$ for any $k=O(\log n)$ (see e.g., Theorem 1 in~\cite{aharonov2002quantum}). Furthermore, Kitaev showed in his seminal work~\cite{kitaev2002classical} that $5$-$\LH$ is $\QMA$-complete.

\begin{theorem}[Kitaev~\cite{kitaev2002classical}] 
\label{thm:kitaev}
Any $\QMA_w[c,s]$ protocol involving an $n$-qubit verifier circuit with $T=\poly(n)$ gates can be turned into a $5$-$\LH[a,b]$ on $\poly(n)$ qubits with $a=O((1-c)/T)$ and $b = \Omega((1-\sqrt{s})/T^3 )$.
\end{theorem}

We will often simply write $\QMA$, $\LH$ when the parameters are unimportant or clear from context.

Next, we need the following lemmas in this work.

\begin{lemma}[Detectability lemma~\cite{anshu2016simple}] Let $\{Q_1,\hdots,Q_m\}$ be a set of projectors and $H = \sum_{i=1}^m  Q_i $. Assume that each $Q_i$ commutes with all but $g$ others. Given a state $\ket{\psi}$, define $\ket{\phi} := \prod_{i=1}^m (\id -Q_i) \ket{\psi}$, where the product is taken in any order, and let $e_\phi = \bra{\phi} H \ket{\phi}/ \|\phi\|^2$. Then
\begin{equation*}
    \left\| \phi \right\|^2 \leq \frac{1}{e_\phi/g^2 +1}.
\end{equation*}
\label{lem:DL}
\end{lemma}

\begin{lemma}[Quantum union bound~\cite{gao2015quantum}] Consider the same setting as in~\Cref{lem:DL}, but this time we do \emph{not} require each $Q_i$ to commute with at most $g$ others. It holds that
\begin{equation*}
    \| \phi \|^2 \geq 1- 4 \bra{\psi} H \ket{\psi}.
\end{equation*}
\label{lem:quantum-union}
\end{lemma}

\begin{lemma}[Jordan's lemma~\cite{jordan1875essai}] Given two projectors $\Pi_1$, $\Pi_2$ acting on a $d$-dimensional complex vector space $\mathcal{H}$, there exists a change of basis such that $\mathcal{H}$ is decomposed as a direct sum of one- or two-dimensional  mutually orthogonal subspaces $\mathcal{H}= \bigoplus_i \mathcal{H}_i$, such that both the projectors leave the subspaces invariant. In other words, we can write $\Pi_1 =  \sum_{i} a_i \ket{u_i}\bra{u_i}$ and $\Pi_2 =  \sum_{i} b_i \ket{v_i}\bra{v_i}$, with $\ket{u_i},\ket{v_i} \in \mathcal{H}_i$ and $a_i,b_i  \in \{0,1\}$.
\label{lem:jordan}
\end{lemma}

\begin{lemma}[Geometric lemma~\cite{kitaev2002classical}] 
\label{lem:geometric}
Let $A$, $B$ be nonnegative Hermitian operators and $\mathrm{g.s.}(A)$, $\mathrm{g.s.}(B)$ be their null subspaces such that the angle between them is $\theta >0$. Suppose further that no nonzero eigenvalue of $A$ or $B$ is smaller than $\gamma$. Then $A+B \geq \gamma (1-\cos \theta).$
\end{lemma}

\subsection{Fault tolerance}

When defining our model in \Cref{sec:model}, we introduced perturbations to make the tensor network injective.
This ensures the existence of an associated Hamiltonian and avoids the model from becoming too powerful~\cite{schuch2007computational}.
Remarkably, the perturbations can be interpreted as undesired Pauli errors, see \Cref{claim:gs}.
These `errors' disturb our computation, leading to a degradation of the output, just as they would in a physical device.
We can remedy this problem by substituting the circuit~$W$ with a fault-tolerant version of itself~$\tilde{W}$, thereby guaranteeing robustness against the errors.
In this section we will briefly summarize some results of fault-tolerant quantum computing that we require for our construction.

\subsubsection{Quantum error correcting codes}
Quantum error correcting codes are subspaces of the full Hilbert space of $n$ bits.
Each quantum code has three parameters:
The number of logical qubits $k$ tells us that the code protects a state vector of $k$ qubits.
The number of physical qubits $n$ refers to the number of qubits into which the $k$ logical qubits are being encoded.
Finally, the distance $d$ refers to the minimum number of single-qubit Pauli errors that are needed to map one encoded state onto another.
In particular, a quantum code of distance $d$ can correct any error acting on less than $d/2$ of the physical qubits.
We refer to Ref.~\cite{terhal2015quantum} for details.

\paragraph{CSS codes} The most studied class of quantum codes are Calderbank-Shor-Steane (CSS) codes, which are specified by an $ r_X \times n$ matrix $H_X$, whose rows represent $X$-checks and a $k \times n$ matrix $L_X$ whose rows represent Pauli $X$-logicals. The $Z$ checks are $r_Z \times n$ matrix $H_Z = \ker \begin{pmatrix}
    H_X\\
    L_X
\end{pmatrix}$ and $Z$-logicals are $k \times n$ matrix $L_Z$. 
The codewords in the logical $Z$ basis are
\begin{align}
    \ket{v}_L  &= \sum_{ u \in \mathbb{F}_2^{r_X}} \ket{uH_X + v L_X} \qquad \text{(strings of 0 and 1)}
        \label{eq:CSScodeword}
\end{align}
For CSS codes, qubit-wise CNOTs between two code blocks apply logical CNOTs between corresponding pairs of logical qubits. Indeed,
\begin{align*}
    \ket{v}_L \ket{w}_L = &\sum_{ u \in \mathbb{F}_2^{r_X}} \ket{u H_X + v L_X} \sum_{u' \in \mathbb{F}_2^{r_X}} \ket{u' H_X + w L_X}\\
    \longrightarrow  &\sum_{ u \in \mathbb{F}_2^{r_X}} \ket{u H_X + v L_X} \sum_{u' \in \mathbb{F}_2^{r_X}} \ket{u' H_X + (v+w) L_X}\\
     = & \ket{v}_L \ket{v+w}_L.
\end{align*}

We note that the existence of quantum codes does not guarantee that quantum computing can be made robust against noise. Manipulating the encoded states via an error prone process leads to errors spreading and it is this spread of errors that needs to be controlled.

\subsubsection{Quantum fault tolerance}
In a seminal result, Shor showed that when any component of a quantum circuit, such as state preparation, gates and measurements, is replaced by a fault-tolerant version, it is possible to reduce errors under the assumption that the error rate per time step is polylogarithmically small in the length of the computation.
Aharonov and Ben-Or~\cite{aharonov1999fault} and Knill, Laflamme and Zurek~\cite{knill1998resilient} extended Shor's approach with a concatenation scheme.
The main idea is as follows:
At the top level each qubit is encoded into a quantum code~$\mathcal{C}_1$ using~$n_1$ physical qubits and with distance $d_1$.
Next, each physical qubit of~$\mathcal{C}_1$ is encoded further into a second code~$\mathcal{C}_2$ using~$n_2$ physical qubits and with distance~$d_2$.
This way, we have effectively a new code using~$n_1 n_2$ physical qubits and which has distance $d_1  d_2$.
Assuming that $\mathcal{C}_1$ and $\mathcal{C}_2$ come with a fault-tolerant set of circuit components, so does the concatenated code.
Crucially, taking $\mathcal{C}_1 = \mathcal{C}_2$ the failure probability of any circuit component in the top level is now bounded by $c \left(c p^2 \right)^2 = c^3 p^4$.
Continuing this process, if we concatenate the same code~$a$ times, the probability of failure of any top level component is bounded by $c^{-1}(c p)^{2^a}$.
Let $s^a$ be the the circuit size at the $a$th level of concatenation.
While the size of the circuit grows exponentially, the error is reduced double exponentially.
Hence, fixing some desired error rate $\epsilon = c^{-1}(c p)^{2^a}$ leads to $s^a = \Theta\left(\operatorname{polylog}\left( \frac{1}{\epsilon}\right)\right)$.
In summary, concatenation allows us to simulate a quantum circuit with component failure rate bounded by an arbitrarily small $\epsilon$ using components with error rate bounded by some constant error rate $p$, as long as the initial error rate is below a threshold value set by the combinatorial factor $c$.

\begin{theorem}[\cite{aharonov1999fault}, Theorem 12]\label{thm:quantumFT}
There exists a noise threshold $\eta_c > 0$ such that for any $\eta < \eta_c$, $\varepsilon > 0$ the following holds. For any $n$-qubit quantum circuit~$C$ with $s$ gates, $\ell$ locations, and depth $D$, there exists a quantum circuit $\tilde{C}$ of size $s \polylog(\ell/\varepsilon)$ (no measurements or classical operations are required) and depth $D \polylog(\ell/\varepsilon)$ operating on $n \polylog(\ell/\varepsilon)$ qubits such that in the presence of local depolarizing noise with error rate $\eta < \eta_c$, the encoded output of~$\tilde{C}$ is $\varepsilon$-close to that of~$C$.
\end{theorem}
The theorem above does assume all-to-all connectivity, i.e.\ gates can be applied on arbitrary sets of qubits.
We can also constrain the circuit to only operate locally on a $d$-dimensional grid of qubits, so that two qubit gates are only applied between neighbours on the grid.
Note that an arbitrary circuit can be turned into a $d$-dimensional circuit by introducing SWAP gates and ancilla qubits, leading to the following result for any $d \geq 1$.

\begin{corollary}[\cite{aharonov1999fault}, Theorem 13]\label{thm:quantumFTdDim}
There exists a noise threshold $\eta_c > 0$ such that for any $\eta < \eta_c$, $\varepsilon > 0$, and $d\geq 1$ the following holds. For any $d$-dimensional $n$-qubit quantum circuit~$C$ with $s$ gates, $\ell$ locations, and depth $D$, there exists a $d$-dimensional quantum circuit $\tilde{C}$ of size $s \polylog(\ell/\varepsilon)$ (no measurements or classical operations are required) and depth $D \polylog(\ell/\varepsilon)$ operating on $n \polylog(\ell/\varepsilon)$ qubits such that in the presence of local depolarizing noise with error rate $\eta < \eta_c$, the encoded output of~$\tilde{C}$ is $\varepsilon$-close to that of~$C$.
\end{corollary}

\section{Soundness properties of the parent Hamiltonian}\label{sec:soundness}

We have seen that the ground state of the parent Hamiltonian $H_\mathrm{parent} = H_\mathrm{in} + H_\mathrm{prop}$ encodes a quantum computation with stochastic depolarizing noise. In this section,  we investigate how robust this encoding is. In particular, we consider combinatorial and low-energy states $H_\mathrm{parent}$ and show that such states of sufficiently low energy density still encode the quantum computation, but now with \emph{adversarial} noise. We define the adversarial computation model in~\Cref{sec:adv-noise}. We prove the main soundness theorems for combinatorial states in~\Cref{subsec:combsound} and low-energy states in~\Cref{subsec:exposound} (with lemmata in~\Cref{sec:deferred-proofs}). We prove a spectral gap lower bound for the parent Hamiltonian in~\Cref{subsec:spectral-gap}.
Finally in~\Cref{sec:advFT} we describe a simple parallel repetition trick to achieve (classical and quantum) fault tolerance against any inverse-polynomial fraction of adversarial noise per layer.

\subsection{Adversarially noisy computation}\label{sec:adv-noise}

Recall that we consider a circuit $W$ with initial state of the form $\ket{0^a}\ket{\xi}$, where $\ket{\xi}$ is any $(n-a)$-qubit representing the witness coming from the $\QMA$ prover (for $\BQP$ computations, $\ket{\xi}$ would be empty). Our starting point is the intuition that violated terms in $H_\mathrm{parent}$ should correspond to faults in the circuit. Informally, violated terms in~$H_\mathrm{in}$ should correspond to errors at a set of locations, denoted $S_0$, in the qubit initialization step, violated terms in the first layer of~$H_\mathrm{prop}$ should correspond to gate errors at locations, denoted $S_1$, in the circuit's first layer, and so on. The rest of this section is to make this connection between violated Hamiltonian terms and gate faults rigorous. However, these gate faults are adversarial in the sense that the faulty locations are chosen arbitrarily by the adversary. 
Thus, let us define the notion of adversarially noisy computations.

\begin{definition} 
Suppose $S = \{S_0,\hdots,S_D\}$, where $S_\ell \subseteq [n]$ for $0\leq \ell \leq D$, is a set of locations in a depth-$D$ $n$-qubit circuit. We define $\err(S)= \{ \Vec{E}  \in \mathcal{P}^{\otimes n(D+1)}: \loc(\Tilde{E}_\ell) \subseteq S_\ell, \, 0 \leq \ell \leq D  \}$ to be the set of Pauli errors supported within the set of locations $S$.
\end{definition}

\begin{definition}[Adversarial computations]\label{def:adversarial-noise} For any sets of locations $S_\ell \subseteq [n]$, for $0 \leq \ell \leq D$ in the circuit $W$, a state $\ket{\psi}$ is said to be an adversarial computation at locations $S=\{S_0,S_1,\hdots,S_D\}$ if 
\begin{align*}
    \ket{\psi} &\in  \mathrm{adv}(W,S) \triangleq \operatorname{span}\{\tilde{E}_{D} W_D  \hdots  \tilde{E}_{1} W_1  \tilde{E}_{0}\ket{0^a}\ket{\xi}: \forall \ket{\xi}, \Vec{E} \in \err(S) \}.
\end{align*}
We consider adversarial computations such that at most $\varepsilon n$ adversarial errors are present in the circuit.
In particular, we
say a state $\ket{\psi}$ is an $\varepsilon$-noisy state if $$\ket{\psi} \in \mathrm{adv}_\varepsilon (W) \triangleq \operatorname{span}\{\mathrm{adv}(W,S) : \sum_{\ell=0}^{D} |S_\ell| \leq \varepsilon n \} .$$ 
A mixed state $\rho$ is $\varepsilon$-noisy if it is a convex combination of $\varepsilon$-noisy pure states.
\end{definition}

Our main theorems are the following soundness results.

\begin{theorem}[Soundness]
\label{thm:energy} Suppose the depth $D = O(\log n)$ and consider any injectivity parameter $\delta= O(D^{-0.51})$.
For any state $\ket{\psi}$ with energy density $\frac{\delta^{200D}}{D+1}$ with respect to $H_\mathrm{parent}$, the reduced $\psi_\mathrm{out}$ in the output column is $\frac{1}{10}$-close in trace distance to a $400\delta^{2}D$-noisy mixed state.
\end{theorem}

We also prove a ``combinatorial'' version.

\begin{theorem}[Combinatorial soundness]
\label{thm:combin} There exists a constant $\varepsilon_0$ such that the following holds. Consider any injectivity parameter $\delta = O(D^{-0.51})$ and any $10\delta \sqrt{D}< \varepsilon < \varepsilon_0$. Then for any state $\ket{\psi}$ that satisfies all but $\frac{\varepsilon}{D+1}$ fraction of terms in $H_\mathrm{parent}$, the reduced state~$\psi_\mathrm{out}$ in the output column is $e^{-99n}$-close in infidelity to an $8\varepsilon$-noisy mixed state.
\end{theorem}

\begin{remark} The theorem statements and proofs below are presented assuming all $n$ qubits are intialized at the beginning of the computation for simplicity. However, they can be readily adapted to the setting where qubits are initialized at varying times such as in quantum fault tolerance. In this case, $D$ is defined to be the longest elapse time between an output qubit and the initialization of any qubit causally connected to it.
\end{remark}

\subsection{Proof of \Cref{thm:combin} (Combinatorial soundness)}
\label{subsec:combsound}

\noindent \textbf{Proof idea:} The combinatorial state $\ket{\psi}$ has the property that the (unnormalized) state $\Lambda^{\otimes nD}\ket{\psi}$ has a nice form - $\left(\bigotimes_{i \notin S_0} \ket{0}_i  \right) \left(\bigotimes_{\loc(U) \notin S} \ket{\Phi_U}\right) \otimes \ket{\psi''}$. This means that we have the correct state $\ket{0}$ or $\ket{\Phi_U}$ corresponding to the satisfied Hamiltonian terms and an arbitrary state $\ket{\psi''}$ at the violated terms. If $\ket{\psi''}$ were of the form $\bigotimes_{j\in S}\ket{\Phi_{U_j}}$ for some 2-qubit unitaries $U_j$, then we could simply view the state $\ket{\psi}$ as encoding the circuit with iid noise on non-faulty locations, and adversarial noise at faulty locations. This would be a perfectly fine combination of stochastic error and small number of adversarial errors. But $\ket{\psi''}$ can be a superposition of the states of above form, which can arbitrarily correlate the noise at non-faulty locations! We appeal to the injectivity of the local maps $\Lambda$ to argue that despite this possible correlation of noise at non-faulty locations, the fraction of errors stays at $O(\delta^2)$ (with high probability). Thus a damaging situation, for example all the non-faulty locations experiencing a Pauli error, continues to occurs with very small probability.

\vspace{0.1in}

\noindent \textbf{Proof:} 
Consider a $\frac{\varepsilon}{D+1}$-combinatorial state $\ket{\psi}$ and let $S=\{S_0, S_1,\hdots, S_D\}$ be the sets of faulty locations in each layer of the circuit corresponding to the violated Hamiltonian terms in $H_\mathrm{parent}= H_\mathrm{in} + H_\mathrm{prop}$. Since $H_\mathrm{prop}$ consists of at most $nD$ terms and $H_\mathrm{in}$ consists of $a \leq n$ terms, it holds that $\sum_{\ell} |S_\ell| \leq  2 \varepsilon n$ (assuming circuit consists of two-qubit gates). Below, we refer to the last column of qubits in the tensor network as the \textit{output column}, the first two columns as the \textit{first layer}, the next two columns as the \textit{second layer}, and so on. Given a $n(2D+1)$-qubit PEPS state $\ket{\psi}$, we denote by $\psi^{(\ell)}_j$ the two-qubit reduced state on the $j$-th row of the $\ell$-th layer and by $\psi^\mathrm{out}_j$ the one-qubit reduced state on the $j$-th row of the output column.

\vspace{0.1in}

We first prove the following lemma which asserts that the reduced state on the output column of the combinatorial state $\ket{\psi}$ contains the result of a quantum computation with both stochastic noise \textit{and} adversarial noise. Later we will combine these two noise models into just adversarial noise.

\begin{claim} Let $\widetilde{W}_{\Vec{P},\Vec{E}}=\tilde{E}_{D} W_D \tilde{P}_{D}  \hdots \tilde{E}_{1} W_1 \tilde{P}_{1} \tilde{E}_{0}$ denote the erroneous circuit with Pauli errors $\Vec{P}$ and $\Vec{E}$. Suppose $\ket{\psi}$ is a normalized state which satisfies all but terms at locations $S$ in $H_\mathrm{parent}$. Then the reduced state on the output column is
\begin{equation}
    \psi_\mathrm{out} \propto  \sum_{\Vec{P} \in \mathcal{P}^{\otimes nD}}
    \delta^{2 |\Vec{P}|}  
    \left( 
 \sum_{\Vec{E} \in \err(S) }   c_{\Vec{E}} \widetilde{W}_{\Vec{P},\Vec{E}} (\ket{0^a}\otimes \ket{\xi_{\Vec{E}}}) \right) \left( 
 \sum_{\Vec{E} \in \err(S) } c_{\Vec{E}}   (\bra{0^a} \otimes \bra{\xi_{\Vec{E}}} ) \widetilde{W}^\dagger_{\Vec{P},\Vec{E}} \right),
\end{equation}
where $\ket{\xi_{\Vec{E}}}$ are normalized states and the real coefficients $c_{\Vec{E}}$ satisfy $\sum_{\Vec{E} \in \err(S)} c_{\Vec{E}}^2 =1$. In other words, the state $\ket{\psi}$ encodes a noisy computation where the errors come from two sources: (1) stochastic noise coming from the tensor network injectivity and (2) \textit{adversarial} errors coming from the energy violations. 
\end{claim}

\begin{proof}
Let us analyze the terms in each of $H_\mathrm{in}$ and $H_\mathrm{prop}$. Consider the state $\ket{\psi'} \triangleq \frac{\Lambda^{\otimes nD} \ket{\psi}}{\|\Lambda^{\otimes nD} \ket{\psi}\|}$, where we recall that
\begin{equation}
    \Lambda= \delta \ket{\Phi_I} \bra{\Phi_I} +  \sum_{p \in \{X, XZ,Z\}}  \ket{\Phi_p} \bra{\Phi_p}.
\end{equation}
Consider a satisfied initialization term $h_j^\mathrm{in} = \Lambda ( \ket{1}\bra{1}_j) \Lambda $ (which acts on the $j$-th qubits of the first and second columns) in $H_\mathrm{in}$. Since $h_j^\mathrm{in} \ket{\psi}=0$, the reduced state of $\ket{\psi'}$ on the first column's qubit $i$ is $\ket{0}_i$. Similarly, for a satisfied propagation term $h_U=\Lambda^{\otimes 4}(\id - \ket{\Phi_U} \bra{\Phi_U})\Lambda^{\otimes 4}$ corresponding to a gate $U$ acting on qubits $i,j$ at time $t$, the reduced state  of $\ket{\psi'}$ on the 4 qubits of $\loc(U)$ must exactly be $\ket{\Phi_U}$. So it holds that 
\begin{equation}
    \ket{\psi'} = \left(\bigotimes_{i \notin S_0} \ket{0}_i  \right) \left(\bigotimes_{\loc(U) \notin S} \ket{\Phi_U}\right) \otimes \ket{\psi''},
    \label{eq:psi'}
\end{equation}
where $\ket{\psi''}$ is an arbitrary state supported on the remaining qubits (corresponding to the faulty locations in the circuit, including initialization, and the arbitrary witness state). The state $\ket{\psi''}$ can be further expressed in the orthonormal bases $\{\ket{0},X \ket{0}\}$ on the input qubits and $\{\ket{\Phi_U}, \id \otimes X \ket{\Phi_U}, \id 
 \otimes XZ \ket{\Phi_U}, \id \otimes Z \ket{\Phi_U}\}$ on the qubit pairs in the ``bulk'', such that
 \begin{equation}
     \ket{\psi'} =  \sum_{\Vec{E} \in \err(S) } c_{\Vec{E}} \left(\Tilde{E}_0 \ket{0^{a}} \ket{\xi_{\Vec{E}}} \right) \bigotimes_{\ell \in D} \left((\id \otimes \Tilde{E}_\ell )\bigotimes_{U \in W_\ell} \ket{\Phi_U}\right),
 \end{equation}
where $\ket{\xi_{\Vec{E}}}$ are normalized states and the coefficients $c_{\Vec{E}}$ are real (w.l.o.g) and satisfy $\sum_{\Vec{E} \in \err(S)} c_{\Vec{E}}^2=1$. Note that $c_{\Vec{E}}$ is nonzero only when $\Tilde{E}_0$ consists of only the Pauli operators $I$ and $X$.

Next, we undo the maps $\Lambda$ to obtain the original combinatorial state by applying the map $Q = \ket{\Phi_I} \bra{\Phi_I} + \delta \sum_{p \in \{X, XZ,Z\}}  \ket{\Phi_p} \bra{\Phi_p} \propto  \Lambda^{-1}$ (see~\Cref{eqn:perturbed_projectors}) on $\ket{\psi'}$. We have that
\begin{align}
    \ket{\psi} &\propto \ket{\chi} =  Q^{\otimes nD} \ket{\psi'} \\ & \propto \sum_{\Vec{P} \in \mathcal{P}^{\otimes nD}}
    \delta^{|\Vec{P}|}  \ket{\Phi_{\Vec{P}}} \bigotimes \left(\sum_{\Vec{E} \in \err(S) } c_{\Vec{E}}    \tilde{E}_{D} W_D \tilde{P}_{D}  \hdots \tilde{E}_{1} W_1 \tilde{P}_{1} \tilde{E}_{0}\ket{0^a}\ket{\xi_{\Vec{E}}} \right),
    \label{eq:chi}
\end{align}
where the last equality follows from linearly extending~\Cref{claim:iidnoise}.

Tracing out the bulk EPR states $\ket{\Phi_{\Vec{P}}}$ we obtain the statement of the claim.
\end{proof}

We now show that the normalized state $\psi_\mathrm{out}$ is exponentially close to an $(\alpha+2\varepsilon)$-noisy (mixed) state $\rho$ which is obtained by removing from $\ket{\chi}$ the summands $\Vec{P}$ whose weight is larger than $\alpha n$ and then normalizing properly, for some constant $\alpha$ to be specified shortly.

\begin{claim} Let $\Pi$ be the projector onto the hight-weight EPR-basis states in the bulk which contains at least $\alpha n$ nontrivial EPR states
\begin{equation}
    \Pi = \sum_{\Vec{P} \in \mathcal{P}^{\otimes nD}: |\Vec{P}| \geq \alpha n}   \ket{\Phi_{\Vec{P}}}\bra{\Phi_{\Vec{P}}}.
\end{equation}
If $\delta= O(D^{-0.51})$, then choosing $\alpha = 6 \varepsilon$ we can get the following bound as long as $\varepsilon = \Omega(D^{-0.01})$
    \begin{equation}
        \bra{\psi} \Pi \ket{\psi} \leq e^{-\Omega(n)}.
    \end{equation}
In other words, $\ket{\psi}$ is $e^{-\Omega(n)}$-close in fidelity to a $8\varepsilon$-noisy mixed state.
\end{claim}

\begin{proof} Note that the distribution over $\Vec{P}$ is not simply the i.i.d. distribution $\delta^{|\Vec{P}|}$ since the linear combination over the adversarial error $\Vec{E}$ in~\Cref{eq:chi} can change the norm of the state in the output column. So we need a more careful analysis.

With a slight abuse of notation, we use $\{S_1,S_2,\hdots,S_D\}$ to denote the locations of the of the original Choi states (which encode the gates) and $S_0$ to denote input column qubits (which are ancilla qubits) that correspond to the violated Hamiltonian terms (see \Cref{fig:combin}). Recall that $|S| \leq 
2 \varepsilon n$ by assumption. Let $R^c$ be the shifted EPR locations that do \textit{not} overlap with $S$ and let $R$ be the rest of shifted EPR locations. Note that $|R| \leq 2|S| \leq 4\varepsilon n$ and $|R| + |R^c| = nD$. Consider the ``partially undone'' state $\ket{\chi'}= Q_{R^{c}} \ket{\psi'}$, in which we only apply $Q$ on $R^c$, such that $\ket{\chi} = Q_{R} \ket{\chi'}$. Let $\Pi'$ be the projector onto the high-weight EPR-basis states supported on $R^c$ defined as
 \begin{equation}
     \Pi' = \sum_{\Vec{P} \in \mathcal{P}^{R^c}: |\Vec{P}| \geq (\alpha - 4\varepsilon) n}   \ket{\Phi_{\Vec{P}}}\bra{\Phi_{\Vec{P}}}.
 \end{equation}
Note that $\Pi \preceq \Pi'$ and $\delta \cdot \id \preceq Q \preceq \id$. Furthermore, $\Pi'$ and $Q$ (and $\Lambda$) commute for being both diagonal in the EPR basis. So we have that
\begin{equation}
    \bra{\psi} \Pi \ket{\psi} \leq   \bra{\psi} \Pi' \ket{\psi}  = \frac{\bra{\chi'} Q_{R} \Pi' Q_{R} \ket{\chi'}}{\bra{\chi'} Q_{R} Q_{R} \ket{\chi'}} \leq \frac{1}{\delta^{2|R|}}\frac{ \bra{\chi'} \Pi' \ket{\chi'}}{\braket{\chi'}{\chi'}}.
    \label{eq:bound-pi}
\end{equation}
Substituting $\ket{\chi'}= Q_{R^{c}} \ket{\psi'}$ into the RHS we get
\begin{equation}
    \bra{\psi} \Pi \ket{\psi} \leq \frac{1}{\delta^{2|R|}} \frac{\sum_{\Vec{P} \in \mathcal{P}^{R^c}: |\Vec{P}| \geq (\alpha - 4\varepsilon)n } \delta^{2 |\Vec{P}|} \| (\id_{R} \otimes \bra{\Phi_{\Vec{P}}}) \ket{\psi'}\|^2  }{\sum_{\Vec{P} \in \mathcal{P}^{R^c}} \delta^{2 |\Vec{P}|} \|(\id_R \otimes \bra{\Phi_{\Vec{P}}}) \ket{\psi'}\|^2 },
\end{equation}
where
$\ket{\psi'}$ is defined in \Cref{eq:psi'} and repeated here for convenience
\begin{equation}
    \ket{\psi'} = \left(\bigotimes_{i \notin S_0} \ket{0}_i  \right) \left(\bigotimes_{\loc(U) \notin S} \ket{\Phi_U}\right) \otimes \ket{\psi''}.
\end{equation}
Observe that $(\id_R \otimes \bra{\Phi_{\Vec{P}}}) \ket{\psi'}$ has the same norm for any $\vec{P} \in \mathcal{P}^{R^c}$ since $\bra{\Phi_{\Vec{P}}}$ does not act on $\ket{\psi''}$. Therefore,
\begin{equation}
    \bra{\psi} \Pi \ket{\psi} \leq  \frac{1}{\delta^{2|R|}} \frac{\sum_{\Vec{P} \in \mathcal{P}^{R^c}: |\Vec{P}| \geq (\alpha - 4\varepsilon)n } \delta^{2 |\Vec{P}|} }{\sum_{\Vec{P} \in \mathcal{P}^{R^c}} \delta^{2 |\Vec{P}|}}.
    \label{eq:RHS}
\end{equation}
The RHS can now be straightforwardly bounded by the Chernoff bound.

\begin{fact} Let $X= X_1 + \hdots + X_N $ where $X_i \in \{0,1\}$ are i.i.d. binary random variables with $\mathbb{E}[X_i]=\mu$. Then $\Pr[ X \geq (1+\eta)\mu N ] \leq 2e^{-\eta^2 \mu N/3}$.
\end{fact}

We apply the Chernoff's bound with $N=|R^c|$, $\mu = 3\delta^2/(1+3\delta^2)$, and $\eta = (\alpha-4\varepsilon)/ \mu D -1$. Note that $|R| \leq 4 \varepsilon n$. Choosing $\alpha = 6\varepsilon$, assuming $\varepsilon \ll 1$ so that $nD/2 \leq |R^c| \leq nD$ and $\delta$ is sufficiently small such that $\eta \approx 2\varepsilon/\mu D $, we obtain the following bound on the RHS of \Cref{eq:RHS}
\begin{equation}
    \text{RHS of \Cref{eq:RHS} } \leq 2 e^{8 \log(1/\delta) \varepsilon n} e ^{-\varepsilon^2 n/ \delta^2 D}.
\end{equation}

The above bound to decays exponentially when $\varepsilon > 8 \log(1/\delta) \delta^2 D$. We can choose, say, $\delta = O(D^{-0.51})$. Then for any $\varepsilon > 10 \delta \sqrt{D}$, we obtain a bound of $e^{-99n}$ on $ \bra{\psi} \Pi \ket{\psi}$.
\end{proof}

\subsection{Proof of \cref{thm:energy} (Soundness)}
\label{subsec:exposound}

\vspace{0.1in}

\noindent \textbf{Proof idea:} We take inspiration from Kitaev's analysis where the clock Hamiltonian is analyzed in a suitable rotated basis. Here as well, we will carry out the proof in a ``rotated'' basis, which is defined by the $n(2D+1)$-qubit unitary $V$ in \Cref{eq:undounit}. 
In particular, we analyze the properties of a low-energy state $\ket{\psi}$ of $H_\mathrm{parent}=H_\mathrm{in} + H_\mathrm{prop}$ and its rotated version $\ket{\psi'}=V^\dagger \ket{\psi}$. Note that in the rotated basis, the ground states (\Cref{claim:gs}) are of the form
\begin{equation}
    V^\dagger \ket{\Psi_{W,\xi}} \propto \left(\ket{\Phi_I} + \delta\sum_{p\in \{X,XZ,Z\}}  \ket{\Phi_p} \right)^{\otimes nD} \ket{0^a}\ket{\xi}.
    \label{eq:groundstate}
\end{equation}
Our starting observation is based on a surprising effect - despite the fact that $H_\mathrm{in}$ enforces $\ket{0}^{\otimes n}$ on the first column, the state $\ket{0}^{\otimes n}$ appears on the last column in $V^\dagger \ket{\Psi_{W,\xi}}$. We view this as a \textit{teleportation of $H_\mathrm{in}$}, highlighting that its a noiseless teleportation under `zero energy' constraint, despite the tensor network performing noisy gate-by-gate teleportation. Given this, we focus on establishing two properties for low energy states: 
\begin{itemize}
    \item \textit{Robust teleportation of $H_\mathrm{in}$:} Upon rotating with $V$, the low energy states should look like $\ket{0}$ in most of the qubits (that do not include witness qubits) in the last column. This amounts to $H_{\mathrm{in}}$ effectively acting on the last column under the constraint of low energy.
    \item The number of Pauli errors is small enough in a low energy state.
\end{itemize}
The proof below carries both these properties.

\vspace{0.1in}

\noindent \textbf{Proof:} First we recall the notations from~\Cref{fig:combin}. We refer to the last column of qubits in the PEPS as the \textit{output column} and note that the layers of shifted EPR locations have a correspondence with circuit layers. In particular, the first two columns as the \textit{first layer}, the next two columns as the \textit{second layer}, and so on. The unitary $V$ can be interpreted as applying a noisy circuit on the output column conditioned on the noise pattern in the bulk. Given a $n(2D+1)$-qubit PEPS state $\ket{\psi}$, we denote by $\psi^{(\ell)}_j$ the two-qubit reduced state on the $j$-th row of the $\ell$-th layer and by $\psi^\mathrm{out}_j$ the one-qubit reduced state on the $j$-th row of the output column.

The advantage of working in the rotated basis is that we can employ the following lemmas, whose proofs are provided in \Cref{sec:deferred-proofs}. 

\begin{remark} W.l.o.g., we assume the last layer of the circuit consists of single-qubit identity gates.
\end{remark}

\begin{lemma}[Last layer]
\label{lem:last-col} For each $j \in [n]$, let $h^{(D)}_{j}$ be the $3$-qubit term in $H_\mathrm{prop}$ corresponding to the identity gate on qubit $j$. Furthermore, let $\ket{\phi_0}=  \frac{1}{\sqrt{1+3\delta^2}} (\ket{\Phi_I} + \delta \sum_{p\in \{X,XZ,Z\}}  \ket{\Phi_p})$.  If $\bra{\psi} h^{(D)}_j \ket{\psi} \leq \alpha$, then it holds that 
\begin{equation}
    \Tr(\psi_j^{'(D)} \ket{\phi_0}\bra{\phi_0}) \geq 1 - 4\alpha.
\end{equation} 
\end{lemma}

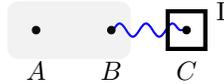
\begin{figure}[h]
    \centering
\begin{tikzpicture}[scale=0.5]
  \foreach \x in {0} {
    \foreach \y in {0} {
      \path [draw=blue,snake it, thick] (\x-0.5,\y+2.5) -- (\x+1.5,\y+2.5); 
      
      \draw[ultra thick] (\x+1,\y+2) rectangle (\x+2,\y+3) node[right] {$\id$}; 
      
      \draw [fill=black] (\x-0.5,\y+2.5) circle [radius=0.1] node[below=0.3] {$B$}; 
      
      \draw [fill=black] (\x+1.5,\y+2.5) circle [radius=0.1] node[below=0.3] {$C$}; 

      \draw [fill=black] (\x-2.5,\y+2.5) circle [radius=0.1] node[below=0.3] {$A$};  

    \draw[gray,rounded corners,fill=gray, opacity=0.1] (\x-3.25, \y+1.75) rectangle (\x,\y+3.25); 
    
    }
  }
\end{tikzpicture}
    \caption{Propagation term $h_j^{(D)} = \Lambda_{AB} (\id - \ket{\Phi_I} \bra{\Phi_I}_{BC}) \Lambda_{AB}$ corresponding to a single-qubit identity gate in the last layer. If a global state $\ket{\psi}$ has low energy with restpect to $h_j^{(D)}$, then~\Cref{lem:last-col} asserts that $V^\dagger \ket{\psi}$ is locally close to $\ket{\phi_0}$ (\Cref{eq:phi0}) on qubits $A,B$.}
    \label{fig:last-layer}
\end{figure}
In the proof of~\Cref{lem:last-col}, we will show that $V^\dagger h^{(D)}_j V$ is in fact a local Hamiltonian term, despite $V$ being global. In particular, denoting $h^{(D)}_{j} = \Lambda_{AB} (\id - \ket{\Phi_I} \bra{\Phi_I}_{BC}) \Lambda_{AB} $ where $A,B$, and $C$ denote the qubits acted upon by $h^{(D)}_{j}$ (see~\Cref{fig:last-layer}), then $V^\dagger h^{(D)}_{j} V$ is a $2$-local term acting on qubits $A,B$ in the rotated basis
\begin{equation}
    V^\dagger (h^{(D)}_{j})_{ABC} V = \Lambda_{AB} \left( \sum_{p \in \mathcal{P}} \ket{\Phi_p} \bra{\Phi_{p}}_{AB} - \frac{1}{4}\sum_{p,p' \in \mathcal{P}} \ket{\Phi_p} \bra{\Phi_{p'}}_{AB} \right) \Lambda_{AB}.
\end{equation}

The rotated propagation terms $V^\dagger h^{(\ell)}_{i,j} V$ for $\ell < D$ (corresponding to two-qubit gates acting on 
 qubits $i,j$ in layer $\ell$) are, however, generally non-local\footnote{We show in~\Cref{appendix:Clifford} that they are local if the associated gate is Clifford.}. Here, we will instead utilize the following lemma about a property of them in a certain subspace related to the state
 \begin{equation}
     \ket{\phi_0} \triangleq  \frac{1}{\sqrt{1+3\delta^2}} (\ket{\Phi_I} + \delta \sum_{p\in \{X,XZ,Z\}}  \ket{\Phi_p}).
     \label{eq:phi0}
 \end{equation}

\begin{lemma}[Bulk propagation] Consider a propagation term $h^{(\ell)}_{i,j}=  \Lambda^{\otimes 4} (\id - \ket{\Phi_U}\bra{\Phi_U} ) \Lambda^{\otimes 4} $, where $\ket{\Phi_U}$ is the Choi state encoding the two-qubit gate $U$ acting on qubits $i,j$ in layer $\ell < D$. It holds that
\begin{equation}
    \bra{\phi_0}^{\otimes 2} V^\dagger h^{(\ell)}_{i,j} V \ket{\phi_0}^{\otimes 2}  = \frac{16 \delta^4}{(1+3\delta^2)^2} \Lambda^{\otimes 2} \left( \sum_{\Vec{p} \in \mathcal{P}^{\otimes 2}} \ket{\Phi_{\Vec{p}}} \bra{\Phi_{\Vec{p}}}  - \frac{1}{16} \sum_{\Vec{p}, \Vec{q} \in \mathcal{P}^{\otimes 2}} \ket{\Phi_{\Vec{p}}} \bra{\Phi_{\Vec{q}}} \right)\Lambda^{\otimes 2},
\end{equation}
where $\ket{\phi_0}^{\otimes 2}$ acts on the shifted EPR locations $(i,j)$ in EPR layer $\ell+1$. Furthermore, a robust version of the previous statement also holds. Let $\ket{\psi}$ be a normalized state such that $\Tr( \psi'^{(\ell+1)}_{i,j}  \phi_0^{\otimes 2}) \geq 1 - \eta$ in  for some $\ell \leq D-1$. If additionally $\bra{\psi} h^{(\ell)}_{i,j} \ket{\psi} \leq \alpha$, then $\Tr ( \psi'^{(\ell,\ell+1)}_{i,j} \phi_0^{\otimes 4}) \geq 1 - \frac{\eta}{\delta^8} - \frac{\alpha}{\delta^{16}}$.
\label{lem:bulk-prop}
\end{lemma}
Intuitively, the above lemma says that, if the qubits $i,j$ in layer $\ell+1$ of a slighly violated propagation Hamiltonian term are in the ``good'' state $\phi_0$, then the qubits $i,j$ in layer $\ell$ are also in the good state $\phi_0$.

\begin{figure}[h]
    \centering
\begin{tikzpicture}[scale=0.5]
  \foreach \x in {0} {
    \foreach \y in {0, 2} {
      \path [draw=blue,snake it, thick] (\x-0.5,\y+2.5) -- (\x+1.5,\y+2.5); 

      \draw [fill=black] (\x-0.5,\y+2.5) circle [radius=0.1]; 
      
      \draw [fill=black] (\x+1.5,\y+2.5) circle [radius=0.1];

      \draw [fill=black] (\x-2.5,\y+2.5) circle [radius=0.1] ;  
      
      \draw [fill=black] (\x+3.5,\y+2.5) circle [radius=0.1] ; 

    \draw[gray,rounded corners,fill=gray, opacity=0.1] (\x-3.25, \y+1.75) rectangle (\x,\y+3.25); 
    
      \draw[gray,rounded corners,fill=gray, opacity=0.1] (\x+0.75,\y+1.75) rectangle (\x+4,\y+3.25); 
    }
    
  }

      \draw[ultra thick] (1,2) rectangle (2,5) node[right] {$U$}; 

\end{tikzpicture}
    \caption{Propagation term $h_U$ corresponding to a two-qubit gate $U$ in the bulk of the circuit. According to~\Cref{lem:bulk-prop}, if a global state $\ket{\psi}$ has low energy with respect to $h_U$ and its rotated version $V^\dagger \ket{\psi}$ is locally close to $\ket{\phi_0}^{\otimes 2}$ on the EPR locations to the right, then $V^\dagger \ket{\psi}$ is close to $\ket{\phi_0}^{\otimes 4}$ on all 4 EPR  locations.}
    \label{fig:bulk-prop}
\end{figure}
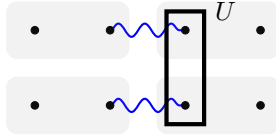

\begin{lemma}[Robust teleportation of $H_\mathrm{in}$]
\label{lem:H-teleport}
Consider an initialization term $h^\mathrm{in}_j=  \Lambda \Pi_j \Lambda $ in $H_\mathrm{in}$, where $\Pi_j$ is the input check on qubit $j$. It holds that $ \bra{\phi_0} V^\dagger h^\mathrm{in}_j V \ket{\phi_0}  = \frac{4\delta^2}{1+3\delta^2}  \Pi^\mathrm{out}_j$, where $\ket{\phi_0}$ acts on the first layer at EPR location $j$, and $\Pi^\mathrm{out}_j$ means $\Pi_j$ acts on the output column. Furthermore, a robust version of the previous statement also holds. Let $\ket{\psi}$ be a normalized state and $\ket{\psi'}$ be its rotated version, such that $\Tr( \psi'^{(1)}_{j} \ket{\phi_0}\bra{\phi_0}) \geq 1 - \eta$. If additionally $\bra{\psi'} V^\dagger h^\mathrm{in}_j V \ket{\psi'} \leq \alpha$, then $\Tr ( \psi' (\id - \Pi_j^\mathrm{out})) \geq 1 - \frac{\eta}{\delta^2} - \frac{\alpha}{\delta^{2}}$.
\end{lemma}

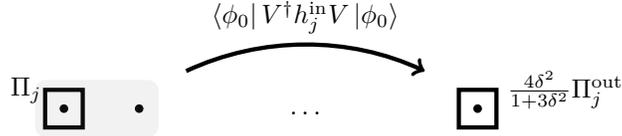
\begin{figure}[h]
    \centering
    \begin{tikzpicture}[scale=0.5]
          \draw[ultra thick] (-3,+2) rectangle (-2, 3) node[left=0.4] {$\Pi_j$}; 
          
          \draw [fill=black] (-0.5,+2.5) circle [radius=0.1]; 
    
          \draw [fill=black] (-2.5,+2.5) circle [radius=0.1];  
          
        \draw[gray,rounded corners,fill=gray, opacity=0.1] (-3.25, 1.75) rectangle (0,3.25); 
        
        \draw[ultra thick, ->] (0.75, 3.5) arc (115:65:7.5) node[midway, above=0] {$\bra{\phi_0} V^\dagger h^\mathrm{in}_j V \ket{\phi_0}$} node[midway, below=0.75] {$\hdots$};
        
        \draw[ultra thick] (8, 2) rectangle (9, 3) node[right=0] {$\frac{4\delta^2}{1+3\delta^2}\Pi_j^\mathrm{out}$}; 
        
        \draw [fill=black] (8.5,+2.5) circle [radius=0.1];    

    \end{tikzpicture}
    \caption{Initialization term $h^\mathrm{in}_j=  \Lambda \Pi_j \Lambda $ in $H_\mathrm{in}$ is teleported to the output column according to~\Cref{lem:H-teleport}.}
    \label{fig:enter-label}
\end{figure}

Therefore,~\Cref{lem:H-teleport} says that if the first layer is in the good state $\ket{\phi_0}$, then $H_\mathrm{in}$ is teleported to the output column.

We now proceed to prove \Cref{thm:energy}, here we assume depth $D= O(\log n)$ with a sufficiently small constant factor. 

\begin{proof}[Proof of \Cref{thm:energy}]
Consider a state $\ket{\psi}$ with energy density $\frac{\varepsilon}{D+1} $ such that $\bra{\psi} H_\mathrm{parent} \ket{\psi} \leq \varepsilon n$, where $\varepsilon= \delta^{200D}$. It follows that at most $\varepsilon n/\alpha $ terms in $H_\mathrm{parent}$ have energy greater than $\alpha$, for some value $\alpha$ to be specified later (we will choose $\alpha=\delta^{100D}$). We refer to these terms as ``strongly violated'' and the terms with energy smaller than $\alpha$ as ``slightly violated'' .
Let $\ket{\psi'}=V^\dagger \ket{\psi}$ be the rotated state.

We overview the proof. First, according to~\Cref{lem:last-col} $\psi'^{(D)}_j$ is close to $\ket{\phi_0}$ for many indices $j \in [n]$. Next, we repeatedly apply \Cref{lem:bulk-prop} to propagate the ``good'' states $\ket{\phi_0}$ to the first layer. Then we use \Cref{lem:H-teleport} to conclude that most of intialization terms in $H_\mathrm{in}$ get teleported (approximately) to the output column. This makes sure that most of the ancilla qubits are initialized (approximately) correctly to $\ket{0}$. Finally, we use a similar argument to \Cref{thm:combin}'s proof to combine the stochastic errors with the adversarial errors.

We now provide the proof details. We start by looking at the local reduced states on the last layer. For at least $(1-\varepsilon/\alpha)n$ many indices $j$, the propagation term $h^{(D)}_j$ in the last layer is slightly violated. So we invoke~\Cref{lem:last-col} to obtain $\Tr(\psi_j^{'(D)} \ket{\phi_0}\bra{\phi_0}) \geq 1 - 4\alpha$.

We ``propagate'' these good states from the last layer to the first layer. According to~\Cref{lem:bulk-prop}, a sufficient condition for $\psi^{'(1)}_j$ to be good is that all the Hamiltonian terms associated to gates in the forward lightcone of qubit $j$, have energy bounded by $\alpha$. We denote the forward lightcone of qubit $j$ by $\operatorname{LC}(j)$. Note that we only consider $D=O(\log n)$. Assume all propagation terms in $\operatorname{LC}(j)$ are slightly violated, then for any locations $r,s \in \operatorname{LC}(j)$ in the last layer we have $\Tr(\psi_{r,s}^{'(D)} (\ket{\phi_0}\bra{\phi_0}^{\otimes 2})) \geq 1 - 8\alpha$ due to the previous paragraph. Next, we repeatedly apply~\Cref{lem:bulk-prop} on the propagation terms in $\operatorname{LC}(j)$ to obtain that $\Tr(\psi^{'(1)}_j \ket{\phi_0}\bra{\phi_0}) \geq 1- \frac{\alpha}{\delta^{16D}}$ for sufficiently small $\delta$.
Thus, we can invoke the robust version of~\Cref{lem:H-teleport} to obtain $\Tr(\psi' \Pi_j^\mathrm{out}) \geq 1 - \frac{\alpha}{\delta^{16D+2}}$.

The number of locations whose forward lightcone is ``bad'' (i.e., it contains a strongly violated Hamiltonian term) is bounded above by $2^{D-1}\varepsilon n/\alpha $. Hence, $\Tr(\psi' \Pi_j^\mathrm{out}) \geq 1 - \frac{\alpha}{\delta^{16D+2}}$ for at least a fraction of $ 1- 2^{D-1}\varepsilon /\alpha$ of the qubits $j$. We refer to the initialization locations without this guarantee as ``strongly faulty'' initialization locations. Denote these locations as $S_0$, we have $|S_0| \leq 2^{D-1} \varepsilon n/\alpha$. Similarly, there are at most $2^{D-1}\varepsilon n/\alpha$ EPR locations where we do not have the guarantee $\Tr(\psi'^{(\ell)}_j \ket{\phi_0}\bra{\phi_0}) \geq 1- \frac{\alpha}{\delta^{16D}}$. 
We refer to them as ``strongly faulty'' gate locations in the circuit and denote $S=\{S_1,\hdots,S_D\}$. 

For each slightly faulty location $j$ at layer $1\leq \ell \leq D$, we have the following distance guarantee due to Fuchs--van de Graf inequality and by choosing, say, $\alpha = \delta^{50D}$
\begin{equation}
     \| \psi_{j}^{'(\ell)} - \ket{\phi_0}\bra{\phi_0} \|_1 \leq 2 \sqrt{\frac{\alpha}{\delta^{16D}}} \leq  \delta^{10D}.
\end{equation}
It follows that
\begin{align}
      \left| \Tr (\psi_{j}^{'(\ell)} \ket{\Phi_p}\bra{\Phi_p}) - |\bra{\phi_0} \ket{\Phi_p}|^2 \right| \leq    \delta^{10D}, \qquad p \in \mathcal{P}
      \label{eq:marginal1}
\end{align}
In other words, the Pauli errors at the slightly violated locations approximately follows the depolarizing channel with probability $p=\delta^2/(1+3\delta^2)$ for each of $X,Y,Z$ errors.

Similarly, we obtain the following bound at the (approximately) correctly initialized locations in the output column
\begin{align}
    \Tr (\psi_{j}^{'\mathrm{out}} \ket{0}\bra{0}) \geq    1- \delta^{10D}, 
          \label{eq:marginal2}
\end{align}

With $\alpha = \delta^{50D}$ and $\varepsilon=\delta^{200D}$ we also have the following bound on the total number of strongly faulty locations
\begin{equation}
    |S| = \sum_{\ell=0}^{D} |S_\ell| \leq 2^D \varepsilon n/\alpha  \leq \delta^{50D} n.
\end{equation}

Similar to the proof of~\Cref{thm:combin}, we denote by $\Tilde{E}_\ell$ the adversarial errors at locations $S_\ell$ coming from the strongly faulty locations and by $\Tilde{P}_\ell$ be the \textit{almost} local depolarizing noise coming from the slightly faulty locations.
We can expand $\ket{\psi'}$ as
\begin{equation}
    \ket{\psi'} = \sum_{\Vec{E}  \in \mathcal{P}^{S} \atop \Vec{P} \in \mathcal{P}^{S^c}} c_{\Vec{E},\Vec{P}}  \bigotimes_{1\leq \ell \leq D}\left(\ket{\Phi_{\Tilde{P}_\ell}} 
    \ket{\Phi_{\Tilde{E}_\ell}} \right)\bigotimes (\Tilde{E}_0 \otimes \Tilde{P}_0 \ket{0^a}) \otimes \ket{\xi_{\Vec{P},\Vec{E}}},
\end{equation}
where $\ket{\xi_{\Vec{P},\Vec{E}}}$ are normalized states and $c_{\Vec{E},\Vec{P}}$  are (w.l.o.g.) real coefficients such that $|c_{\Vec{E},\Vec{P}}|^2$ define a probability distribution whose local marginals on $S^c$ are constrained by Equations~\eqref{eq:marginal1}, \eqref{eq:marginal2}.

Finally, we can combine $\Vec{P}$ and $\Vec{E}$ together and treat them as adversarial errors by truncating the summands with  high-weight $\Vec{P}$. Let $\Pi$ be the projector onto  high-weight EPR states in $S^c$
 \begin{equation}
     \Pi = \sum_{\Vec{P} \in \mathcal{P}^{S^c}: |\Vec{P}| \geq (\beta-\delta^{50D}) n}   \ket{\Phi_{\Vec{P}}}\bra{\Phi_{\Vec{P}}},
 \end{equation}
where $\beta$ is a parameter to be specified. Below, we will truncate the high-weight $\Vec{P}$ in $\ket{\psi'}$ to obtain the state $\ket{\chi'} = \frac{(\id-\Pi) \ket{\psi} }{\|(\id-\Pi) \ket{\psi}\|}$, and our goal is to show $\ket{\chi'}$ is close to $\ket{\psi'}$. Note that $\ket{\chi'}$ only contains terms with at most $\beta n$ adversarial errors as desired.

Observe that $\mathbb{E}_{\Vec{P} \sim \ket{\psi'}}[|\Vec{P}|: \Vec{P} \in \mathcal{P}^{S^c}] \leq 3(\frac{\delta^2}{1+3\delta^2} + \delta^{10D})nD$ according to Equations~\eqref{eq:marginal1}, \eqref{eq:marginal2}. So using Markov's inequality we can bound
\begin{equation}
    \Tr(\Pi \ket{\psi'}\bra{\psi'}) \leq \frac{4 \delta^2 D}{\beta}.
\end{equation}
Assuming $\delta^2 D \ll 1$ and chosing $\beta = 400 \delta^2 D$ and using gentle measurement lemma we have
\begin{equation}
    \frac{1}{2}\|\ket{\chi'}\bra{\chi'} - \ket{\psi'}\bra{\psi'} \|_1 \leq \frac{1}{10}.
\end{equation}
The same trace distance bound holds on the unrotated states $V\ket{\psi'}$ and $V \ket{\chi'} \triangleq \ket{\chi}$, as well as their reduced states on the output column $\psi_\mathrm{out}$ and $\chi_\mathrm{out}$ due to unitary-invariance and monoticity of the trace distance:
\begin{equation}
    \frac{1}{2}\|\psi_\mathrm{out} - \chi_\mathrm{out}\| \leq \frac{1}{10}.
\end{equation}
The reduced state $\chi_\mathrm{out}$ is a $\beta$-noisy state
\begin{equation}
    \chi_\mathrm{out} = \sum_{\tilde{E}_1\,\hdots \tilde{E}_D: |\Vec{E}| \leq \beta n }  \left( \sum_{\tilde{E}_0: |\Vec{E}| \leq \beta n } c'_{\Vec{E}} \widetilde{W}_{\Vec{E}} (\Tilde{E}_0 \ket{0^a}) \otimes \ket{\xi_{\Vec{E}}} \right) \left( \sum_{\tilde{E}_0: |\Vec{E}| \leq \beta n } c'_{\Vec{E}} \widetilde{W}_{\Vec{E}} (\Tilde{E}_0 \ket{0^a}) \otimes \ket{\xi_{\Vec{E}}} \right) ^\dagger,
\end{equation}
where $\widetilde{W}_{\Vec{E}} \triangleq W_D \Tilde{E}_D\hdots W_1 \Tilde{E}_1 $. This concludes the proof of~\Cref{thm:energy}.
\end{proof}

\subsection{Analysis of $H_\mathrm{prop}$ and $H_\mathrm{in}$: deferred proofs}\label{sec:deferred-proofs}
For convenience, recall the change of basis 
\begin{equation}
    V = \sum_{\Vec{P} \in \mathcal{P}^{\otimes nD}} \ket{\Phi_{\Vec{P}}} \bra{\Phi_{\Vec{P}} } \otimes  (W_{D} \tilde{P}_{D} \hdots W_1 \tilde{P}_{1}),
\end{equation}
and the inverse local injective map
\begin{equation}
    \Lambda = \delta \ket{\Phi_I} \bra{\Phi_I} + \sum_{p \in \{X, XZ,Z\}}  \ket{\Phi_p} \bra{\Phi_p}.
\end{equation}

\subsubsection{Proof of~\Cref{lem:last-col} (Good states in last layer of $H_\mathrm{prop}$)}
As stated in the lemma, we assume the last layer of gates in the circuit are identity gates for simplicity in calculating the rotated $H_g$ terms.  We have that
\begin{align*}
    V^\dagger (h^{(D)}_{j})_{ABC} V &= V^\dagger \Lambda_{AB} (\id - \ket{\Phi_I} \bra{\Phi_I}_{BC}) \Lambda_{AB} V \\
    & = V^\dagger \left(\sum_{p,p' \in \mathcal{P}} \delta^{2-|(p,p')|} \ket{\Phi_p} \bra{\Phi_{p'}}_{AB} \otimes  \bra{\Phi_p}_{AB}(\id - \ket{\Phi_I} \bra{\Phi_I}_{BC}) \ket{\Phi_{p'}}_{AB} \right) V \\
    & = V^\dagger \left(\sum_{p,p' \in \mathcal{P}} \delta^{2-|(p,p')|} \ket{\Phi_p}\bra{\Phi_{p'}} \left( \mathds{1}_{p,p'} \id - \frac{1}{4}(p^* p'^\top)_C  \right) \right) V \\
    & =\sum_{p \in \mathcal{P}} \delta^{2-2 |p|} \ket{\Phi_p} \bra{\Phi_{p}}_{AB} \\
    & -\frac{1}{4} \sum_{\Vec{P} \in \mathcal{P}^{\otimes n(D-1)}} \sum_{p,p' \in \mathcal{P}} \delta^{2-|(p,p')|} \ket{\Phi_p} \bra{\Phi_{p'}}_{AB} \otimes \ket{\Phi_{\Vec{P}}} \bra{\Phi_{\Vec{P}}} \otimes \left((\widetilde{W}_{\Vec{P}}^{<D})^\dagger p^\dagger p^* p'^\top p' (\widetilde{W}_{\Vec{P}}^{<D}) \right) \\
    & = \sum_{p \in \mathcal{P}} \delta^{2-2|p|} \ket{\Phi_p} \bra{\Phi_{p}}_{AB} - \frac{1}{4}\sum_{p,p'} \delta^{|(p,p')|} \ket{\Phi_p} \bra{\Phi_{p'}}_{AB}\\
    & = \Lambda_{AB} \left( \sum_{p} \ket{\Phi_p} \bra{\Phi_{p}}_{AB} - \frac{1}{4}\sum_{p,p'} \ket{\Phi_p} \bra{\Phi_{p'}}_{AB} \right) \Lambda_{AB},
\end{align*}
where
$\sum_{\Vec{P} \in \mathcal{P}^{\otimes n(D-1)}}$ denotes the sum over the Pauli noise $\tilde{P}_{\ell}$ for $\ell < D$ and $\widetilde{W}_{\Vec{P}}^{<D} \triangleq W_{D-1} \Tilde{P}_{D-1} \hdots W_1 \Tilde{P}_1$.

Note that $V^\dagger h^{(D)}_{j} V$ has ground state $\ket{\phi_0}$ and spectral gap $\gamma \geq 1/4$, so
\begin{align*}
    \frac{1}{4} \Tr(\psi_j^{'(D)} (\id - \ket{\phi_0}\bra{\phi_0}) ) &\leq \Tr (\ket{\psi}\bra{\psi} h_j^{(D)}) \leq \alpha, \\
     \Tr(\psi_j^{'(D)} \ket{\phi_0}\bra{\phi_0}) &\geq 1 - 4\alpha.
\end{align*}

\subsubsection{Proof of~\Cref{lem:bulk-prop} (Bulk propagation of good states)}

We first prove the following claim, which is the ``noiseless'' version of \Cref{lem:bulk-prop}.

\begin{claim}
\label{claim:projector} Let $\ket{\phi_0}$ be defined as in~\Cref{eq:phi0}. Consider a propagation term $h_U$ corresponding to a two-qubit gate as shown in~\Cref{fig:bulk-prop}. Furthermore, define $\ket{\Phi_{\Vec{p}}} \triangleq \ket{\Phi_{p_1}} \ket{\Phi_{p_2}}$, for $\Vec{p} \in \mathcal{P}^{\otimes 2}$. It holds that
\begin{align}
    \bra{\phi_0}^{\otimes 2}  V^\dagger h_U V \ket{\phi_0}^{\otimes 2} & = \frac{16 \delta^4}{(1+3\delta^2)^2} \Lambda^{\otimes 2} \left( \sum_{\Vec{p} \in \mathcal{P}^{\otimes 2}} \ket{\Phi_{\Vec{p}}} \bra{\Phi_{\Vec{p}}}  - \frac{1}{16} \sum_{\Vec{p}, \Vec{q} \in \mathcal{P}^{\otimes 2}} \ket{\Phi_{\Vec{p}}} \bra{\Phi_{\Vec{q}}} \right) \Lambda^{\otimes 2},
    \label{eq:27463}
\end{align}
where $\ket{\phi_0}^{\otimes 2}$ ($\ket{\Phi_{\Vec{p}}}$) acts on the EPR locations to the right (left) of $h_U$ (see~\Cref{fig:bulk-prop}).
\end{claim}
The above claim implies that if $\psi$ fully satisfies $h_U$ and on one side of $h_U$ the state $V^\dagger \ket{\psi}$ is equal to $\ket{\phi_0}^{\otimes 2}$, then so it is on the other side since this is the unique ground state of the matrix to the RHS of~\Cref{eq:27463}.

\begin{proof}[Proof of Claim]
For simplicity, we prove the claim for Hamiltonian propagation terms corresponding to single-qubit gates. The generalization to the multi-qubit case is straightforward as explained later. Consider a gate $U$ in layer $\ell < D$ acting on qubit $j$. We refer to the Hamiltonian term corresponding to this gate as $h_U=\Lambda^{\otimes 2}_{AB,CD}(\id - \ket{\Phi_U} \bra{\Phi_U}_{BC})\Lambda^{\otimes 2}_{AB,CD}$, which acts on qubits $A,B$ (EPR layer $\ell$) and  $C,D$ (EPR layer $\ell+1$).

\begin{center}
\begin{tikzpicture}[scale=0.5]
  \foreach \x in {0} {
    \foreach \y in {0} {
      \path [draw=blue,snake it, thick] (\x-0.5,\y+2.5) -- (\x+1.5,\y+2.5); 
      
      \draw[ultra thick] (\x+1,\y+2) rectangle (\x+2,\y+3) node[right] {$U$}; 
      
      \draw [fill=black] (\x-0.5,\y+2.5) circle [radius=0.1] node[below=0.3] {$B$}; 
      
      \draw [fill=black] (\x+1.5,\y+2.5) circle [radius=0.1] node[below=0.3] {$C$};

      \draw [fill=black] (\x-2.5,\y+2.5) circle [radius=0.1] node[below=0.3] {$A$};  
      
      \draw [fill=black] (\x+3.5,\y+2.5) circle [radius=0.1] node[below=0.3] {$D$}; 

    \draw[gray,rounded corners,fill=gray, opacity=0.1] (\x-3.25, \y+1.75) rectangle (\x,\y+3.25); 
    
      \draw[gray,rounded corners,fill=gray, opacity=0.1] (\x+0.75,\y+1.75) rectangle (\x+4,\y+3.25); 
    }
  }
\end{tikzpicture}
\end{center}

Let $\ket{\Phi_{\Vec{p}}}_{AB,CD} \triangleq \ket{\Phi_{p_1}}_{AB}\ket{\Phi_{p_2}}_{CD}$, for $p_1,p_2 \in \mathcal{P}$. We will omit the system labels when they are clear from the context. The rotated term is of the form
\begin{equation*}
\begin{aligned}
    V^\dagger h_U V &= V^\dagger \Lambda^{\otimes 2}_{AB,CD} (\id - \ket{\Phi_U} \bra{\Phi_U}_{BC})  \Lambda^{\otimes 2}_{AB,CD} V \\
    & = V^\dagger \left( \sum_{\Vec{p},\Vec{q} \in \mathcal{P}^{\otimes 2}} \delta^{4-|(\Vec{p}, \Vec{q})|} \ket{\Phi_{\Vec{p}}} \bra{\Phi_{\Vec{q}}}_{AB,CD} \cdot  \bra{\Phi_{\Vec{p}}}(\id - \ket{\Phi_U} \bra{\Phi_U}_{BC}) \ket{\Phi_{\Vec{q}}} \right) V\\
    &= V^\dagger \left( \sum_{\Vec{p},\Vec{q} \in \mathcal{P}^{\otimes 2}} \delta^{4-|(\Vec{p}, \Vec{q})|} \ket{\Phi_{\Vec{p}}} \bra{\Phi_{\Vec{q}}}_{AB,CD} \cdot \left(\mathds{1}_{\Vec{p},\Vec{q}} - \frac{1}{8} \operatorname{Tr}(p_1^* q_1^\top U^\dagger q_2^\top p_2^* U ) \right) \right) V\\
    &= \sum_{\Vec{p} \in \mathcal{P}^{\otimes 2}} \delta^{4-2|\Vec{p}|} \ket{\Phi_{\Vec{p}}} \bra{\Phi_{\Vec{p}}} \\
    &- \frac{1}{8} \sum_{\Vec{P}\in \mathcal{P}^{\otimes n(\ell-1)}}  \sum_{\Vec{p},\Vec{q}} \delta^{4-|(\Vec{p}, \Vec{q})|} \ket{\Phi_{\Vec{p}}} \bra{\Phi_{\Vec{q}}}  \operatorname{Tr}(p_1^* q_1^\top U^\dagger q_2^\top p_2^* U ) \otimes \ket{\Phi_{\Vec{P}}} \bra{\Phi_{\Vec{P}}} \otimes (\widetilde{W}_{\Vec{P}}^{< \ell })^\dagger p_1^\dagger U^\dagger  p_2^\dagger q_2 U q_1 (\widetilde{W}_{\Vec{P}}^{< \ell }),
\end{aligned}
\end{equation*}
where $\widetilde{W}_{\Vec{P}}^{<\ell}\triangleq W_{\ell-1} \Tilde{P}_{\ell-1} \hdots W_1 \Tilde{P}_1$. Above, $\mathds{1}_{\Vec{p},\Vec{q}}$ denotes the Kronecker delta symbol. The sum $\sum_{\Vec{P}\in \mathcal{P}^{\otimes n(\ell-1)}} $ is over the Pauli noise $\Vec{P}$ in layers preceding the gate $U$. We can also drop the complex conjugate ``$*$'' because $\mathcal{P}=\{I,X,XZ,Z\}$ are real matrices.

Next, we project qubits $C,D$ onto $\ket{\phi_0}$. Doing so on the term $\sum_{\Vec{p}} \delta^{4-2|\Vec{p}|} \ket{\Phi_{\Vec{p}}} \bra{\Phi_{\Vec{p}}}_{AB,CD}$ in $V^\dagger h_U V$ yields the following two-qubit term acting on qubits $A,B$
\begin{equation}
    \frac{4\delta^2}{1+3\delta^2}\sum_{p_1 \in \mathcal{P}} \delta^{2-2|p_1|} \ket{\Phi_{p_1}} \bra{\Phi_{p_1}}_{AB}.
    \label{eq:7226}
\end{equation}
We analyze the second term in $V^\dagger h_U V$. For each summand $\Vec{P}$, projecting project qubits $C,D$ onto $\ket{\phi_0}$ gives
\begin{equation}
    \frac{\delta^2}{1+3\delta^2} \sum_{\Vec{p},\Vec{q}\in \mathcal{P}^{\otimes 2}} \delta^{2-|(p_1,q_1)|} \ket{\Phi_{p_1}} \bra{\Phi_{q_1}}  \operatorname{Tr}(p_1 q_1^\top U^\dagger q_2^\top p_2 U ) \otimes \ket{\Phi_{\Vec{P}}} \bra{\Phi_{\Vec{P}}}  \otimes (\widetilde{W}_{\Vec{P}}^{< \ell })^{\dagger} p_1^\top U^\dagger  p_2^\top q_2 U q_1 (\widetilde{W}_{\Vec{P}}^{< \ell }).
    \label{eq:462}
\end{equation}
Next, we apply the following identity
\begin{equation}
    \sum_{p_2,q_2 \in \mathcal{P}} \operatorname{Tr}(p_1 q_1^\top U^\dagger q_2^\top p_2 U ) p_2^\top q_2 = 8 U p_1 q_1^\top U^\dagger
\end{equation}
to simplify~\Cref{eq:462} to
\begin{equation}
    \frac{8\delta^2}{1+3\delta^2} \sum_{p_1,q_1} \delta^{2-|(p_1,q_1)|} \ket{\Phi_{p_1}} \bra{\Phi_{q_1}}  \otimes \ket{\Phi_{\Vec{P}}} \bra{\Phi_{\Vec{P}}}  \otimes \id.
\end{equation}
Overall, summing over $\Vec{P}$, the second term in $V^\dagger h_U V$ is equal to
\begin{equation}
     -\frac{\delta^2}{1+3\delta^2} \sum_{p_1,q_1} \delta^{2-|(p_1,q_1)|} \ket{\Phi_{p_1}} \bra{\Phi_{q_1}}.
\end{equation}
Combining this with~\Cref{eq:7226} we get
\begin{equation}
    \bra{\phi_0}_{CD} V^\dagger h_U V \ket{\phi_0}_{CD} = \frac{4\delta^2}{1+3\delta^2} \Lambda_{AB} \left( \sum_{p_1} \ket{\Phi_{p_1}} \bra{\Phi_{p_1}}_{AB} - \frac{1}{4}\sum_{p_1,q_1} \ket{\Phi_{p_1}} \bra{\Phi_{q_1}}_{AB} \right) \Lambda_{AB}.
\end{equation}
A completely similar analysis for two-qubit gates gives the lemma  statement.
\end{proof}

We now prove \Cref{lem:bulk-prop}.

\begin{proof}[Proof of~\Cref{lem:bulk-prop}]
    
Let $\Pi_1 = \id \otimes \phi_0^{\otimes 2}$ and $\Pi_2$ be the projector onto the ground space of $V^\dagger h_U V$. As a reminder, it is assumed that $\Tr(\Pi_1 \psi' ) \geq 1 - \eta $ and $\Tr(V^\dagger h_U V \psi') \leq \alpha$, and the goal is to show $\Tr(\Pi_2 \psi') \geq 1 - \frac{\eta}{\delta^8} - \frac{\alpha}{\delta^{16}}$.

According to~\Cref{claim:projector}, the operator $\bra{\phi_0}^{\otimes 2}  V^\dagger h_U V \ket{\phi_0}^{\otimes 2} $ has a spectral gap $\geq 15 \delta^8$ for sufficiently small $\delta$. Therefore,
\begin{equation}
    15 \delta^8 (\Pi_1 - \phi_0^{\otimes 4}) \leq \Pi_1 V^\dagger h_U V \Pi_1.
\end{equation}
However, observe the following inequality which follows from $\|h_U\| \leq 1$
\begin{equation}
 \Pi_1 V^\dagger h_U V \Pi_1 \leq (\Pi_1 - \Pi_1 \Pi_2 \Pi_1)
\end{equation}
Combining the previous inequalities we obtain
\begin{equation}
    \Pi_1 \Pi_2 \Pi_1 \leq  \Pi_1 - 15 \delta^8 (\Pi_1 - \phi_0^{\otimes 4}).
    \label{eq:jordan}
\end{equation}

Next, we apply Jordan's lemma to decompose $\Pi_1$ and $\Pi_2$ into $1\times 1$ and $2 \times 2$ blocks. Observe that \Cref{claim:projector} implies $\phi_0^{\otimes 4}$ is the \textit{unique} intersection of $\Pi_1$ and $\Pi_2$, as also evident from~\Cref{eq:jordan}. Consider two corresponding $2 \times 2$ blocks $\ket{u}\bra{u}$ in $\Pi_1$ and $\ket{v}\bra{u}$ in $\Pi_2$,~\Cref{eq:jordan} then implies that $|\braket{u}{v}|^2 \leq 1- 15 \delta^8$.

On the other hand, letting $\gamma \geq \delta^{8}$ be the spectral gap of $h_U$\footnote{We have $h_U^2 = \Lambda^{\otimes 4}(\id - \ket{\Phi_U}\bra{\Phi_U}) (\Lambda^2)^{\otimes 4} (\id - \ket{\Phi_U}\bra{\Phi_U})   \Lambda^{\otimes 4} \geq \delta^8 \Lambda^{\otimes 4}(\id - \ket{\Phi_U}\bra{\Phi_U}) \Lambda^{\otimes 4} = \delta^8 h_U.$}, we have
\begin{equation}
    \gamma  (\id - \Pi_2)  \leq V^\dagger h_U V  \Longrightarrow \Tr(\Pi_2 \psi') \geq 1 -  \frac{\alpha}{\delta^8}.
\end{equation}
The following expressions follows by writing the projectors $\Pi_1, \Pi_2$ according to Jordan's lemma $\Pi_1 = \phi_0^{\otimes 4} + \sum_{i} \ket{u_i}\bra{u_i}$ and $\Pi_2 = \phi_0^{\otimes 4} + \sum_{i} \ket{v_i}\bra{v_i}$
\begin{align}
    \Tr(\psi' \phi_0^{\otimes 4}) + \sum_{i} \Tr(\ket{u_i}\bra{u_i}  \psi') &\geq 1- \eta,\\
    \Tr(\psi' \phi_0^{\otimes 4}) + \sum_{i} \Tr(\ket{v_i}\bra{v_i}  \psi') &\geq 1- \frac{\alpha}{\delta^8}. 
\end{align}
Using $\ket{u_i}\bra{u_i} + \ket{v_i}\bra{v_i} \leq (1+|\braket{u_i}{v_i}|)P_i \leq  (2-15 \delta^8 ) P_i$, where $P_i$ is the projector onto the $2\times 2$ Jordan block $i$, we get
\begin{align}
    2\Tr(\psi' \phi_0^{\otimes 4}) + (2-15 \delta^8) \Tr(\sum_{i} P_i \psi' ) \geq 2 -\eta - \frac{\alpha}{\delta^8}
\end{align}
Using $\sum_i P_i + \phi_0^{\otimes 4} \leq  \id$ and rearranging we get
\begin{equation}
    \Tr(\psi' \phi_0^{\otimes 4}) \geq 1 - \frac{\eta}{15 \delta^8} - \frac{\alpha}{15 \delta^{16}}\geq 1 - \frac{\eta}{\delta^8} - \frac{\alpha}{\delta^{16}},
\end{equation}
This concludes the proof of~\Cref{lem:bulk-prop}.
\end{proof}

Let us mention that a generalized statement of \Cref{lem:bulk-prop} for $k$-local gates also holds.
\begin{lemma}
Consider a propagation term $h^{(\ell)}_{S}=  \Lambda^{\otimes 2k} (\id - \ket{\Phi_U}\bra{\Phi_U} ) \Lambda^{\otimes 2k} $, where $\ket{\Phi_U}$ is the EPR state encoding the $k$-qubit gate $U$ acting on qubit set $S$ of size $|S|=k$ in layer $\ell < D$. Let $\ket{\psi}$ be a normalized state and $\ket{\psi'}$ be its rotated version, such that $\Tr( \psi'^{(\ell+1)}_{S}  \phi_0^{\otimes k}) \geq 1 - \eta$. If additionally $\bra{\psi} h^{(\ell)}_{S} \ket{\psi} \leq \alpha$, then $\Tr ( \psi'^{(\ell,\ell+1)}_{S} \phi_0^{\otimes 2k}) \geq 1 - \frac{\eta}{\delta^{4k}} - \frac{\alpha}{\delta^{8k}}$.\footnote{More precisely, the overlap lower bound can be found to be $\Tr ( \psi'^{(\ell,\ell+1)}_{S} \phi_0^{\otimes 2k}) \geq 1 - \left(\frac{\eta}{\delta^{4k}} + \frac{\alpha}{\delta^{8k}}\right) \frac{(1+3\delta^2)^k}{4^k}$. But the RHS is at least $1 - \left(\frac{\eta}{\delta^{4k}} + \frac{\alpha}{\delta^{8k}}\right)$ when $\delta$ is sufficiently small.}
\label{lem:prop-k-local}
\end{lemma}

 Its proof follows the same line as in the 1-qubit and 2-qubit cases above, so we will not spell out the entire proof. In short, we first derive the rotated Hamiltonian term $V^\dagger h_U V$ when all $k$ EPR locations to the right of gate $U$ are projected to $\ket{\phi_0}^{\otimes k}$, generalizing \Cref{claim:projector}.
\begin{align}
    \bra{\phi_0}^{\otimes k}  V^\dagger h_U V \ket{\phi_0}^{\otimes k} & = \frac{4^k \delta^{2k}}{(1+3\delta^2)^k} \Lambda^{\otimes k} \left( \sum_{\Vec{p} \in \mathcal{P}^{\otimes k}} \ket{\Phi_{\Vec{p}}} \bra{\Phi_{\Vec{p}}}  - \frac{1}{4^k} \sum_{\Vec{p}, \Vec{q} \in \mathcal{P}^{\otimes k}} \ket{\Phi_{\Vec{p}}} \bra{\Phi_{\Vec{q}}} \right) \Lambda^{\otimes k}.
\end{align}
Then we can follow the derivation using Jordan's lemma in the preceding page to derive \Cref{lem:prop-k-local}.
 
\subsubsection{Proof of~\Cref{lem:H-teleport} (Teleportation of $H_\mathrm{in}$ to output column)}

We have the ``noiseless'' version
\begin{align}
     \bra{\phi_0} V^\dagger h^\mathrm{in}_j V \ket{\phi_0}& =\bra{\phi_0} \left( \sum_{p,p' \in \mathcal{P}} \delta^{2-|(p,p')|} \ket{\Phi_p} \bra{\Phi_{p'}} \otimes (p^\dagger p')_j^{\mathrm{out}}  \bra{\Phi_p} (\Pi_j \otimes \id) \ket{\Phi_{p'}}
    \right)\ket{\phi_0}\\
     & =\bra{\phi_0} \left( \sum_{p,p' \in \mathcal{P}} \delta^{2-|(p,p')|} \ket{\Phi_p} \bra{\Phi_{p'}} \otimes (p^\dagger p')_j^{\mathrm{out}}  \frac{1}{2}\Tr(p'^\top p^*  \Pi_j )
 \right)\ket{\phi_0}\\
 & = \frac{1}{1+3\delta^2} \sum_{p,p' \in \mathcal{P}} \delta^{2} (p^\dagger p')_j^{\mathrm{out}}  \frac{1}{2}\Tr(p'^\top p^*  \Pi_j ) \\
& = \frac{4 \delta^2}{1+3\delta^2} \Pi_j^\mathrm{out}.
\end{align}
The proof of the robust version is completely similar to that of~\Cref{lem:bulk-prop}.

A generalization of \Cref{lem:H-teleport} to $k$-local input check terms also hold, whose proof we do not spell out, but is entirely similar to the one-qubit case.
\begin{lemma}
    \label{lem:Hteleport-k-local}
Consider an initialization term $h^\mathrm{in}_S =  \Lambda \Pi_S \Lambda $ in $H_\mathrm{in}$, where $\Pi_S$ is the input check on an input qubit set $S$ of size $|S|=k$. Let $\ket{\psi}$ be a normalized state and $\ket{\psi'}$ be its rotated version, such that $\Tr( \psi'^{(1)}_{S} \ket{\phi_0}\bra{\phi_0}) \geq 1 - \eta$. If additionally $\bra{\psi'} V^\dagger h^\mathrm{in}_S V \ket{\psi'} \leq \alpha$, then $\Tr ( \psi' (\id - \Pi_S^\mathrm{out})) \geq 1 - \frac{\eta}{\delta^{2k}} - \frac{\alpha}{\delta^{2k}}$.
\end{lemma}
This lemma is useful, for example, if we want to impose that the witness provided by the prover is a codeword of stabilizer code with low check weight, e.g., a quantum LDPC code.

\subsection{Spectral gap lowerbound}\label{subsec:spectral-gap}
The previous subsections characterize states with sufficiently low energy as encoding a noisy execution of the quantum circuit. We can also consider other basic properties of the parent Hamiltonian such as spectral gap. Here we give a lower bound on the spectral gap, which will be used to give a new proof of QMA-completeness of the local Hamiltonian problem in later sections.

\begin{theorem} Suppose that all gates in $H_\mathrm{prop}$ and input check terms in $H_\mathrm{in}$ have locality $k$. Then the spectral gap of $H_\mathrm{parent}$ is lower bounded by $\gamma = \delta^{8k(D+1)}/\poly(nD)$. More generally, if we vary the injectity parameter $\delta$ and gate locality in the PEPS, then the factor $\delta^{8kD}$ is replaced by the product of the injectivity parameters across the depth of the circuit, i.e., $\gamma= \frac{1}{\poly(nD)} \prod_{\ell=0}^D \delta_\ell^{8k_\ell}$, where $\delta_\ell$ and $k_\ell$ are the injectivity and gate locality in layer $\ell$ ($\ell=0$ corresponds to initializations).
\label{thm:spectral-gap}
\end{theorem}

\begin{proof}[Proof of \Cref{thm:spectral-gap}]
        We first consider the case $k=2$ and $\delta$ is fixed throughout the circuit. The proof is essentially a simpler version of \Cref{thm:energy}'s proof in \Cref{subsec:exposound}. The difference is that, there, we wanted to work with states of energy density $\delta^{16(D+1)}$ even when $D$ is constant; whereas now we are willing to consider states with energy density $\delta^{16(D+1)}/\poly(nD)$ , which scales as $1/\poly(n)$ even when $D$ is a constant. Hence we will be brief here.
        
        Consider any state $\ket{\psi}$ with energy $E= \bra{\psi} H_\mathrm{parent} \ket{\psi}<\delta^{16(D+1)}/\poly(nD) $. Our goal is to show that it has large overlap with the ground space -- a fact that will immediately imply a spectral gap lower bound.
        As before,
        let $\ket{\psi'}=V^\dagger \ket{\psi}$ be the rotated state.

        The proof proceeds as before. We start from~\Cref{lem:last-col} to show that $\ket{\psi'}$ has large overlap with $\ket{\phi_0}$ on all locations in the circuit's last layer,
        and then repeatedly applying~\Cref{lem:bulk-prop} (bulk propagation lemma) we find that $\Tr(\psi_j^{'(\ell)} \ket{\phi_0}\bra{\phi_0}) \geq 1 - E\delta^{-16 D} \geq 1 - \delta^{16}/\poly(nD)$ for \textit{any} location indexed by $j \in [n]$ and layer $\ell \in [D]$.
        
        Next, we invoke~\Cref{lem:H-teleport} (Hamiltonian teleportation lemma) to obtain $\Tr(\psi' (\id - \Pi_\mathrm{a}^\mathrm{out})) \geq 1 - 1/\poly(nD)$, which guarantees (approximately) correct initialization of ancilla qubits specified by the projectors $\Pi_a \in H_\mathrm{in}$. E.g., if ancilla qubits are intialized to $\ket{0}$, we use $\Pi_{a}=\ket{1}\bra{1}_a$, or if we want to verify the input is encoded in a stabilzer quantum code, then we use $\Pi_a = \frac{1}{2}(\id- S_a)$. Since all the projectors $ (\phi_0)_j^{(\ell)}$,  $\Pi_a^\mathrm{out}$ are commuting, we find via the union bound that
        \begin{align}
            \Tr(\psi' \underbrace{\phi_{0}^{\otimes nD} \prod_{a \in H_\mathrm{in}} (\id- \Pi_a^\mathrm{out})}_{\Pi_\mathrm{g.s.}}) \geq 1 - 1/\poly(nD).
        \end{align}
        Note that the polynomial in $\poly(nD)$ is different from before (it is smaller by a factor proportional to $nD$ and the number of terms in $H_\mathrm{in}$).
        In other words, any state $\ket{\psi'}= V^\dagger \ket{\psi}$ with energy below $\delta^{16(D+1)}/\poly(nD)$ has large overlap with the ground space of $V^{\dagger} H_\mathrm{parent} V$. This implies the stated spectral gap for $H_\mathrm{parent}$.

        The statement for $H_\mathrm{parent}$ with varying injectivity parameter and gate locality follows by straightforwardly generalizing from the above proof by using, layer by layer,~\Cref{lem:prop-k-local} and~\Cref{lem:Hteleport-k-local}, the generalized $k$-local versions of~\Cref{lem:bulk-prop} and~\Cref{lem:H-teleport}.
\end{proof}

\subsection{Adversarial fault tolerance against inverse-polynomial adversarial noise}\label{sec:advFT}

Here we note that a repetition argument similar to \cite{Chen20} (see~\Cref{append:QMApoly}) also suffices to protect a computation against an inverse-polynomial fraction of adversarial noise for any desired polynomial, at the cost of increasing the circuit size by a corresponding polynomial factor.

\vspace{0.1in}

\noindent{\bf Classical case:}

Given a circuit $C$ on $n$ bits with $T$ gates, lets us run the circuit in parallel $k$ times, for $k$ to be chosen shortly. Let $C_1,C_2, \ldots C_k$ be these runs of $C$. The repeated circuit has $kT$ gates. For a $\delta>0$, we would like to protect against $(kT)^{1-\delta}$ adversarial errors. Note that even if there was at most 1 error per $C_i$, the number of circuits with no error is $k-(kT)^{1-\delta}=k\left(1-\frac{(T)^{1-\delta}}{k^{\delta}}\right)$. Choosing $k= (100 T)^{\frac{1}{\delta}}$, we see that at least $0.99k$ circuits have no error and the output of the computation can be read by considering the majority value. 

In fact, we do not need to do fault tolerant majority computation. We simply put a Hamiltonian $H_{out} = \sum_i \ketbra{1}_i$ on the $k$ output bits. This Hamiltonian penalizes if most of the outputs are $1$. Further, note that for any constant $\delta$, this is a polynomial sized transformation.

\vspace{0.1in}

\noindent{\bf Quantum case:}

Identical argument works in the quantum case if the adversarial error does not occur in superposition and the quantum circuit computes the correct outcome with probability $0.9$. This happens in the case of combinatorial soundness, where the error locations are fixed. It is far from clear if general superposition over low weight errors can be handled. But at the same time, the low energy states may not admit an arbitrary superposition over errors. We leave this understanding for the future work.

\section{Verifying $\QMA$ via shallow circuits}
\label{sec:shallow_circuits}
As shown in~\Cref{sec:soundness}, the parent Hamiltonian robustness properties only depend on circuit depth, so it is desirable to restrict our attention to shallow circuits. Here we show that any QMA protocol can be replaced by one involving a constant depth quantum circuit followed a logarithmic depth classical circuit. The high-level idea is to first use the Feynman-Kitaev mapping to turn an arbitrary $\QMA$ protocol into a local Hamiltonian, and then construct a short-depth $\QMA$ circuit to measure the energy of the resulting Hamiltonian. For this, we need a low-degree version of the FK mapping.

\begin{claim}[Degree reduction for FK Hamiltonian]
\label{claim:modifiedFK} Any $\QMA$ protocol involving an $n$-qubit verifier circuit $V$ with $T=\poly(n)$ two-qubit gates can be mapped into a $5$-$\LH[a,b]$ on $\poly(n)$ qubits with $a=2^{-\poly(n)}$ and $b = a + 1/\poly(n)$. Furthermore, each qubit is involved in at most $7$ terms in the Hamiltonian.
\end{claim}
\begin{proof}
    W.l.o.g., we assume the circuit has been amplified by~\Cref{lem:amplify} or~\Cref{lem:strongamplify}, such that its completeness is $c=1-2^{-r}$ and $s = 2^{-r}$ with $r= \poly(n)$.
    
     We first recall the FK Hamiltonian~\cite{kitaev2002classical} here to observe that it is not sparse. For all $T\in \mathbb{N}$ and $t\leq T$, we define the unary clock states as $|u(t,T)\rangle= |1^{t}\rangle \otimes |0^{T-t}\rangle$. The clock qubits are index by $t \in [T]$ and the data qubits are indexed by $i \in [n]$. Let $m$ be the number of ancilla qubits, so that the witness has $n-m$ qubits. The FK Hamiltonian consists of four parts acting on a unary clock register and a data register: (1) initialization terms 
     \begin{align*}
         H_\mathrm{in} = \ket{0}\bra{0}_{t=0}\otimes \left(\sum_{i=1}^{m} \ket{1}\bra{1}_i \right),
         \label{eq:Hinit}
     \end{align*}
     (2) propagation terms (note there are no clock qubits $-1$ and $T+1$)
     \begin{align*}
         H_\mathrm{prop} = \frac{1}{2} \left(\sum_{t=1}^{T} (\ket{100}\bra{100} + \ket{110}\bra{110})_{t-1,t,t+1}
         - \ket{110}\bra{100_{t-1,t,t+1}}\otimes U_t - \ket{100}\bra{110}_{t-1,t,t+1} \otimes U_t^\dagger \right),
     \end{align*}
     (3) clock validity terms 
     \begin{align*}
         H_\mathrm{clock} = \sum_{t=1}^T \ket{01}\bra{01}_{t-1,t},
     \end{align*}
     (4) and output check term
     \begin{align*}
         H_\mathrm{out} = \ket{1}\bra{1}_T \otimes \ket{0}\bra{0}_1.
     \end{align*}
     As it is, the FK Hamiltonian has high degree due to the $t=0$ clock qubit, which participates in $m$ terms in $H_\mathrm{in}$ and the data qubits, which participate in as many terms in $H_\mathrm{prop}$ as the number of nontrivial gates acting on the qubit. 

     \vspace{0.1in}
     
     \begin{remark}
     Here, we justify the claim in Table \ref{table:clock}
     that  there is a combinatorial state that violates a $O(\frac{1}{T})$ fraction of terms. For example, the state $\ket{0100\ldots}_{\mathrm{clock}}\otimes \ket{0}^{\otimes n}_{\mathrm{data}}$ contains a fixed invalid clock configuration and hence satisfies all the terms in the Feynman-Kitaev Hamiltonian, except $2$ terms from $H_{\text{clock}}$.
     \label{remark:FKfail}
     \end{remark}

    \vspace{0.1in}

     We reduce the degree of data qubits by transforming $V$ into a new circuit $V'$ that acts on $n'=n T$ qubits divided into $T$ $n$-qubit blocks. After applying the first gate in $V$ on the first qubit block, we apply $n$ SWAP gates to swap the first and second blocks. Then, the second gate in $V$ is applied on the second block of $V'$, and so on. This way, the qubits in $V'$ are acted on by at most $3$ nontrivial gates. The number of nontrivial gates in $V'$ is $T'=O(nT)$.

     We reduce the degree of the $t=0$ clock qubit by observing that the initialization of ancilla qubit $i$ only need to be verified right before the first gate acting on it. Let  $t_i \in [T]$ be this gate, then we use the following initialization term (note there are no clock qubits $-1$ and $T+1$)
     \begin{equation}
         H_{\mathrm{in},i} = \ket{10}\bra{10}_{t_i-1,t_i} \otimes \ket{1}\bra{1}_i.
         \label{eq:modifiedHinit}
     \end{equation}
     
     Applying this modified FK mapping (with modified $H_\mathrm{in}$) to the circuit $V'$ we obtain a 5-local Hamiltonian $H_\mathrm{FK}$ in which each qubit involves in at most $7$ terms.
    The energy analysis in~\cite{kitaev2002classical} still applies for this modified construction. Indeed, according to~\cite{kitaev2002classical}, $H_\mathrm{prop} + H_\mathrm{clock}$ has ground states of the form
     \begin{equation}
         |\Psi\rangle:=\frac{1}{\sqrt{T'+1}}\sum_{t=0}^{T'}|u(t,T')\rangle\otimes U_t\cdots U_1 \ket{\psi}, \text{ for any } \ket{\psi} \in (\mathbb{C}^{2})^{n'}
     \end{equation}
     
     In the completeness case, setting $\ket{\psi}=|0^{m'}\rangle \ket{\xi}$,
     where $\ket{\xi}$ is the witness that $V'$ accepts with probability $c$,
     gives an energy of $a = O((1-c)/T')= 2^{-\poly(n')}$. 
     
     In the soundness case,
     the main step of the proof is Equation 14.17 in~\cite{kitaev2002classical} in which the author bounds $\max_{\ket{\psi}} \bra{\Psi} \Pi_1 \ket{\Psi}$ where $\Pi_1$ is the projector onto the nullspace of $H_\mathrm{in} +H_\mathrm{out}$. However, it can be seen that modifying $H_\mathrm{in}$ as in \Cref{eq:modifiedHinit} does not change this quantity which remains to be
    \begin{equation}
        \bra{\Psi} \Pi_1 \ket{\Psi} = 1 - \frac{1}{T'+1} \left( \bra{\psi}( \sum_{i=1}^{m'} \ket{1}\bra{1}_i ) \ket{\psi} + \bra{\psi} V'^{\dagger}\ket{0}\bra{0}_1 V' \ket{\psi} \right)
    \end{equation}
     by noting that $U_1^\dagger \hdots U_{t_i-1}^\dagger (\ket{1}\bra{1}_i) U_{t_i-1} \hdots U_1 = \ket{1}\bra{1}_i$ for any $i$. Therefore, according to~\cite{kitaev2002classical}, any state has energy no smaller than $b = \Omega((1-\sqrt{s})/T'^3)= 1/\poly(n')$. 
\end{proof}

\begin{claim}[Log-depth $\QMA$] Any $\QMA$ protocol involving an $n$-qubit verifier circuit $V$ with $T=\poly(n)$ two-qubit gates can be converted into a $O(\log n)$-depth $\QMA$ protocol on $\poly(n)$ qubits, whose completeness is $1-2^{-r}$ and soundness is $2^{-r}$ with $r=\poly(n)$. More specifically, the $O(\log n)$-depth circuit involves a constant-depth quantum circuit that ends with computational basis measurements, followed by a $O(\log n)$-depth classical circuit.
\end{claim}
\begin{proof} 
Given any $\QMA$ protocol $V_0$ (with $n_0$ qubits including the size of the witness and number of gates $T_0=\poly(n_0)$), we first convert it into the low-degree FK Hamiltonian using~\Cref{claim:modifiedFK}. The Hamiltonian $H_\mathrm{FK}=  \sum_{i=1}^{m} h_i$ acts on $n=\Theta(n_0 T_0)$ qubits, contains $m= \Theta(n_0 T_0) =\Theta(n)$ projectors that are at most $5$-local, and has a promise gap of $b-a=\Omega(m^{-3})$. Each qubit participates in at most $7$ terms $h_i$.

Next, we construct a constant-depth circuit $V$ extracting the satisfiability of the Hamiltonian terms in $H_\mathrm{FK}$. The circuit consists of $m$ ancillas initialized to $\ket{0}$. Upon receiving an $n$-qubit witness state $\ket{\xi}$, $V$ applies unitaries of the form $\mathrm{C_{h_i}NOT}=(\id  - h_i) \otimes I_i +  h_i \otimes X_i$, which, conditioned on the reduced state of $\ket{\xi}$ being in $\mathrm{supp}(h_i)$, flip ancilla qubit $i$. In particular, consider the decomposition of the terms $h_i$ into $L=O(1)$ groups, $H_1,\hdots,H_L$ such that the terms in each group are pairwise non-overlapping. Let $\Pi_{\ell} = \bigotimes_{i: h_i \in H_\ell} (\id - h_i) $. The layer $\ell \in [L]$ of $V$ is $V_\ell = \prod_{i: h_i \in H_\ell} \mathrm{C_{h_i}NOT}$. After applying $V_Q \triangleq V_L \hdots V_{2}V_1$, we measure the ancillas in the $Z$ basis to get a bitstring $x \in \{0,1\}^{m}$, and compute the OR function on $x$ and output $\overline{\mathrm{OR}(x)}$. The OR function on $m$ bits can be computed by a $O(\log m)$-depth Boolean circuit\footnote{A log-depth OR circuit is as follows: the first layer computes pairwise OR's $(x_1 \vee x_2)$, $(x_3 \vee x_4)$, $\hdots$, the second layer similarly computes OR pairwise on the output of the first layer, and so on.}, which can in turn be made reversible with constant space overhead~\cite[Section 3.2.5]{nielsen2000quantum}.

The circuit $V$ output $1$ if and only if the ancillas are measured in the all-zeros string, which happens with probability $\Pr[x=0^m]= \left| \bra{0^m} V_Q  \ket{\xi} \otimes \ket{0^m} \right|^2=  \Tr(\mathrm{DL}^\dagger\mathrm{DL} \ket{\xi} \bra{\xi} )$ where $\mathrm{DL} \triangleq \Pi_{L} \hdots \Pi_{2}\Pi_1$ is the detectability lemma operator~\cite{anshu2016simple}.

Below we show that, if $V_0$ accepts, then $V$ accepts $1-2^{-\poly(n)}$ and if $V_0$ rejects then $V$ accepts with $1-\Omega(n^{-2})$. In addition, the soundness can be depth-efficiently improved to $2^{-\poly(n)}$.

\paragraph{Completeness} The prover sends the witness state $\ket{\xi}$ such that $\bra{\xi} H_\mathrm{FK} \ket{\xi} \leq 2^{-\poly(n)}$ (the case of mixed state witness follows by linear extension). Also $\bra{\xi}h_i \ket{\xi} \leq 2^{-\poly(n)}$ for any $i \in [m]$. Using the quantum union bound (\Cref{lem:quantum-union}) on $ H_\mathrm{FK}$ and $\ket{\xi}$ we can bound
\begin{equation}
    1- \Tr(\mathrm{DL}^\dagger\mathrm{DL} \ket{\xi} \bra{\xi} ) \leq 4 \sum_{i} \bra{\xi} h_i \ket{\xi} \leq 2^{-\poly(n)}.
\end{equation}
So $V$ outputs $1$ with probability at least $c = 1- 2^{-\poly(n)}$.

\paragraph{Soundness} According to~\Cref{claim:modifiedFK}, for any state $\ket{\xi}$ we have $\bra{\xi} H_\mathrm{FK} \ket{\xi} \geq  \Omega(n^{-3})$. Observe that the terms in $H_\mathrm{FK}$ are projectors and each of them overlaps with at most $g=34$ others, so we can apply the detectability lemma (\Cref{lem:DL}) on $H_\mathrm{FK}$ and $\ket{\xi}$
\begin{equation}
     \Tr(\mathrm{DL}^\dagger\mathrm{DL} \ket{\xi} \bra{\xi} ) \leq \frac{1}{ \Omega(n^{-3}) +1} \leq 1 - \Omega(n^{-3}) = s.
\end{equation}

Finally, soundness can be amplified to $2^{-\poly(n)}$ while keeping the depth logarithmic via the weak amplification procedure in~\Cref{lem:amplify}.
In particular, this procedure~\cite{kitaev2002classical} works by using $q=\poly(n)$ copies of $V$ in parallel. The prover is expected to send $q$ copies of an accepting state. We perform OR on the $q$ decision bits of the copies in depth $O(\log q)=O(\log n)$. It is a standard fact that we can assume w.l.o.g. the prover sends an unentangled state between these $q$ copies (e.g., see~\cite[Lemma 14.1]{kitaev2002classical}). Thus, a simple application of Chernoff's bound achieves the amplified soundness whenever $q$ is a sufficiently large polynomial in $n$.

\end{proof}

We note that the technique in~\cite{rosgen2007distinguishing}, where the author studies the $\QIP$-hardness of distinguishing log-depth circuits, also gives a log-depth verification procedure for $\QMA$, see~\Cref{app:logdepth}. However, their quantum circuit is necessarily logarithmic-depth due to the use of $n$-qubit controlled SWAP gates. This is to be compared with our log-depth construction, where the quantum circuit is constant-depth and followed by log-depth classical circuit. This could be a useful feature for fault tolerance protocols and possible implications for the quantum PCP conjecture that we discuss in this work. The work \cite{JiWu09} also gives a constant depth circuit followed by the threshold majority computation, which can be used as a verification protocol in $\QMA$.

\section{Proof of $\QMA$-hardness of local Hamiltonian problem}\label{sec:qma}
The goal of this section is to apply our circuit-to-Hamiltonian mapping to give a new proof for the $\QMA$-hardness of the local Hamiltonian problem, from first principles and without using the Feynman-Kitaev clock construction. The proof will be somewhat similar to Kitaev's original proof~\cite{kitaev2002classical}: we first lower bound the spectral gap of the circuit Hamiltonian (without output Hamiltonian check terms), and then convert this spectral gap into the $\QMA$ promise gap when adding the output Hamiltonian check terms. However, the extra component in our proof is to construct a fault-tolerant version of the $\QMA$ circuit such that the spectral gap is lower bounded.

\subsection{Applying the construction to $\QMA$ circuits}
We apply the parent Hamiltonian construction to convert any $\QMA$ protocol $W$ into a local Hamiltonian as follows. First, we consider the fault-tolerant version of the verifier cicuit (against local iid depolarizing noise), denoted $W_\FT$ on $n_\FT$ qubits and with depth $D_\FT$. The fault-tolerant verifier expects an encoded witness $\ket{\xi}_L$ from the prover encoded in a quantum error-correcting code. Before running $W_\FT$ on the witness, it may run a tester circuit, $W_\mathrm{test}$, e.g., stabilizer measurements, to verify that $\ket{\xi}_L$ is a valid codeword. Alternatively, if the QECC has low-weight checks, we can instead introduce a code Hamiltonian into the parent Hamiltonian construction and omit this test of the verifier. Specifically, let $\{S_j\}_j$
be the QECC stabilizer generators, then we may add the following Hamiltonian into $H_\mathrm{parent}$
\begin{align}
    H_\mathrm{stab} = \sum_{\text{stabilizer } j} \Lambda^{\otimes N(j)} \frac{1-S_j}{2} \Lambda^{\otimes N(j)},
\end{align}
where $N(j)$ is the set of EPRs in the PEPS's first layer that have intersecting support with the stabilizer check $S_j$.

We also introduce a frustration-free Hamiltonian $H_\mathrm{out}=\sum_{i=1}^p h^\mathrm{out}_i $ acting on the last column of qubits in order to penalize the output bit(s) of the $\QMA$ circuit, where we expect every term in $H_\mathrm{out}$ is satisfied. We assume the terms $h^\mathrm{out}_i$ are commuting projectors. The locality and number of terms in $H_\mathrm{out}$ may vary and could depend on the fault tolerance scheme being used.
For example, we could logically decode and store the $\QMA$ answer qubit into a physical qubit at the end of $W_\FT$, so that we can use a single $1$-local Hamiltonian term $H_\mathrm{out} = \ket{0}\bra{0}^\mathrm{out}$ on the PEPS's output column. 
But later we will also consider the case $p=\poly(n)$ and $h_{i}^\mathrm{out}$ is $O(\log n)$-local. The reason for this generality in our definition of $H_\mathrm{out}$ (as compared to the single-term $H_\mathrm{out}$ in Kitaev's proof) is because the depth $D_{\FT}$ of the fault-tolerant QMA circuit will turn out to determine the promise gap. In particular, we will want to keep the depth $D_{\FT}$ small, which is usually a challenge in quantum fault tolerance. On the other hand, the $\QMA$ verification circuit used in our proof will involve a majority computation on $p=\poly(n)$ binary measurement outcomes, which can be offloaded to $H_\mathrm{out}$ to save the circuit depth.

The total Hamiltonian is
\begin{align}
    H_\mathrm{total} = \underbrace{H_\mathrm{in} + H_\mathrm{prop} +  \text{(optionally $H_\mathrm{stab}$)}}_{H_\mathrm{parent}} + \mathbf{C} H_\mathrm{out},
    \label{eq:Htotal}
\end{align}
where $\mathbf{C}$ is an energy scaling factor (which will be taken to be proportional to the spectral gap of $H_\mathrm{parent}$). The origin for this scaling factor will be explained later in the proof of~\Cref{thm:qma}.
Note that $H_\mathrm{total}$ has implicit dependence on the injectivity parameters, which we can vary across the PEPS.

The following claim asserts how spectral gap of the parent Hamiltonian and fault tolerance can be combined to prove $\QMA$-hardness.

\begin{claim}[$\QMA$-hardness from a fault-tolerant verifier]
\label{thm:qma} Consider a (noiseless) $\QMA$ verifier circuit $W$ where $p$ commuting binary projective measurements are performed at the end, such that either (completeness) there exists a witness input state such that all of the $p$ measurements return $1$ with probability at least $c$, or (soundness) for any witness state, with probability at least $1-s$, one of the measurements returns $0$.
Suppose there exists a fault-tolerant version $W_\FT$ such that its PEPS parent Hamiltonian ($H_\mathrm{parent}$ in~\Cref{eq:Htotal}) has spectral gap at least $\gamma$ and the probability of a logical error in the PEPS ground state is at most $\delta_L$, such that $\max\{1-c, \delta_L\} < \frac{1-s}{16p^2}$. Let $H_\mathrm{total} = H_\mathrm{parent}+ \gamma H_\mathrm{out}$, where $H_\mathrm{out}$ has $p$ terms corresponding to the $p$ output measurements in $W_\FT$. Then, the following holds:
\begin{itemize}
    \item In the completeness case, $H_\mathrm{total}$ has an eigenvalue smaller than $\frac{\gamma(1-s)}{8p}$.
    \item In the soundness case, all eigenvalues of $H_\mathrm{total}$ are at least $\frac{\gamma(1-s)}{4p}$.
\end{itemize}
In other words, determining the ground energy of $H_\mathrm{total}$ to precision $\frac{\gamma(1-s)}{8p}$ is $\QMA$-hard.
\end{claim}

\begin{proof}
    In the completeness case, we consider the state $\ket{\Psi_{W_\FT, \xi_L}}$, a normalized PEPS ground state state of $H_\mathrm{parent}$ with an accepting encoded witness $\xi_L$. This state has zero energy on $H_\mathrm{parent}$. Furthermore, due to fault tolerance, all of the output measurements return 1 with probability at least $c-\delta_L $. Thus, the energy of $H_\mathrm{out}$ on $\ket{\Psi_{W_\FT, \xi_L}}$ is at most $p \gamma (1- c + \delta_L)  \leq \gamma(1-s)/8p$.

    In the soundness case, we first lower bound the angle between the ground spaces of $H_\mathrm{out}$ and $H_\mathrm{parent}$. Let $\ket{\Psi_{W_\FT, \xi_L}}$ be any ground state of $H_\mathrm{parent}$ and $\theta$ be the angle such that $\cos \theta =  \| \Pi_{\mathrm{g.s.}(H_\mathrm{out})} \ket{\Psi_{W_\FT, \xi_L}}\|$. By the fault tolerance of $W_\FT$ and the soundness $s$ of $W$ we have that
    \begin{align}
        \Tr(H_\mathrm{out} \ket{\Psi_{W_\FT, \xi_L}}\bra{\Psi_{W_\FT, \xi_L}}) \geq (1-s-\delta_L) \geq \frac{1}{2} (1- s),
    \end{align}
    where the last inequality holds whenever $s < 1/3$ and $ \delta_L < 1/3$.
    
    On the other hand, we have $\Tr(H_\mathrm{out} \ket{\Psi_{W_\FT, \xi_L}}\bra{\Psi_{W_\FT, \xi_L}}) \leq (\sin \theta )^2 \|H_\mathrm{out}\| \leq p(\sin \theta )^2$. Thus, $1 - \cos \theta \geq \frac{1-s}{4p}$. Applying the geometric lemma (\Cref{lem:geometric}) we find that
    \begin{align}
        H_\mathrm{total} = H_\mathrm{parent} + \gamma H_\mathrm{out} \geq  \frac{1-s}{4p} \gamma .
    \end{align}
    We note that the condition $1-c < \frac{1-s}{16p^2}$ is a simple technical requirement that can be satisfied by preprocessing the original (noiseless) circuit $W$ by $\QMA$ amplification, while the condition $\delta_L < \frac{1-s}{16p^2}$ just means that the noise needs to be sufficiently suppressed. Finally, the reason for the energy scaling factor $\gamma$ in front of $H_\mathrm{out}$ is because even in the completeness case, the logical error rate $\delta_L$ could create a large penalty on $H_\mathrm{out}$, so we set the energy scale of $\gamma$ to bring this penalty below the spectral gap.
\end{proof}

\paragraph{Fault tolerance depth overhead and parent Hamiltonian spectral gap.}
To apply~\Cref{thm:qma} to prove the $\QMA$-hardness of the local Hamiltonian problem, we need to devise a fault-tolerant $\QMA$ circuit and lower bound the spectral gap $\gamma$ of the associated parent Hamiltonian. For a circuit $W_\FT$ composed of 2-qubit gates, of depth $D_\FT$, on $n_\FT=\poly(n)$ qubits,~\Cref{thm:spectral-gap} states that $\gamma \geq \delta^{16D_\FT}/\poly(n_\FT, D_\FT,\delta^{-1})$, which is exponentially small in the depth. So it is desirable to have $D_\FT=O(\log n_\FT)=O(\log n)$. To start, we saw in~\Cref{sec:shallow_circuits} that the original (noiseless) verifier circuit $W$ can indeed be assumed to have depth $D=O(\log n)$. This was proved by using a low-degree interacting variant of the Feynman-Kitaev Hamiltonian to convert any $\QMA$ protocol into an equivalent protocol in which the verifier consists of a constant-depth quantum circuit followed by a log-depth classical circuit. However, we will avoid this circuit in this section because our goal here is to give a $\QMA$-completeness proof that is independent of the Feynman-Kitaev mapping. Instead, we observe another way to make $\QMA$ protocols log-depth, using the technique in Rosgen's work~\cite{rosgen2007distinguishing} that proved $\QIP$-completeness of the log-depth circuit distinguishability problem.

\begin{claim}[Log-depth $\QMA$ using SWAP tests]
    Any $\QMA$ protocol can be turned into a logarithmic depth circuit as shown in~\Cref{fig:claim6.3}, with completeness $c=1 - 2^{-\poly (n)}$ and soundness $s=1-1/n^\alpha$, for a positive constant $\alpha \leq 3$.
    \label{claim:rosgen}
\end{claim}
\begin{figure}
    \centering
    \includegraphics[width=0.99\textwidth]{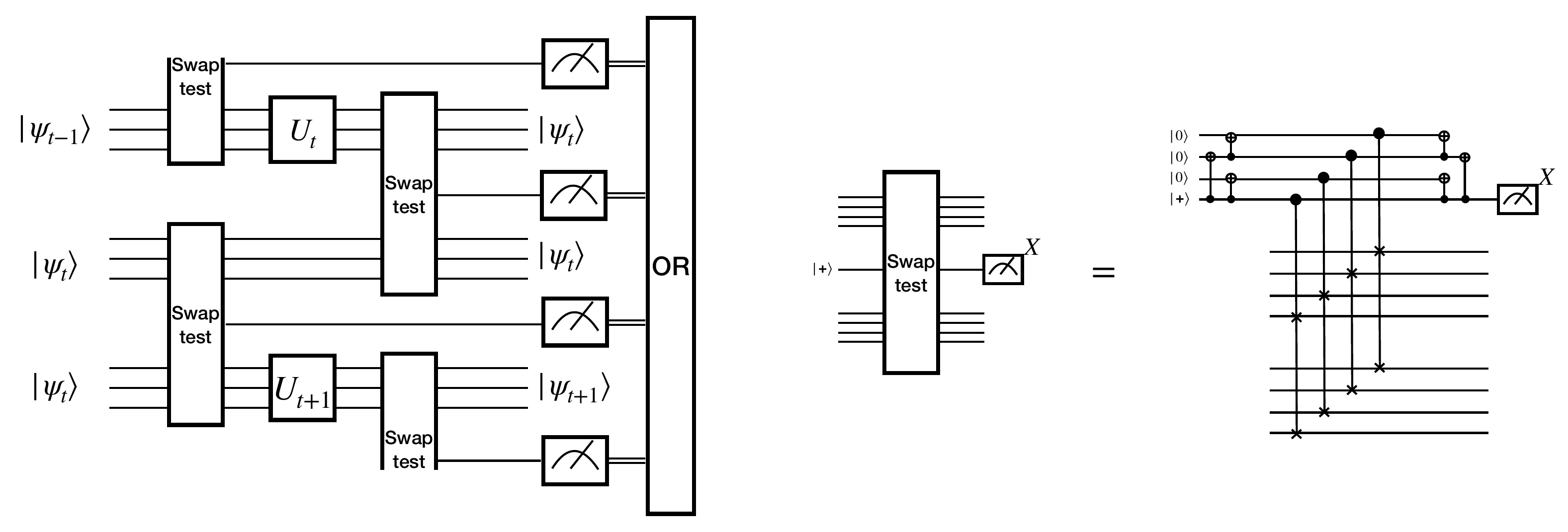}
    \caption{\textit{Left}: Parallelized $\QMA$ protocol with SWAP tests in \Cref{claim:rosgen}. The verifier expects all measurements to output 1. \textit{Right}: Log-depth implementation of SWAP test.}
    \label{fig:claim6.3}
\end{figure}

We summarize here the idea of this circuit and refer to \Cref{app:logdepth} for the proof. Following the idea in Rosgen's construction~\cite{rosgen2007distinguishing}, we convert a circuit $W=U_T\hdots U_2 U_1$ to a parallel version. In the new parallelized circuit, we expect the honest prover to provide a witness of the form
\begin{align}
\ket{\xi} = \ket{\psi_0}\ket{\psi_1}^{\otimes 2}\hdots \ket{\psi_{T-1}}^{\otimes 2}\ket{\psi_T},
\label{eq:rosgen-history}
\end{align}
where $\ket{\psi_t}= U_t\hdots U_1 \ket{\psi_0}$. The verifier then verifies consistency between time steps via SWAP tests. It first performs SWAP tests between the two registers $\ket{\psi_t}\ket{\psi_t}$. Then it applies $U_{t+1}$ on a register $\ket{\psi_t}$ and performs SWAP tests between $U_{t+1}\ket{\psi_t}$ and $\ket{\psi_{t+1}}$. The verifier accepts if and only if all SWAP tests accept and the final measurement on the first qubit of $\ket{\psi}_T$ accepts. Since SWAP test can be done in log depth, this circuit is log depth. See \Cref{fig:claim6.3}.
Standard weak amplification procedure (\Cref{lem:amplify}) can be used to boost the soundness to $2^{-\poly (n)}$ while maintaining log depth.

There are two simple ways to make this log-depth circuit $W$ in \Cref{fig:claim6.3} fault-tolerant. These methods give rise to quasi-polynomial promise gap by applying~\Cref{thm:qma}. The first method is to simply set the injectivity parameter $\delta = 1/\poly(nD)$ and use the original unencoded log-depth circuit $W_\FT=W$, i.e., no fault tolerance protocol is needed. We use a single $1$-local Hamiltonian term $H_\mathrm{out} = \ket{0}\bra{0}^\mathrm{out}$ to check $W_\FT$'s output qubit. This leads to $2^{-O(\log^2 n)}$ promise gap. The second method is to set $\delta= \Omega(1)$ and use the Aharonov-Ben-Or fault tolerance scheme, which gives the circuit $W_\FT$ with polylog depth and width overheads. We set $\delta$ such that the noise rate is below the threshold. The polylog depth overhead leads to $2^{-\polylog(n)}$ promise gap. Thus, we have just shown that the local Hamiltonian problem with $n^{-O(\log n)}$ promise gap is $\QMA$-hard.

To achieve the desired $1/\poly(n)$ promise gap, we need to construct a fault-tolerant $\QMA$ circuit $W_\FT$ with logarithmic depth. This hints towards the need for a quantum fault tolerance scheme with constant depth overhead (including noisy classical computation). This is indeed possible in the classical setting~\cite{von1956probabilistic, dobrushin1977upper, pippenger1985networks} for any classical circuits. To our knowledge, it is unknown if this is possible quantumly without assuming that classical computation is free and noiseless. However, note that we only need such a scheme for some log-depth $\QMA$ circuit.

\subsection{$\QMA$-completeness of log-local Hamiltonian problem}
The previous subsection discussed how a fault-tolerant log-depth $\QMA$ circuit combined with \Cref{thm:qma} would give the desired proof of $\QMA$-completeness of $O(1)$-$\LH(1/\poly(n))$. Here, we construct such a fault-tolerant version of the $\QMA$ circuit from \Cref{claim:rosgen} at the cost of using a few layer of log-local gates. This suffices to prove $\QMA$-completeness of $O(\log n)$-$\LH(1/\poly n)$. With some slight modification we can also prove $\QMA$-hardness of $O(1)$-$\LH (n^{-O(\loglog n)})$. Notably, the proofs rely on a recent linear-distance quantum code family~\cite{leverrier2023decoding} and its local parallel decoder.

We will use a low-density parity-check CSS code family that admits a parallel decoder defined below. Below $|\cdot|_R$ denotes Pauli weight up to stabilizers.

\begin{definition}[QECC with parallel decoder]
\label{def:parallel-decoder} Let $\mathcal{Q}$ be an LDPC CSS code specified by parity check matrices $H_X \in \mathbb{F}_2^{r_X \times n}$ and $H_Z \in \mathbb{F}_2^{r_Z \times n}$. Let $e = (e_X,e_Z) \in \mathbb{F}_2^{2n}$ be a data error and $\sigma=(\sigma_X,\sigma_Z)=(H_Z e_X, H_X e_Z) \in \mathbb{F}_2^{r_X+r_Z}$ be the corresponding syndrome. A decoder for $\mathcal{Q}$ is $(\alpha,A,B)$-suppressing, for $\alpha < 1$, if there exists a constant $A$ such that, for $P\in \{X,Z\}$, whenever $
    |e_P|  \leq An$,
the decoder from given input $\sigma_{P}$ runs a computation in constant depth $B$ (using 2-bounded gates) to compute a correction $f_P \in \mathbb{F}_2^n$ such that $|e_P+f_P| \leq \alpha |e_P|$.
\end{definition}

This requirement is satisfied by the following family of constant-rate and linear-distance quantum LDPC codes and other good quantum LDPC code families with sufficient expansion~\cite{leverrier2023decoding, gu2023single}. Constant rate is not a necessary condition for us as we will use only one logical qubit per block. For QECC that encodes more than one logical qubits per block, we ignore all but one qubit which will be used in the computation.
 
\begin{lemma}[Parallel decoding Leverrier-Zémor code~\cite{leverrier2023decoding}]
\label{thm:LZcode} Consider the asymptotically good quantum Tanner code family in~\cite{leverrier2022quantum}. There exists a constant $A$ such that the parallel small-set flip algorithm with $k$ iterations (Algorithms 1 and 3 in~\cite{leverrier2023decoding}) is $(2^{-\Omega(k)}, A, O(k))$-suppressing for any positive number $k$. In particular, when $k=O(\log n)$, where $n$ is the codeblock size, the error suppression is $\alpha=0$.
\end{lemma}

We note that the above definition was stated for adversarial errors and hence we would need a code with linear distance. So $A$ can be understood to be some constant smaller than the code's relative distance. It is possible, and more standard, to have a weaker definition for stochastic noise by requiring~\Cref{def:parallel-decoder} to hold only with high probability.
Since our later analysis readily generalizes to the stochastic noise model via a union bound on the entire circuit, we will stick with~\Cref{def:parallel-decoder} and uses the linear distance code from~\Cref{thm:LZcode} in our later analysis for simplicity.

\begin{theorem}\label{thm:log-local-qma} The problem of deciding whether the ground energy density of $O(\log n)$-local Hamiltonians is $\leq a$ or $\geq a +  1/\poly(n)$, for some given number $a$, is $\QMA$-complete.
\end{theorem}

\begin{proof}
    As noted in~\Cref{sec:background}, the log-local Hamiltonian problem is in $\QMA$. We here show it is $\QMA$-hard.
    Let $n$ and $D=O(\log n)$ be the width and depth of the circuit $W$ from~\Cref{claim:rosgen}. Our goal is to construct a fault-tolerant version of the circuit $W$ in~\Cref{fig:claim6.3} while keeping the depth $O(\log n)$, and then apply~\Cref{thm:qma}. As discussed in the previous subsection, we can offload the measurements at the end of the circuit to the output Hamiltonian terms in $H_\mathrm{out}$, so we don't need to make that part fault-tolerant.

    For each qubit in $W$, we simulate it by an instance of the linear-distance CSS QECC family in~\Cref{thm:LZcode} of block size $m=O(\log n D)=O(\log n)$. Without loss of generality, we assume the gates $U_1,\hdots U_T$ in $W$ are CCZ, Hadamard, and CNOT. The SWAP tests can also be \emph{exactly} implemented with this gate set. We will use the following gadgets to fault-tolerantly simulate the circuit $W$ in \Cref{fig:claim6.3} using a circuit $W_\FT$ of width $n_\FT=O(n \log n)$ and depth $D_\FT=O(D)=O(\log n)$:
    \begin{enumerate}
        \item (Offline) Logical state preparations of $\ket{0}_L$, $\ket{+}_L$, $\mathrm{CZ}_{L}\ket{++}_L$, $\mathrm{CCZ}_L\ket{+++}_L$.
        \item Error correction gadget that incurs constant space and time overheads per round.
        \item Logical CZ, CCZ, and Hadamard gates via gate teleportation.
        \item Transversal logical CNOT gate.
    \end{enumerate}
    \textbf{Proof overview}: The important observation is that, apart from the log-depth CNOT gates part inside the SWAP test and the global OR computation at the end (see~\Cref{fig:claim6.3}), the verifier circuit $W$ is constant depth. The rest of the proof shows we can achieve gadget 2 for constant noise rate, while gadget 4 is trivially true as we are using CSS codes (see~\Cref{sec:background}). However, we do not know how to achieve gadgets 1 and 3 with constant noise rate without incurring super-constant depth overheads, so we instead use $O(m)$-local (i.e., $O(\log n)$-local) gates to perform gadgets 1 and 3 in ``one go''. Although these gates correspond to $O(\log n)$-local Hamiltonian terms, there are only $O(1)$ layers of them. Thus the spectral gap of $H_\mathrm{parent}$ can still be lower bounded by $1/\poly(n)$ by applying~\Cref{thm:spectral-gap} for circuits with varying gate localities. Finally we show that the fault tolerance of $W_\FT$ allows converting this spectral gap to $\QMA$ promise gap.

    We now explain in more details how to achieve each gadget. 
    
    \textbf{Gadget 4} is trivial for CSS codes.
    
    \textbf{Gadget 1.} 
The quantum state that we want to initialize is a logical code word $\ket{\psi}_L$ of size $m=O(\log n)$. For this, we simply use a $2m$-local initialization Hamiltonian term $h_\mathrm{in} = \Lambda^{\otimes m} (\id - \ket{\psi}\bra{\psi}_L) \Lambda^{\otimes m}$ in the injective PEPS construction. In the quantum circuit encoded in the injective PEPS ground state of $H_\mathrm{parent}$ (recall \Cref{sec:injectiveTN}), this corresponds to a perfect initialization of the ancilla state $\ket{\psi}_L$ followed by iid depolarizing noise on each qubit in $\ket{\psi}_L$. Thus, we avoid the need for a fault tolerance state preparation.

   \textbf{Gadget 2.} We use the linear-distance QECC with a local parallel decoder from~\Cref{thm:LZcode}, with suppresion parameters $(\alpha,A,B)=(2^{-\Omega(k)}, A, O(k))$ as defined in~\Cref{def:parallel-decoder}. For concreteness we choose $k$ to be a small constant such that the suppresion factor $\alpha=1/10$. Accordingly, $A$ and $B$ are constants independent of the code instance size.

   Suppose the input state $\ket{\psi}$ of the gadget deviates from a codeword by $\varepsilon m$ Pauli errors, where $\varepsilon < A/3$, so that the the guarantees of the~\Cref{thm:LZcode}'s decoder apply. As usual for CSS codes, we first show how to deal with $X$ errors, the same argument applies for $Z$ errors. Recall from~\Cref{sec:model} that the depolarizing noise rate is $\eta=3\delta^2/(1+3\delta^2)$, where $\delta$ is the PEPS injectivity parameter. We choose $\delta$ to be a sufficiently small constant such that the noise rate per gate in the circuit satisfies $\eta \ll \frac{A}{B2^B}$.
   We show that a \emph{noisy} implementation of the depth-$B$ circuit, denoted $\mathsf{EC}$, that performs~\Cref{thm:LZcode}'s parallel decoding algorithm still perform well to keep the number of Pauli errors under controlled throughout the computation. More specifically, $\mathsf{EC}$ initializes $O(Bm)$ ancilla qubits, performs reversibly the depth-$B$ computation of~\Cref{thm:LZcode}'s algorithm and writes each $i$-th bit of the X correction vector $f_X \in \mathbb{F}_2^m$ into an $i$-th ancilla qubit, and then performs the Pauli X correction transversally controlled on these $m$ ancilla qubits. The ancilla qubits used for reversible computation are uncomputed and discarded. $\mathsf{EC}$ is followed by a similar circuit that performs $Z$ error correction. All operations are subject to depolarizing noise of rate $\eta$ specified above.
   First, suppose the circuit $\mathsf{EC}$ had no faulty gates, then according to~\Cref{thm:LZcode} the residual error would be reduced to below $ \alpha \varepsilon m = \varepsilon m /10$. When including the random noise in the gates of $\mathsf{EC}$ back, we can apply the Chernoff's bound on the $O(mB)$ gates in $\mathsf{EC}$ and find that with probability $1-e^{-\Omega(mB)}$, there are at most $O(B) \cdot \eta m$ faulty gates within $\mathsf{EC}$. These faults cause at most $O(B) 2^B \eta m \triangleq \gamma m/2$ extra errors in the output state. Therefore, the number of $X$ errors in the output state is below $\varepsilon m /10 + \gamma m/2$. The same argument applies for $Z$ errors. We choose the PEPS injectivity parameter such that $\gamma < A/30$.
   
   In summary, if the input state of the gadget contains $\varepsilon m$ Pauli errors where $\varepsilon < A/3$, then with probability $1-e^{-\Omega(Bm)}$, the error count is suppressed to $\varepsilon m/5 + \gamma m < Am/15 + \gamma m < Am/10$ in the output state\footnote{This argument gives a short proof of the single-shot property (resilience to measurement errors) of the parallel small-set flip decoder for the quantum Tanner code, as also shown in~\cite[Theorem 1.3]{gu2023single}. This is because syndrome measurement error is a special case of gate faults in the error correction gadget circuit. Thus, our argument here shows that the $O(1)$-iteration parallel algorithm of~\cite{leverrier2023decoding} is single-shot (see Definition 1.1 in~\cite{gu2023single}). The single-shot property of the $O(\log m)$-iteration version of the algorithm also follows by repeating the argument $O(\log m)$ times. In particular, if the initial error weight is $\varepsilon m$, then after applying the $\mathsf{EC}$ circuit $O(\log m)$ times, the remaining error is $\gamma m (1 + \frac{1}{5} +\hdots + (\frac{1}{5})^{O(\log m)} ) +(\frac{1}{5})^{O(\log m)} \varepsilon m < 2\gamma m$.}. We will insert an EC gadget in between every two computational steps. This error suppression creates room for up to $(1/3-1/10)Am$ new errors to happen in the computation in between two EC gadgets. Finally, we choose the code block size to be $m=O(\log n)$ with a sufficiently large constant in front of $\log m$, such that each EC gadget only fails with probability $e^{-\Omega(mB)} = 1/\poly(W_\FT)$, allowing us to union bound over the entire circuit $W_\FT$. 

    \textbf{Gadget 3.} We use the gate teleportation circuits in \Cref{fig:teleport}. Here, we introduce the CZ gate into the already-complete gate set $\{\mathrm{CCZ},\mathrm{H}\}$ for convenience because CZ is used as a subroutine in the gate teleportation circuits for H and CCZ.
        \begin{figure}
        \centering
        \includegraphics[width=0.95\textwidth]{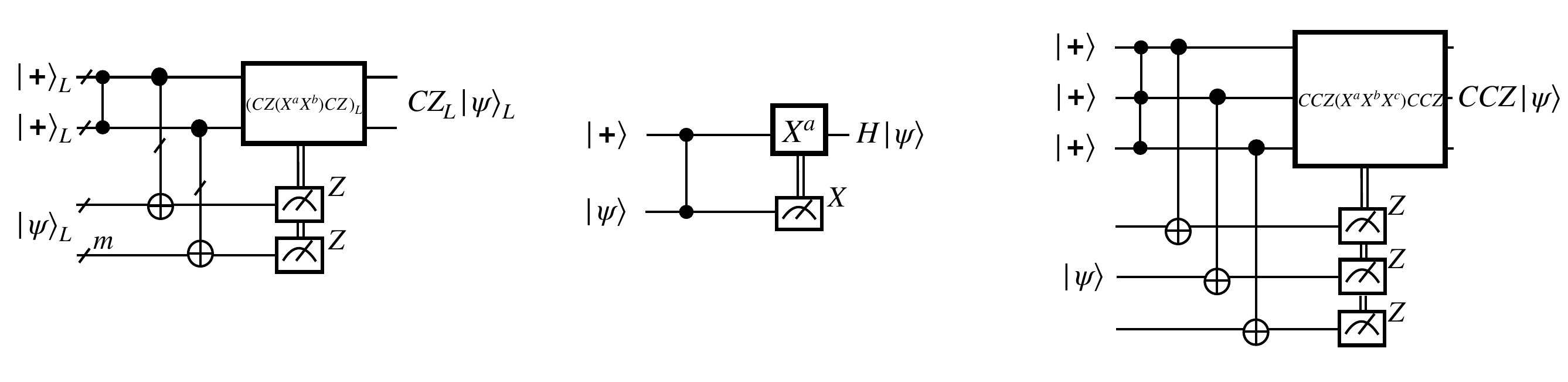}
        \caption{\textit{Left}: CZ gate teleportation using the ancilla state $CZ\ket{++}$. The correction $CZ(X^a X^b)CZ$ is a Pauli operation. \textit{Middle}: Hadamard gate teleportation. \textit{Right}: CCZ gate teleportation. The correction $CCZ(X^a X^b X^c)CCZ$ is a product of Pauli operators and CZ. For example, $CCZ(X\otimes \id \otimes \id)CCZ= X \otimes CZ$. The logical X/Z measurement can be transversally done in the physical X/Z basis for CSS codes.
        }
        \label{fig:teleport}
    \end{figure}

    We explain how to fault-tolerantly perform these circuits using log-local gates. As an illustration, consider the CZ gate teleportation circuit in~\Cref{fig:teleport}. The logical ancilla state $\mathrm{CZ}_L\ket{++}_L$ can be prepared using gadget 1, which corresponds to a $4m$-local initialization term $h^{CZ}_\mathrm{in} = \Lambda^{\otimes 2m}(\id - \mathrm{CZ}_L\ket{++}_L\bra{++}_L \mathrm{CZ}_L^\dagger) \Lambda^{\otimes 2m}$ in the parent Hamiltonian. Next, the logical CNOT gates are applied transversally with 2-qubit physical gates. Finally, we use a $4m$-qubit gate $U^{CZ}_\mathrm{teleport}$ to perform, in a single step, both the logical measurement and the logical Pauli correction. It has the following form
    \begin{align}
        U^{CZ}_\mathrm{teleport} = \sum_{x \in \{0,1\}^{2m}} \ket{x}\bra{x} \otimes f(x),
    \end{align}
     where $f: \{0,1\}^{2m} \mapsto \mathcal{P}^{\otimes 2 m} $ is the logical decoding function that maps a $2m$-bit word $x$ (the measurement outcomes) to a $2m$-qubit Pauli string corresponding to the associated logical Pauli correction. Note that the depth of the gadget is 3: initialization, transversal CNOT, $U^{CZ}_\mathrm{teleport}$.
    
\begin{figure}
    \centering
\includegraphics[width=0.99\textwidth]{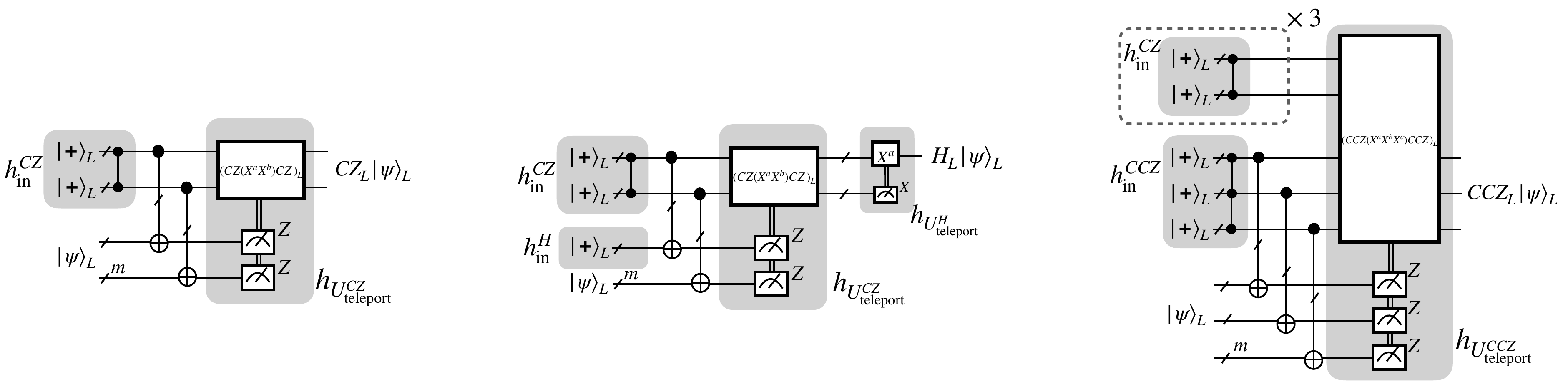}
    \caption{Fault-tolerant realizations of the encoded versions of the teleportation circuits in~\Cref{fig:teleport} using $O(m)$-local gates and Hamiltonian terms. The Hadamard circuit uses CZ as a subroutine. The CCZ circuit adaptively uses up to 3 CZ's within its Clifford correction.}
    \label{fig:log-local}
\end{figure}
    
    The gate $U^{CZ}_\mathrm{teleport}$ corresponds to a $16m$-local propagation Hamiltonian term $h_{U^{CZ}_\mathrm{teleport}} = \Lambda^{\otimes 8m} (\id - \ket{\Phi_{U^{CZ}_\mathrm{teleport}}}\bra{\Phi_{U^{CZ}_\mathrm{teleport}}} ) \Lambda^{\otimes 8m}$, where $\ket{\Phi_{U^{CZ}_\mathrm{teleport}}}= U^{CZ}_\mathrm{teleport} \ket{\Phi_I}^{\otimes 4m}$. In the quantum circuit encoded in the injective PEPS ground state of $H_\mathrm{parent}$, this corresponds to a layer of iid depolarizing noise on the $4m$ input qubits, followed by $U^{CZ}_\mathrm{teleport}$. Importantly, this gate is fault-tolerant despite being nonlocal. This is because the function $f$ is robust to input noise: for a linear-distance QECC, flipping a sufficiently small fraction of bits in $x$ leaves $f(x)$ unchanged. Therefore, any Pauli error of sufficiently small Pauli weight (up to a constant fraction of $m$) preceding $U^{CZ}_\mathrm{teleport}$ does not cause a logical error in $f(x)$ and is also not spread out by the transversal Pauli $f(x)$.
    
    We can similarly fault-tolerantly realize the Hadamard and CCZ gate teleportation circuits (which use the CZ gate as a subroutine) at the physical level by using $O(m)$-local propagation Hamiltonian terms. This is straightforward for the Hadamard circuit, see middle panel of~\Cref{fig:log-local}. The CCZ circuit is slightly more involved since it \emph{adaptively} uses up to 3 CZ's within the Clifford correction $CCZ(X^a X^b X^c)CCZ$. For example, $CCZ(X\otimes \id \otimes \id)CCZ= X \otimes CZ$. So we need to prepare 3 ancilla states $CZ_L \ket{++}_L$ in addition to the magic state $CCZ_L\ket{+++}_L$. Accordingly, we use a $12m$-qubit gate $U^{CCZ}_\mathrm{teleport}$ as shown in the right panel of~\Cref{fig:log-local}. An explicit formula for $U^{CCZ}_\mathrm{teleport}$ is rather lengthy, so we do not include the details here. 
    However, as for $U^{CZ}_\mathrm{teleport}$, this gate can be seen to be fault-tolerant because it consists of transversal gates (including Paulis, CNOTs, and possibly SWAPs) controlled on the values of logical decoding functions that are robust to input noise.

    \paragraph{Combining the gadgets.} The fault-tolerant circuit $W_\FT$ is obtained by replacing the ideal gates in $W$ (\Cref{fig:claim6.3}) by the described gadgets, except that $W_\FT$ does not include the final OR computation on the SWAP test outcomes and the output qubit in the final time step register $\ket{\psi_T}$ (recall~\Cref{eq:rosgen-history}). Instead, we offload this computation to $H_\mathrm{out}$ as explained later. Since each gadget has constant depth, $W_\FT$ has depth $D_\FT= O(D) = O(\log n)$. Furthermore, by the linear distance of the QECC and Chernoff’s bound we find that each gadget fails with probability at most $e^{-\Omega(m)}$. Thus, choosing $m=O(\log n)$ (with a sufficiently large constant in front of $\log n$) and using the union bound, we obtain that the logical error rate in $W_\FT$ is at most $\delta_L = 1/\poly(n)$ for any large polynomial $\poly(n)$ of choice. The parent Hamiltonian $H_\mathrm{parent}$ corresponding to $W_\FT$ has $O(\log n)$ locality and injectivy $\delta$ everywhere, where $\delta$ is chosen to be a sufficiently small constant such that the analysis in gadget 2 holds.
    
    The main observation is that, apart from the log-depth transversal CNOTs part within each SWAP test, $W_\FT$ is constant-depth. This allows us to obtain a $1/\poly(n)$ lower bound on the spectral gap of $H_\mathrm{parent}$.

    \begin{claim} The parent Hamiltonian corresponding to $W_\FT$ has spectral gap at least $\gamma = 1/\poly(n)$.
    \end{claim}
    \begin{proof}
        We apply the version of~\Cref{thm:spectral-gap} that allows gate locality to vary in the circuit.~\Cref{thm:spectral-gap} states that the spectral gap is lowerbounded by $\gamma = \frac{1}{\poly(n_\FT D_\FT)} \prod_{\ell=0}^{D_\FT} \delta^{8k_\ell}$, where $k_\ell$ is the gate locality at layer $\ell$ of $W_\FT$. But $W_\FT$ has depth $D_\FT= O(\log n)$ and only a constant number of layers have $O(\log n)$ locality. Thus, $\gamma = 1/\poly(n)$.
    \end{proof}
    
    Next, we describe how to offload the global computation on the measurement outcomes at the end of~\Cref{fig:claim6.3} to the output Hamiltonian terms. This computation expects that all of the logical measurements at the end of $W_\FT$ to output logical value 1. However, instead of including these logical measurements and global OR computation, we use the following output check Hamiltonian to check each of the measurement outcomes
\begin{align}
    H_\mathrm{out} =  \sum_{\text{measured block } j}  \Pi_j^\mathrm{out},
\end{align}
    which consists of $p=O(n)$ identical $O(\log n)$-local projector terms $\Pi_j^\mathrm{out}$ acting on the last qubit column in the tensor network. The projectors have the form 
    \begin{align*}
        \Pi = \mathrm{span}\{ E\ket{0_L}: E \text{ is correctable Pauli}\},
    \end{align*}
    i.e., $E$ has weight smaller than half the QECC distance.
    
    Altogether, we have the following $O(\log n)$-local Hamiltonian on $O(n_\FT D_\FT)= O(n \log n)$ qubits:
    \begin{align}
        H_\mathrm{total} = H_\mathrm{parent} + \gamma H_\mathrm{out}.
    \end{align}
    Finally, we apply~\Cref{thm:qma} to conclude that deciding the ground energy of $H_\mathrm{total}$ to precision $\frac{\gamma(1-s)}{8p} = \frac{1}{\poly (n)}$ is $\QMA$-hard. Note that the requirement $\max\{1-c, \delta_L\} < \frac{1-s}{16p^2}$ for~\Cref{thm:qma} to hold can be satisfied by choosing the QECC block size to be sufficiently large (recall from~\Cref{claim:rosgen} that the completeness is $c=1-2^{-\poly(n)}$ and $s= 1-1/n^\alpha$ for a constant $\alpha \leq 3$). This concludes the proof of~\Cref{thm:log-local-qma}.

\end{proof}

For the case of $O(1)$-local Hamiltonian, we can also prove $\QMA$-hardness with an inverse superpolynomial promise gap.

\begin{theorem}\label{thm:loglogqma} The problem of deciding whether the ground energy density of $O(1)$-local Hamiltonians is $\leq a$ or  $\geq a+ n^{-O(\loglog n)}$, for some given number $a$, is $\QMA$-hard.
\end{theorem}
\begin{proof}
    We use essentially the construction in the proof of~\Cref{thm:log-local-qma}. The difference is that to implement gadgets 1 and 3, we make use of the ability to vary the injectivity in the parent Hamiltonian construction. This allows us to use $O(1)$-local gates and thus $O(1)$-local parent Hamiltonian.
    In particular, in the state preparation and CZ/CCZ/Hadamard gate teleportation gadgets, we simply use the `bare' non-fault-tolerant circuit to implement them and set the injectivity within these gadgets to be $\delta'=1/\poly(n)$, making them effectively noiseless. Importantly, in the absence of noise, these gadgets can be done in depth logarithmic in the code block size, i.e., $O(\log m)=O(\loglog(n))$.
    For example, logical Z measurement and logical measurement outcome decoding in the CZ and CCZ teleportation circuits can be noiselessly done by `copying' the state in Z basis using CNOTs to workspace ancillas initialized to $0$'s, perform a classical computation that decodes the logical information in $O(\log m)$ depth, and apply the Pauli/Clifford correction circuit controlled (gate-by-gate) on the decoded logical value\footnote{More specficially, after `copying' the Z basis CSS codewords onto workspace ancillas, we first use the $k=O(\log m)$-iteration decoder in \Cref{thm:LZcode} to remove all errors present in the input state $\ket{\psi}$ and the ancilla logical patch. Now, the workspace ancillas are in a superposition of CSS codeword bit strings \Cref{eq:CSScodeword}. Performing the modulo 2 dot product between such a bit string and the bitstring representing logical Pauli Z then yields the logical value of the qubit. The dot product can be done in $O(\log m)$ depth. The decoded logical value is then `copied' in $O(\log m)$ depth into $\poly(m)$ bits that is used to controlledly perform the logical correction circuits $CZ(X^a X^b)CZ$ and $CCZ(X^aX^bX^c)CCZ$. A similar procedure can be done for the logical X measurement in Hadamard gate teleportation since we are using CSS codes.}. The logical ancillas in the state injection step can also be prepared non-fault-tolerantly by $O(\log m)$-depth circuits since Clifford circuits are parallelizable~\cite{broadbent2009parallelizing} to log-depth. The parameter injectivity in other locations of the circuit $W_\FT$ remains constant $\delta=\Theta(1)$. 

        Finally, also using injectivity $\delta'=1/\poly(n)$ we can decode (using the `bare' circuits) the SWAP test measurements and the $\QMA$ decision measurement and write each decoded outcome onto a physical qubit. To verify the measurement outcomes, we use the 1-local Hamiltonian $H_\mathrm{out}= \sum_{\text{measurement } j} \ket{0}\bra{0}_j$. By choosing the code block size $m=O(\log n)$ to be sufficiently large and the injectivity $\delta'=1/\poly(n)$ to be sufficiently small in the gate teleportation circuit and final logical measurement decoding (and $\delta=\Theta(1)$ at everywhere else), we can guarantee $W_\FT$ has logical error rate at most $1/n^{\beta}$ for any constant $\beta$ of choice. Next we apply \Cref{thm:spectral-gap} and find that the spectral gap of $H_\mathrm{parent}$ is at least $\gamma = \frac{1}{\poly(n_\FT D_\FT)} \prod_{\ell=0}^{D_\FT} \delta^{O(1)}$, where $\delta_\ell$ is the injectivity at layer $\ell$ of $W_\FT$. But $W_\FT$ has depth $D_\FT= O(\log n)$ and only $O(\log m)=O(\loglog n)$ layers of it use injectivity $\delta'=1/\poly(n)$, so $\gamma = n^{-O(\loglog n)}$. The total Hamiltonian $H_\mathrm{total} = H_\mathrm{parent}+ \gamma H_\mathrm{out}$, which is $O(1)$-local, can then be seen to have a $\QMA$ promise gap of $n^{-O(\loglog n)}$ by~\Cref{thm:qma}.
\end{proof}

\section{Computational complexity of injective PEPS}\label{sec:bqp-hardness}
We now discuss a hardness result on the creation and contraction of certain tensor network states that follows from our construction.
A summary of our results can be found in \Cref{tab:PEPScomplexity}.

\begin{definition}[PEPS]
    A \emph{projected entangled pair state (PEPS)} is any (unnormalized) state that can be obtained by the following procedure:
    consider a graph and associate to each vertex $v$ as many $D$-dimensional spins as there are edges incident to $v$.
    Assume that the spins associated to the end points of an edge form maximally entangled states $\ket{\mathrm{EPR}_D}=\sum_{i=1}^D \ket{i}\ket{i}$.
    The PEPS is obtained by applying a linear map $P_v : \mathbb{C}^D \otimes \dotsb \otimes \mathbb{C}^D \to \mathbb{C}^d$ at each vertex~$v$. Without affecting the computational complexity, we further allow the virtual states to be any maximally entangled states of the form $(\id \otimes U)\ket{\mathrm{EPR}_D}$. We can also assume $\|P_v\| \leq 1$. 
\end{definition}

In \cite{schuch2007computational} it was shown that preparing PEPS as a quantum state is $\PostBQP$-hard, where $\PostBQP$ is a large complexity class that contains $\QMA$. The idea of the proof is that measurement-based quantum computation with the power to post-select on the measurement outcomes reduces to preparing a PEPS.
The power to post-select on the outcomes of a quantum computation is due to the fact that the local maps $P_v$ are allowed to be non-invertible.
Hence, it is natural to ask what happens when we reduce the power of preparing arbitrary PEPS by removing the ability to post-select.
We do this by considering a subclass of tensor networks called \emph{injective PEPS}~\cite{schuch2010peps}.

\begin{definition}[Injective PEPS] A PEPS on $n$ spins is \emph{$\delta(n)$-injective} if the local maps $P_v$ are non-singular matrices with  singular values bounded from below by $\Omega(\delta(n))$.
\end{definition}

Our construction gives the following hardness result on the preparation of injective PEPS.
\begin{theorem}\label{thm:PEPS_BQP} Preparing constant-injective PEPS states in two or higher dimensions with bond dimension $D \geq 4$ and physical dimension $d \geq 4$ allows solving $\BQP$-hard problems.
\end{theorem}
\begin{proof}
\Cref{claim:gs} tells us that we can encode a noisy quantum computation suffering from i.i.d.\ depolarizing noise with constant error probability $\delta^2/(1+3\delta^2)$ into a $\delta$-injective PEPS state.
Choosing $\delta$ to be a sufficiently small constant, the quantum fault-tolerance threshold theorem (\Cref{thm:quantumFT}) states that we can $\varepsilon$-approximate any noise-free circuit~$C$ with a noisy circuit~$\tilde{C}$ with polylogarithmic overhead in~$1/\varepsilon$. 
Note that the threshold theorem holds even when the circuit connectivity is restricted to one dimension (\Cref{thm:quantumFTdDim}), so the computational hardness persists on two-dimensional injective PEPS.
\end{proof}
The above $(\mathsf{state})\BQP$-hardness result can be understood as a complement to previous works in efficient quantum algorithms for preparing injective PEPS under assumptions on the parent Hamiltonian spectral gap~\cite{schwarz2012preparing, ge2016rapid}. 
We leave it as an open question to obtain tight upper bound on the complexity of preparing injective PEPS states (an upper bound is $(\mathsf{state})\PostBQP$ due to~\cite{schuch2007computational}).

We next discuss the classical complexity of injective PEPS. PEPS is conceived as an efficient classical description of quantum states and an important application is contracting a PEPS in order to evaluate the value of a given observable.

\begin{task}[PEPS observable contraction]
\label{task}
Given a PEPS describing an unnormalized state $\ket{\psi}$ and a local observable $O$,
calculate the normalized expectation value $\frac{\bra{\psi} O \ket{\psi}}{\braket{\psi}}$.
\end{task}

Ref.~\cite{schwarz2017approximating} gave a quasi-polynomial time classical algorithn to contract injective PEPS under assumptions on the parent Hamiltonian spectral gap. The complexity of injective PEPS has also been implicitly studied in \cite{haferkamp2020contracting}, where the authors showed that random PEPS, whose local maps are i.i.d.\ Gaussian are $\SharpP$-hard to contract to exponential additive precision. 
However, the random PEPS ensemble of \cite{haferkamp2020contracting} has injectivity $1/\Omega(\poly (n))$ with high probability, and their result does not necessarily indicate that a $\SharpP$-hard PEPS instance would have constant injectivity. Similarly, the $\SharpP$-hard PEPS instances in~\cite{schuch2007computational} can be seen to be non-injective, even after blocking\footnote{This is because their local maps have the form $\ket{0}\bra{0}$, which will remain being rank-1 after blocking.}~\cite{schuch2010peps}. Using the construction in this work, we obtain the following hardness results for contracting PEPS with \textit{constant} injectivity.

\begin{theorem} 
\label{thm:contraction-bqp}
    For constant-injective PEPS states in two or higher dimensions with bond dimension $D \geq 4$ and physical dimension $d \geq 4$, evaluating local observable expectation values to $O(1)$ additive error is $\BQP$-hard. 
\end{theorem}
\begin{proof}
    The $\BQP$-hardness, similar to \Cref{thm:PEPS_BQP}, follows from encoding a noisy $\BQP$ computation  $C$ into our injective PEPS $\ket{\Psi}$ and invoking the threshold theorem. The difference is that at the end of the fault-tolerant circuit $\tilde{C}$ that simulates the noiseless circuit $C$ in~\Cref{thm:quantumFT}, we further perform a fault-tolerant decoding circuit that transforms the encoded output into a physical output state. For concatenated-code fault tolerance, this procedure is described in Section 4 of \cite{knill1996concatenated} (also see Section 6 of~\cite{gottesman2014overhead}). This decoding circuit results in a physical error rate per physical qubit of the output state which is bounded by some constant value. We encode the entire fault-tolerant circuit, including the fault-tolerant decoding part, into our injective PEPS with noise rate below the threshold. Then evaluating to $O(1)$-additive error the Pauli-$Z$ expectation on the first qubit of the PEPS output column decides the $\BQP$ computation $C$.
\end{proof}

\begin{table}[h]
    \centering
    \bgroup
    \def\arraystretch{1.5}%
    \begin{tabular}{|c|c|c|}
         \hline 
         \textbf{Task} & \textbf{PEPS} & \textbf{Injective PEPS}\\
         \hline
         State preparation & $\PostBQP$-complete & $\BQP$-hard \\
         \hline 
         Multiplicative-error contraction & $\SharpP$-complete & $\SharpP$-complete$^*$ \\
         \hline
         Additive-error contraction & $\BQP$-hard & $\BQP$-hard \\
         \hline
    \end{tabular}
    \egroup
    \caption{Computational complexity of general PEPS~\cite{schuch2007computational} and constant-injective PEPS. $^*$The $\SharpP$-hardness of injective PEPS requires a specific non-local observable in~\Cref{thm:sharpP}.}
    \label{tab:PEPScomplexity}
\end{table}

If we instead consider $O(1)$-multiplicative error expectation value evaluation of PEPO non-local observables in \Cref{task}, then we obtain a classical hardness matching that of general non-injective PEPS~\cite{schuch2007computational}. Is is a simple observation that \textit{exact} observable evaluation of \Cref{task} for general PEPS and PEPO observables\footnote{\emph{Projected entangled pair operators (PEPO)} are an efficiently describable family of operators, which can be thought of as the operator version of PEPS, i.e., local maps are $P_v : \mathbb{C}^D \otimes \dotsb \otimes \mathbb{C}^D \to \mathbb{C}^{d \times d}$.} is in $\SharpP$. For this, we reduce this task to norm evaluation of PEPS, which was shown to be in $\SharpP$ in~\cite{schuch2007computational}. Observe that $\bra{\psi}O \ket{\psi}= (\bra{\psi}(O+ \id)(O + \id) \ket{\psi} - \bra{\psi}O O \ket{\psi} - \bra{\psi} \ket{\psi})/2$. Each of $(O + \id) \ket{\psi}$, $O \ket{\psi}$, and $\ket{\psi}$ are PEPS states since $O$ is a PEPO. Thus, evaluating $\bra{\psi}O \ket{\psi}/\bra{\psi}\ket{\psi}$ is in $\SharpP$.

In order for our construction to go through for multiplicative errors, we have to constrain the type of observable to be describable as a \emph{tree tensor network (TTN)} which is a PEPO defined on a tree graph~\cite{ran2020tensor}.

\begin{theorem}
\label{thm:sharpP}
   For constant-injective PEPS states in two or higher dimensions with bond dimension $D \geq 4$ and physical dimension $d \geq 4$, evaluating the expectation value of a tree tensor network observable to $O(1)$-multiplicative error is $\SharpP$-hard. 
\end{theorem}

\begin{proof}
    We encode the ``quantum sum'' problem\footnote{Ref.~\cite{schuch2007computational} instead used the ``classical sum'' $\sum_x f(x)$, whose multiplicative error estimation is significantly easier than $\SharpP$. So strictly speaking, $\SharpP$-hardness of multiplicative error estimation of general PEPS contraction does not follow from Ref.~\cite{schuch2007computational}.}, well-known in the random circuit sampling literature~\cite{hangleiter2023computational}, into a fault-tolerant quantum circuit that exponentially suppresses the local depolarizing noise:
    \begin{equation}
        S = \sum_{x \in \{0,1\}^n} (-1)^{f(x)}, \qquad f: \{0,1\}^n \mapsto \{0,1\}.
    \end{equation}
    It is well known that
    $O(1)$-multiplicative error approximation of $S^2$ remains $\SharpP$-hard (under Turing reduction) and can be converted into $O(1)$-multiplicative error estimation of the amplitude $|\bra{0^n}C\ket{0^n}|^2=S^2$ of a quantum circuit $C$~\cite{hangleiter2023computational}.

    This can be further simplified to a measurement on one qubit as follows. Let $W$ be a reversible circuit that computes the inverse OR function, using additional ancillas initialized to a computational basis state $\ket{0^m}_\mathrm{anc}$, and then writing the result on the $(n+1)$-st qubit.
    Note the identity $(\id\otimes \bra{0^m}_\mathrm{anc})W^{\dagger}\ketbra{0}_{n+1}W(\id\otimes \ket{0^m}_\mathrm{anc}) = \ketbra{0}^n$. This implies - defining $C'=WC$ and abbreviating $\ket{0^n}\otimes \ket{0^m}_\mathrm{anc}$ as $\ket{0^{n+m}}$ - that $|\bra{0} C' \ket{0^{n+m}}|^2=|\bra{0^n}C\ket{0^n}|^2$.
    
    Next, we use a concatenated-code fault-tolerant circuit $\Tilde{C}$ from \Cref{thm:quantumFT} that approximates $C'$ to additive error $\varepsilon=e^{-\Omega(n)}$, which implies that $|\bra{\Bar{0}} \Tilde{C} \ket{\Bar{0}^{n+m}}|^2$ approximates $|\bra{0} C' \ket{0^{n+m}}|^2=S^2$ to $O(1)$-multiplicative error. Unlike in the proof of \Cref{thm:contraction-bqp}, we cannot afford to perform the noisy fault-tolerant decoding circuit, as doing so would no longer guarantee a multiplicative-closeness between $|\bra{\Bar{0}} \Tilde{C} \ket{\Bar{0}^{n+m}}|^2$ and $S^2$. Instead, we directly read out the logical information in the output of the noisy execution of $\tilde{C}$ by performing a recursive majority vote on the first logical qubit register (cf. Lemma 9 in~\cite{aharonov1999fault}). 
    The observable that represents this majority vote is a tree tensor network. 
    For concreteness, suppose $\Tilde{C}$ is constructed from concatenating $L$ levels of the $[[7,1,3]]$-Steane code, then this recursive majority vote means that we first take the majority in each block of size $7$, then we take the majority, of 7 such majority bits, and so on for~$L$ levels, to give one output bit. Here, $L=\Theta(\log n)$ for the desired  $\exp(-\Theta(n))$ error suppression, which means each logical code block consists of $\poly(n)$ physical qubits. This recursive majority vote can readily be encoded in an $L$-level tree tensor network operator $O$ as illustrated in \Cref{fig:recursive-maj}. Thus, evaluating the expectation value of $O$ to multiplicative error on our injective PEPS gives a multiplicative error estimation of $S^2$.
\begin{figure}
    \centering
    \includegraphics[width=0.66\textwidth]{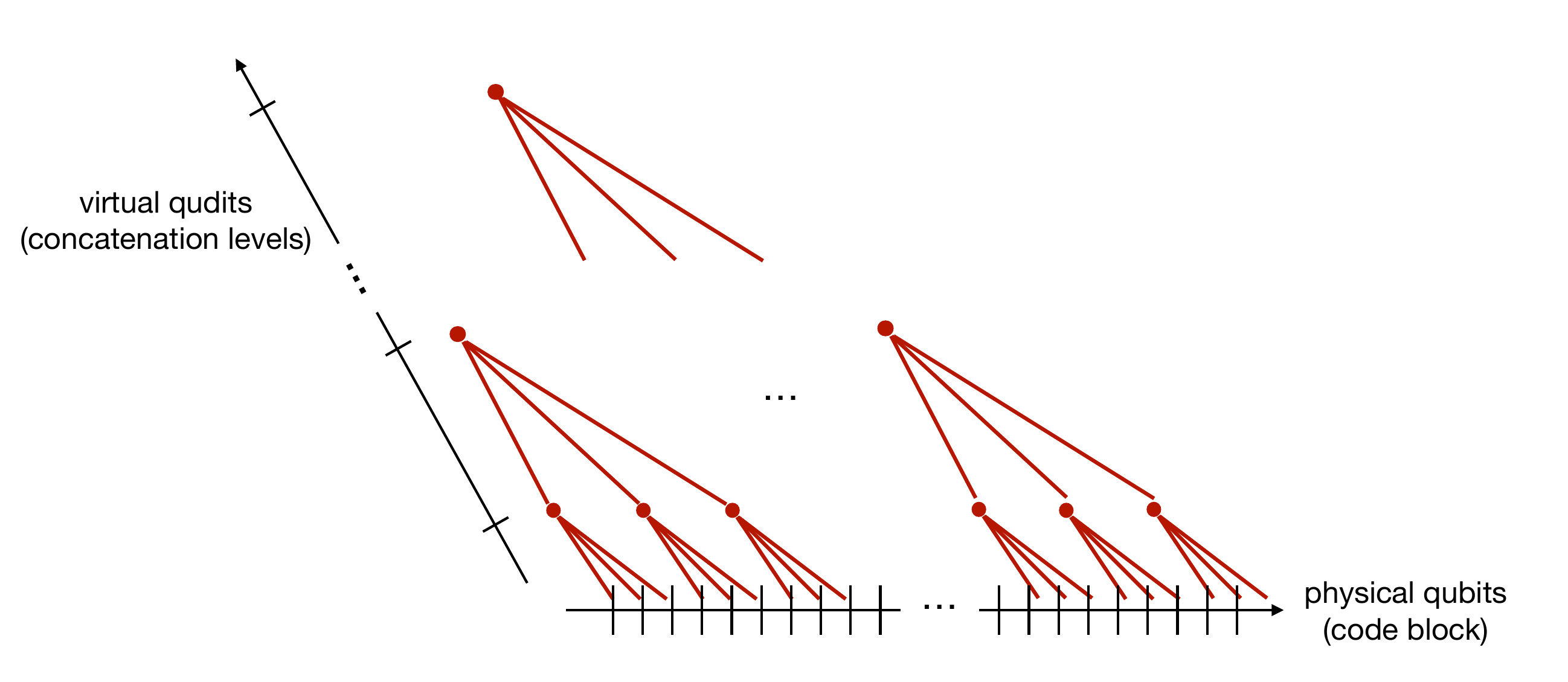}
    \caption{The recursive-majority readout of a codeblock is encoded into a tree tensor network observable.}
    \label{fig:recursive-maj}
\end{figure}
\end{proof}

Since tree tensor networks are themselves easy to contract (similar to an MPO), we conclude from \Cref{thm:sharpP} that the $\SharpP$ hardness must be arising from the injective PEPS itself. Proving the same theorem with local - or product - observables is an interesting open question.

\section{Open questions}\label{sec:open}

This work brings up a series of relevant open questions.

\begin{itemize}
    \item Our main question is if we can achieve a soundness of $1/\poly(D)$. In our proof of \Cref{thm:energy}, two steps are needed - robust teleportation of $H_{in}$ and a small number of Pauli errors in the low energy states. 
    The challenging part is the robust teleportation of $H_\mathrm{in}$, which we do not know how to achieve when the energy is $1/\poly(D)$. It turns out that we can enforce low Pauli errors by adding new Hamiltonian terms with $\poly(D)$ locality.
    Specifically, for each local region $A$ of $\frac{D}{\delta^4}$ EPR locations we can define a Hamiltonian term that penalizes $\geq 10\delta^2\cdot\frac{D}{\delta^4}$ Pauli errors, i.e.,
    \begin{equation*}
     h_{A}^\mathrm{low} = \sum_{\Vec{P} \in \mathcal{P}^{A}: |\Vec{P}| \geq 10\delta^2\cdot\frac{D}{\delta^4}}   \ket{\Phi_{\Vec{P}}}\bra{\Phi_{\Vec{P}}}.
    \end{equation*}
    Since $\delta$ is chosen $\frac{1}{\poly(D)}$, there new Hamiltonian terms are $\poly(D)$ local, which is $\polylog(n)$ when $D=\polylog(n)$.
    In addition, these Hamiltonian terms are unchanged under the unitary $V$ defined in~\Cref{sec:injectiveTN}. Now, lets choose a collection of regions $A_1,\ldots A_m$ (with $m=O(n\delta^4)$ such that each of $nD$ EPRs is involved in $O(1)$ regions) and add the following Hamiltonian to the existing Hamiltonian $H^\mathrm{low}=\frac{1}{m}\sum_ih^\mathrm{low}_{A_i}$. Note that $H^\mathrm{low}$ does not change the ground state energy density too much:
    The ground state achieves energy density of at most $e^{-27 D/\delta^2}$ by Chernoff bound.
    On the other hand, any state with energy density at most $\frac{1}{100D^2}$ has the property that it has at most $O(D\delta^2+\frac{1}{D})$ fraction of Pauli errors with probability $1-1/D$. For this, consider the `fraction of Pauli errors' operator in local regions $\frac{\delta^4}{D}\sum_{j\in A_i}(\id - \Phi_{I_j})$. Note that 
    $\frac{\delta^4}{D}\sum_{j\in A_i}(\id - \Phi_{I_j})\preceq 10\delta^2 (\id-h^\mathrm{low}_{A_i})+h^\mathrm{low}_{A_i}\preceq 10\delta^2 \id +h^\mathrm{low}_{A_i}$. Thus, the
    total fraction of Pauli errors satisfies
    $$\frac{1}{nD}\sum_{j}(\id - \Phi_{I_j})\preceq O(1)\frac{1}{m}\sum_{i=1}^m (\frac{\delta^4}{D}\sum_{j\in A_i}(\id - \Phi_{I_j}))\preceq O(1)\cdot 10\delta^2 \id +\frac{O(1)}{m}\sum_{i=1}^mh^\mathrm{low}_{A_i}.$$
    Thus, the expected fraction of Pauli errors in such low energy states is $O(\delta^2 + 1/D^2)$. The claim follows from Markov's inequality.

    \item In the introduction and~\Cref{append:classicalPCP}, we outline a connection between polylog-PCP and adversarial fault tolerance in the classical setting. It is expected that adversarial fault tolerance may use good classical codes, but we do not see a clear use of local decodability. Could classical polylog-PCP be achieved without strong reliance on local decodability? 
    \item Can the depth of $\BQP$ circuits be reduced to polylogarithmic in the input size? This does not follow from the depth reduction of $\QMA$ due to the presence of witness. Thus, the heart of the question is if the ground state of the tensor network Hamiltonian can be prepared in low depth when witness is absent. One possibility is to run an adiabatic algorithm tuning $\delta$ from $1$ to a smaller value. The spectral gap in this process is likely small -- we can show a spectral gap lower bound of $\Omega(e^{O(D)}/\mathrm{poly}(nD))$ in \Cref{thm:spectral-gap}. But suppose that we go ahead and tune $\delta$ adiabatically for small duration, can we argue that we end up in a low energy state of the parent Hamiltonian (not necessarily the ground state as in standard adiabatic computation)? If that is the case, we would still encode the answer to the computation if we started from a fault-tolerant circuit.

\end{itemize}

\subsection{Acknowledgements}

We thank Fernando Brand\~{a}o, Steve Flammia, Sunny He, Yeongwoo Hwang, Zeph Landau, Spiros Michalakis, Chinmay Nirkhe, Chris Pattison, Mehdi Soleimanifar and Umesh Vazirani for helpful discussions. We especially thank Steve Flammia for pointing us to the reference \cite{bartlett2006simple}, Chinmay Nirkhe for pointing us to \cite{Chen20}, and Sunny He and Chris Pattison for explaining to us the parallel decoder in~\cite{leverrier2023decoding, gu2023single}. AA and QTN acknowledge support through the NSF Award No. 2238836. AA acknowledges support through the NSF award QCIS-FF: Quantum Computing \& Information Science Faculty Fellow at Harvard University (NSF 2013303). QTN acknowledges support through the Harvard Quantum Initiative PhD fellowship. 

\bibliographystyle{unsrt}
\bibliography{bibliography.bib}

\appendix
\section{Adversarial fault tolerance and polylog-weaker classical PCP}
\label{append:classicalPCP}
Here we discuss the connection between adversarial fault tolerance and polylog-weaker classical PCP \cite{gal1995fault}. We take any $\NP$-hard classical Constraint Satisfaction Problem (CSP) $C
_1 = \frac{1}{m}\sum_iC_{1,i}$ on $n$ bits in which each constraint acts on a constant number of bits and each bit participates in a constant number of constraints (for example, a 2D Ising model).

We can verify whether $C_1$ is satisfiable in logarithmic depth via a similar circuit as in~\Cref{sec:shallow_circuits}. Specifically, the verifier (1) adds $m$ ancilla bits ($i$th bit corresponding to the $i$th constraint) initialized to 0, (2) asks prover for the satisfying solution, and for each constraint, (3) flips the corresponding ancilla bit if the corresponding constraint was violated. Step (2) can be done in $O(1)$ depth since the constraints can be divided into $O(1)$ groups such that each group contains only non-intersecting constraints. Once this is done, we can run an $\bar{\text{OR}}$ function in $O(\log n)$ depth on the ancilla bits to accept or reject.

\begin{figure}[h]
    \centering
    \begin{tikzpicture}[scale=0.4]
  \clip (-0.65,-0.5) rectangle (19.65,11.5);
  \foreach \x in {0,8,16,24} {
    \foreach \y in {0,4,8} {
      \draw[ultra thick] (\x+1,\y) rectangle (\x+2,\y+3);
      \draw[thick] (\x-2,\y+0.5) -- (\x+1,\y+0.5); 
      \node at (\x-0.5,\y+0.5){\small \textcolor{red}{x}};
      \draw[thick] (\x-2,\y+2.5) -- (\x+1,\y+2.5); 
      \node at (\x-0.5,\y+2.5){\small \textcolor{red}{x}};
    }
  }
  \foreach \x in {4,12,20} {
    \foreach \y in {-2,2,6,10} {
      \draw[ultra thick] (\x+1,\y) rectangle (\x+2,\y+3);
      \draw[thick] (\x-2,\y+0.5) -- (\x+1,\y+0.5); 
      \node at (\x-0.5,\y+0.5){\small \textcolor{red}{x}};
      \draw[thick] (\x-2,\y+2.5) -- (\x+1,\y+2.5); 
      \node at (\x-0.5,\y+2.5){\small \textcolor{red}{x}};
    }
  }
\end{tikzpicture}
    \caption{
        In Cook-Levin transformation from a classical circuit to a classical CSP, one places a binary variable on each wire (red `x') and enforces a (local) consistency constraint on $4$ variables that are input and output to a gate. Fix an assignment to the variables.
        If the assignment satisfies a local consistency constraint, then we can view it as a correct execution of the gate. On the other hand, if the assignment violates a local consistency constraint, then we can view it as an error in the computation. 
    }
    \label{fig:cooklevin}
\end{figure}
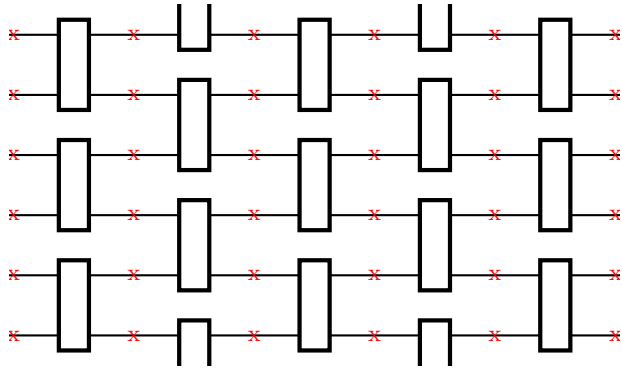

If such $O(\log n)$ depth verification circuit can be transformed into a $\polylog(n)$ depth circuit sound against $\frac{1}{\polylog(n)}$ fraction adversarial errors within the circuits, then applying the Cook-Levin transformation (\Cref{fig:cooklevin}) on this fault-tolerant circuit gives us a CSP $C_2$ in which $\frac{1}{\polylog(n)}$ energy density assignments still encode the accept/reject answer of $C_1$. The crucial observation here is that any violated constraint in $C_2$ can be viewed as an adversarial error on the circuit (\Cref{fig:cooklevin}). Note that we only require adversarial fault tolerance for $\NP$ protocols, which could be different (possibly easier to achieve) from universal adversarially fault-tolerant computation. The fault-tolerant circuit will produce the output encoded in a length $\Omega\left(\frac{n}{\polylog(n)}\right)$ repetition code.\footnote{Given the error budget with adversary, the output cannot be encoded in a code of length smaller than $O\left(\frac{n}{\polylog(n)}\right)$.} We can fault-tolerantly verify the logical output bit by local checks. Let $C_\mathrm{out}$ be the CSP that realizes this check and penalizes the rejecting encoded output in $C_2$.

We claim that the CSP $\frac{1}{2}C_2 +  \frac{1}{2}C_\mathrm{out}$ has a promise gap of $\frac{1}{\polylog(n)}$. In the yes case where $C_1$ is satisfiable, the above fault-tolerant computation accepts for some witness and hence $C_2$ as well as $C_\mathrm{out}$ are satisfiable. In the no case where $C_1$ is unsatisfiable, let $x$ be an assignment to the variables in $C_2$ that have energy at most $\frac{1}{\polylog(n)}$. This assignment encodes the fault-tolerant computation above with $\frac{1}{\polylog(n)}$ fraction of adversarial errors. By assumption on fault-tolerant computation, the logical output bit of this computation should still be an encoding of logical $1$. Hence the penalty from $C_\mathrm{out}$ is $\frac{1}{\polylog(n)}$ due to the distance of the repetition code encoding the logical output bit. Thus, the energy of $\frac{1}{2}C_2 +  \frac{1}{2}C_\mathrm{out}$ is at least $\frac{1}{\polylog(n)}$.

This concludes our claim that adversarial classical fault tolerance with polynomial depth overhead implies a `polylog weaker' version of classical PCP theorem.

\section{QMA-hardness of ground energy estimation up to any inverse-polynomial accuracy}
\label{append:QMApoly}

It can be shown that it is $\QMA$ hard to decide whether the ground energy density is $\leq a$ or $>a+\frac{1}{n^c}$ for any constant $c>0$, by a simple padding trick~\cite{Chen20}. Let $H$ be a $\QMA$-hard local Hamiltonian (family) with energy density promise gap $\frac{1}{n^3}$~\cite{kitaev2002classical}, then consider the Hamiltonian $H' =\frac{1}{m}(H\otimes \id^{\otimes m-1} + \hdots +\id^{\otimes m-1} \otimes H)$ on $N=m n$ qubits. Observe that $H'$ inherits the energy density promise gap $\frac{1}{n^{3}}$ from $H$. By choosing $m=n^{\frac{3}{c} - 1}$ to be a sufficiently large polynomial in $n$, this promise gap is $\frac{1}{N^{c}}$.

We highlight a related aspect of this construction - all inverse-poly energy states of the padded Hamiltonian encode the answer to the quantum computation. For this, given a quantum circuit $C$ (without witness for convenience), lets consider the corresponding Feynman-Kitaev Hamiltonian $H_{FK}$ (without the output term) with spectral gap $\gamma=\frac{\Theta(1)}{n^3}$. The ground state $\ket{\psi}_{FK}$ is the history state and measuring its output qubit $O$ gives the answer to the quantum computation with probability $0.99$ (which can be ensured by padding enough identity gates). Lets take $H' =\frac{1}{m}(H_{FK}\otimes \id^{\otimes m-1} + \hdots +\id^{\otimes m-1} \otimes H_{FK})$ on $N=m n$ qubits, with $m=n^{\frac{3}{c} - 1}$ as before. Any state with energy $\epsilon \leq N^{-c}$ is a superposition over tensor product of ground state (interpreted as no error) and low energy states (interpreted as error). More precisely,

$$\ket{\rho} = \sum_{S\in [m], |S|\leq \epsilon m} \alpha_S \otimes_{i\notin S}\ket{\psi}_{i, FK} \otimes_{i\in S}\ket{\phi_i},$$
where $\ket{\phi_i}$ are orthogonal to $\ket{\psi}_{FK}$. Let $H_{out} = \frac{1}{m}\sum_i \ketbra{1}_{O_i}$ be the Hamiltonian that penalizes the outputs being $1$. We show below that - despite superposition - if C outputs $0$ with probability $0.99$ then $\bra{\rho}H_{out} \ket{\rho} \leq 0.2$. Similarly, we show that if $C$ outputs $1$ with probability $0.99$, then $\bra{\rho}H_{out} \ket{\rho} \geq 0.7$. Thus, the answer to the quantum computation can be read-off from the energy measurement.

For this, suppose C outputs $0$ with probability $0.99$. Then, 
$$\bra{\rho} H_{out} \ket{\rho} = \sum_{S,S'\in [m], |S|, |S'|\leq \epsilon m} \alpha^*_{S'}\alpha_S \left[\otimes_{i\notin S}\bra{\psi}_{i, FK} \otimes_{i\in S}\bra{\phi_i} (H_{out}) \otimes_{i\notin S}\ket{\psi}_{i, FK} \otimes_{i\in S}\ket{\phi_i}\right].$$
If $S=S'$, the term in square brackets is $\leq (1-\epsilon)*0.01 + \epsilon$ (since only $\epsilon$ fraction of locations can be seen as error which can output $1$) and if $S\neq S'$, then the terms matter only when $S$ and $S'$ differ at exactly one index (since terms in $H_{out}$ only look at one index). In that case, suppose $S$ and $S'$ differ at index $j$. Then all terms of $H_{out}$ except $\frac{1}{m}\ketbra{1}_j$ evaluate to $0$. This non-zero term can be evaluated as $\frac{1}{m}\bra{\psi}\ket{1}\bra{1}\ket{\phi_j} \leq \frac{\sqrt{0.01}}{m} \leq \frac{0.1}{m}$. Thus, $\bra{\rho} H_{out} \ket{\rho}$ is at most the max eigenvalue of a matrix whose diagonal entries are at most $0.01+0.99\epsilon$ and off diagonal entries are at most $\frac{0.1}{m}$; and each row has at most $m-1$ off diagonal entries. By Greshgorin disc theorem, the largest eigenvalue is at most $0.1+0.99\epsilon \leq 0.2$ 

Similar argument can be made when $C$ outputs $1$ with probability $0.99$. This completes the proof.

\section{Locality of Clifford propagation terms in the rotated basis}\label{appendix:Clifford}

We show that the rotated propagation term $V^\dagger h_{i,j}^{(\ell)} V$ associated with a Clifford gate $U$ (acting on qubits $i,j$ in layer $\ell <D$) remains local. For simplicity, we first prove this for single-qubit gate. The generalization to the two-qubit case is straightforward. Below we refer to this term as $h_U$ for brevity. 

Recall $h_U=\Lambda^{\otimes 2}_{AB,CD}(\id - \ket{\Phi_U} \bra{\Phi_U}_{BC})\Lambda^{\otimes 2}_{AB,CD}$ acting on qubits $A,B,C,D$ as shown below.

\begin{center}
\begin{tikzpicture}[scale=0.5]
  \foreach \x in {0} {
    \foreach \y in {0} {
      \path [draw=blue,snake it, thick] (\x-0.5,\y+2.5) -- (\x+1.5,\y+2.5); 
      
      \draw[ultra thick] (\x+1,\y+2) rectangle (\x+2,\y+3) node[right] {$U$}; 
      
      \draw [fill=black] (\x-0.5,\y+2.5) circle [radius=0.1] node[below=0.3] {$B$}; 
      
      \draw [fill=black] (\x+1.5,\y+2.5) circle [radius=0.1] node[below=0.3] {$C$};

      \draw [fill=black] (\x-2.5,\y+2.5) circle [radius=0.1] node[below=0.3] {$A$};  
      
      \draw [fill=black] (\x+3.5,\y+2.5) circle [radius=0.1] node[below=0.3] {$D$}; 

    \draw[gray,rounded corners,fill=gray, opacity=0.1] (\x-3.25, \y+1.75) rectangle (\x,\y+3.25); 
    
      \draw[gray,rounded corners,fill=gray, opacity=0.1] (\x+0.75,\y+1.75) rectangle (\x+4,\y+3.25); 
    }
  }
\end{tikzpicture}
\end{center}

As in the main text, let $\ket{\Phi_{\Vec{p}}}_{AB,CD} \triangleq \ket{\Phi_{p_1}}_{AB}\ket{\Phi_{p_2}}_{CD}$ and we will omit the system labels when they are clear from the context.
The rotated term is of the form
\begin{equation*}
\begin{aligned}
    V^\dagger h_U V &= V^\dagger \Lambda^{\otimes 2}_{AB,CD} (\id - \ket{\Phi_U} \bra{\Phi_U}_{BC})  \Lambda^{\otimes 2}_{AB,CD} V \\
    & = V^\dagger \left( \sum_{\Vec{p},\Vec{q} \in \mathcal{P}^{\otimes 2}} \delta^{4-|(\Vec{p}, \Vec{q})|} \ket{\Phi_{\Vec{p}}} \bra{\Phi_{\Vec{q}}}_{AB,CD} \otimes  \bra{\Phi_{\Vec{p}}}(\id - \ket{\Phi_U} \bra{\Phi_U}_{BC}) \ket{\Phi_{\Vec{q}}} \right) V\\
    & = V^\dagger \left( \sum_{\Vec{p},\Vec{q} \in \mathcal{P}^{\otimes 2}} \delta^{4-|(\Vec{p}, \Vec{q})|} \ket{\Phi_{\Vec{p}}} \bra{\Phi_{\Vec{q}}}_{AB,CD} \otimes  \left( \mathds{1}_{\Vec{p},\Vec{q}} - \frac{1}{8} \operatorname{Tr}(p_1^* q_1^\top U^\dagger q_2^\top p_2^* U ) \right)  \right) V \\
    &= \sum_{\Vec{p}} \delta^{4-2|\Vec{p}|} \ket{\Phi_{\Vec{p}}} \bra{\Phi_{\Vec{p}}} \\
    &- \frac{1}{8} \sum_{\Vec{P}\in \mathcal{P}^{\otimes n(\ell-1)}}  \sum_{\Vec{p},\Vec{q}} \delta^{4-|(\Vec{p}, \Vec{q})|} \ket{\Phi_{\Vec{p}}} \bra{\Phi_{\Vec{q}}}  \operatorname{Tr}(p_1^* q_1^\top U^\dagger q_2^\top p_2^* U ) \otimes \ket{\Phi_{\Vec{P}}} \bra{\Phi_{\Vec{P}}} \otimes (\widetilde{W}_{\Vec{P}}^{< \ell })^\dagger p_1^\dagger U^\dagger  p_2^\dagger q_2 U q_1 (\widetilde{W}_{\Vec{P}}^{< \ell }),
\end{aligned}
\end{equation*}
where $\widetilde{W}_{\Vec{P}}^{<\ell}\triangleq W_{\ell-1} \Tilde{P}_{\ell-1} \hdots W_1 \Tilde{P}_1$. Above, $\mathds{1}_{\Vec{p},\Vec{q}}$ denotes the Kronecker delta symbol. The sum $\sum_{\Vec{P}\in \mathcal{P}^{\otimes n(\ell-1)}} $ is over the Pauli noise $\Vec{P}$ in layers preceding the gate $U$. We can also drop the complex conjugate ``$*$'' because $\mathcal{P}=\{I,X,XZ,Z\}$ are real matrices.

Since $U$ is a Clifford operator, the second term above is nonzero if and only if $U^\dagger q_2^\top p_2^* U  = \alpha p_1^* q_1^\top  $ for $\alpha \in \{\pm 1,\pm i\}$. We denote $\Vec{p} 
\overset{U}{\thicksim} \Vec{q}$ if this is the case, leaving the phase $\alpha$ implicit. This notation suggests that the phase $\alpha$ does not show up in $V^\dagger h_U V$. Indeed, it can be verified that
\begin{align*}
    p_1^\top U^\dagger  p_2^\top q_2 U q_1 \operatorname{Tr}(p_1 q_1^\top U^\dagger q_2^\top p_2 U )
    &=  p_1^\top (\alpha^* q_1 p_1^\top) q_1 \operatorname{Tr}(\alpha p_1 q_1^\top p_1 q_1^\top ) \\
    &= \alpha^* \alpha p_1^\top q_1 p_1^\top q_1 \operatorname{Tr}(p_1^\top q_1 p_1^\top q_1) \\
    &= 2 \id.
\end{align*}
Thus,
\begin{align*}
    V^\dagger h_U V &= \sum_{\Vec{p} \in \mathcal{P}^{\otimes 2}} \delta^{4-2|\Vec{p}|} \ket{\Phi_{\Vec{p}}} \bra{\Phi_{\Vec{p}}}  - \frac{1}{4} \sum_{\Vec{p}\overset{U}{\thicksim}\Vec{q}} \delta^{4-|(\Vec{p}, \Vec{q})|} \ket{\Phi_{\Vec{p}}} \bra{\Phi_{\Vec{q}}} \\
    & = \Lambda^{\otimes 2} \left( \sum_{\Vec{p} \in \mathcal{P}^{\otimes 2} } \ket{\Phi_{\Vec{p}}} \bra{\Phi_{\Vec{p}}}  - \frac{1}{4} \sum_{\Vec{p}\overset{U}{\thicksim}\Vec{q}} \ket{\Phi_{\Vec{p}}} \bra{\Phi_{\Vec{q}}} \right) \Lambda^{\otimes 2},
\end{align*}
which is a $4$-qubit operator.

A completely similar analysis for two-qubit gates gives
\begin{align*}
    V^\dagger h_U V & = \Lambda^{\otimes 4} \left( \sum_{\Vec{p} \in \mathcal{P}^{\otimes 4}} \ket{\Phi_{\Vec{p}}} \bra{\Phi_{\Vec{p}}}  - \frac{1}{16} \sum_{\Vec{p}\overset{U}{\thicksim}\Vec{q}} \ket{\Phi_{\Vec{p}}} \bra{\Phi_{\Vec{q}}} \right) \Lambda^{\otimes 4},
\end{align*}
where in this case $\Vec{p}\overset{U}{\thicksim}\Vec{q}$ means $ U^\dagger (q_2^\top p_2) \otimes (q_4^\top p_4) U \propto (q_1^\top p_1) \otimes (q_3^\top p_3)$.

\section{Log-depth $\QMA$ verification using SWAP tests}\label{app:logdepth}
We prove~\Cref{claim:rosgen} in this appendix.

Following the idea in Rosgen's construction~\cite{rosgen2007distinguishing} in proving the $\QIP$-completeness of distinguishing short-depth quantum circuits, we convert a circuit $W$ to a parallel version, denoted $W_1$. In the new parallelized circuit $W_1$, we expect the $\QMA$ prover to provide a witness of the form
\begin{align}
\ket{\xi} = \ket{\psi_0}\ket{\psi_1}^{\otimes 2}\hdots \ket{\psi_{T-1}}^{\otimes 2}\ket{\psi_T},
\label{eq:rosgen-witness}
\end{align}
where $\ket{\psi_t}= U_t\hdots U_1 \ket{\psi_0}$. The verifier can verify consistency between time steps via SWAP tests. It first performs SWAP tests between the two registers $\ket{\psi_t}\ket{\psi_t}$. Then it applies $U_{t+1}$ on a register $\ket{\psi_t}$ and performs SWAP tests between $U_{t+1}\ket{\psi_t}$ and $\ket{\psi_{t+1}}$. Since SWAP test can be done in log depth, $W_1$ is log depth. See \Cref{fig:rosgen}.

Note that a malicious prover could send an entangled state $\ket{\xi}$ across the registers. We will, without too much ambiguity, continue referring to the registers by the states they are expected to hold. The performance of the SWAP test on a general bipartite input state is given by the following fact.

\begin{fact}[Lemma 5.1 in~\cite{rosgen2007distinguishing}] Let $\rho_{AB}$ be a state on two registers $A,B$, then the SWAP test on these two registers accepts with probability at most $\frac{1}{2}(1+ F(\rho_A,\rho_B))$, where $F(\sigma_1,\sigma_2) = \|\sqrt{\sigma_1}\sqrt{\sigma_2}\|_1$ is the fidelity.
\label{fact:swaptest}
\end{fact}

\begin{figure}
    \centering
    \includegraphics[width=0.99\textwidth]{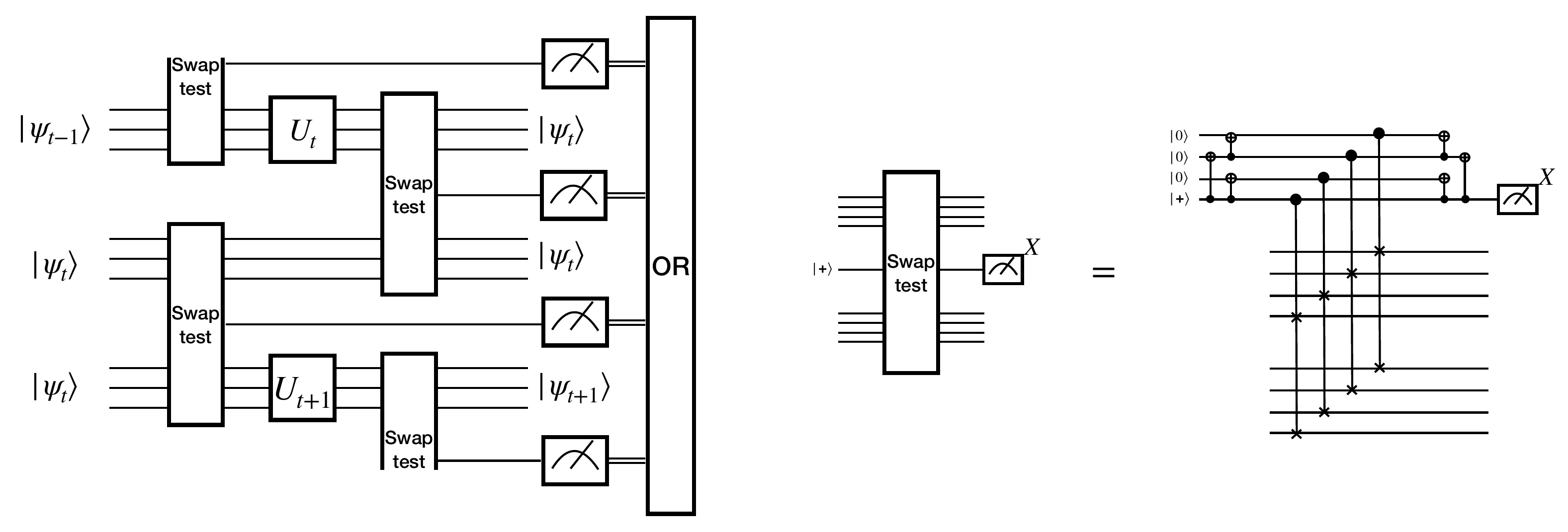}
    \caption{Parallelizing $\QMA$ protocol with SWAP tests. Figure repeated from the main text.}
    \label{fig:rosgen}
\end{figure}

The verifier accepts if and only if the SWAP tests and the decision measurement on the register $\ket{\psi_T}$ all return $1$. We assume the original verifier circuit $W$ has been amplified to completeness $c=1-2^{-\poly(T)}$ and soundness $s = 2^{-\poly(T)}$. 
We have the following claim.

\begin{claim} If the original circuit $W$ accepts with probability at least $c=1-2^{-\poly(T)}$ on some input state, then $W_1$ also accepts with probability at least $c_1=c$ on some input state. On the other hand, if $W$ accepts with probability at most $s$ on any input state, then $W_1$ accepts with probability at most $s_1=1-1/T^{3}$.
\end{claim}

\begin{proof}
    In the completeness case, it is clear that $W_1$ accepts with probability $c_1=c$ upon receiving the desired state from the honest prover. 
    
    In the soundness case, suppose there exists a witness $\ket{\xi}$ that causes all SWAP tests in $W_1$ to accept with probability at least $1-1/T^{\beta}$, where $\beta =3$ (this choice is somewhat arbitrary, any $\beta > 2$ will work), we will show that the final decision measurement only accepts $\ket{\xi}$ with very small probability. Thus, $W_1$ accepts any state with probability at most $s_1 = 1-\frac{1}{T^\beta}$.
    
    We first show that such a witness $\ket{\xi}$ must be close to the expected form in \Cref{eq:rosgen-witness}, and this will imply the decision measurement rejects with probability polynomially close to 1. Indeed, \Cref{fact:swaptest} implies that, for each $t \in [T-1]$, the reduced states of $\ket{\xi}$ on the two `$\ket{\psi_t}$ registers' are $O(\frac{1}{T^{\beta/2}})$-close in the trace distance. Next, let $\xi'$ be resulting state conditioned on all SWAP tests in the first layer (\Cref{fig:rosgen}) succeeding. Performing the same argument on the layer of gates $U_t$ and second layer of SWAP tests, we obtain that, for each $t \in [T-1]$, the consecutive reduced states $\xi'_{t}$ and $U_{t}\xi'_{t-1}$ are $O(\frac{1}{T^{\beta/2}})$-close in the trace distance. On the other hand, we have that $\frac{1}{2}\|\xi' -\xi\|_1 \leq O(\frac{1}{T^{\beta/2}})$ due to gentle measurement lemma. Thus, the same statement holds for the consecutive reduced states of $\xi$. By triangle inequality we thus find
   \begin{align*}
       \frac{1}{2}\|\xi_T - W \xi_0 W^\dagger \| &\leq T\cdot O(\frac{1}{T^{\beta/2}}) + \sum_{t=1}^T \|\xi_t - U_t\xi_{t-1} U_t^\dagger \| \\
       &\leq 2T \cdot O(\frac{1}{T^{\beta/2}}) = O(\frac{1}{T^{\beta/2 - 1}}).
   \end{align*}
    Hence, the decision measurement on register $\xi_t$ accepts with probability at most $s + O(\frac{1}{T^{\beta/2 - 1}})= O(\frac{1}{T^{\beta/2 - 1}})$.
    Thus, $W_1$ accepts with probability at most $s_1= \max\{ 1-\frac{1}{T^\beta},  O(\frac{1}{T^{\beta/2 - 1}})\} = 1-\frac{1}{T^\beta}$ in the soundness case.
\end{proof}
We can also run the standard weak amplification procedure (\Cref{lem:amplify}) to boost the the completeness and soundness of $W_1$ to $1-2^{-\poly(n)}$ and $2^{-\poly(n)}$. The amplified circuit remains log-depth.

\end{document}